\documentclass[pdflatex,sn-nature]{sn-jnl}

\usepackage{setspace}

\usepackage{graphicx}%
\usepackage{multirow}%
\usepackage{amsmath,amssymb,amsfonts}%
\usepackage{amsthm}%
\usepackage{aliascnt}
\usepackage{mathrsfs}%
\usepackage[title]{appendix}%
\usepackage{xcolor}%
\usepackage{textcomp}%
\usepackage{manyfoot}%
\usepackage{booktabs}%
\usepackage{algorithm}%
\usepackage{algorithmicx}%
\usepackage{algpseudocode}%
\usepackage{listings}%
\usepackage{siunitx}
\usepackage[capitalise]{cleveref}
\usepackage{soul}
\usepackage{bm}

\DeclareMathOperator{\dif}{d}

\DeclareMathOperator{\erf}{erf}
\DeclareMathOperator{\erfc}{erfc}

\DeclareMathOperator{\sgn}{sgn}
\DeclareSIUnit\atom{atom}
\DeclareSIUnit\atoms{atoms}
\newcommand{\derivative}[2]{\frac{\dif #1}{\dif #2}}
\newcommand{\infint}[1]{\int_{-\infty}^{\infty} #1}

\theoremstyle{plain}

\newtheorem{theorem}{Theorem}
\newtheorem{lemma}{Lemma}

\theoremstyle{definition}
\newtheorem{definition}{Definition}

\theoremstyle{remark}

\crefname{theorem}{Theorem}{Theorems}
\Crefname{theorem}{Theorem}{Theorems}

\crefname{lemma}{Lemma}{Lemmas}
\Crefname{lemma}{Lemma}{Lemmas}

\crefname{proposition}{Proposition}{Propositions}
\Crefname{proposition}{Proposition}{Propositions}

\crefname{corollary}{Corollary}{Corollaries}
\Crefname{corollary}{Corollary}{Corollaries}
\begin{document}

\title[Article Title]{An extended scattering kernel formalism for multi-scale gas-surface dynamics}


\author*[1]{\fnm{Sabin-Viorel} \sur{Anton}}\email{s.v.anton@tudelft.nl}

\author[1]{\fnm{Bernardo} \sur{Sousa-Alves}}\email{b.s.alves@student.tudelft.nl}

\author[1]{\fnm{Christian} \sur{Siemes}}\email{c.siemes@tudelft.nl}


\author[1]{\fnm{Jose} \sur{van den IJssel}}\email{j.a.a.vandenijssel@tudelft.nl}

\author[1]{\fnm{Pieter N.A.M.} \sur{Visser}}\email{p.n.a.m.visser@tudelft.nl}

\affil*[1]{\orgdiv{Faculty of Aerospace Engineering}, \orgname{Delft University of Technology}, \orgaddress{\street{Kluyverweg 1}, \city{Delft}, \postcode{2629HS}, \country{The Netherlands}}}


\abstract{Gas-particle interactions with non-absorbing surfaces are commonly described using the scattering-kernel formalism. In this framework, an operator $\mathbf{K}$ maps incident velocity distributions to reflected velocity distributions. The operator is self-adjoint and has norm $\lVert \mathbf{K} \rVert = 1$ in an $L^2$ space weighted by the three-dimensional Maxwell-Boltzmann distribution, and must satisfy non-negativity, normalisation, and reciprocity. In standard formulations, $\mathbf{K}$ represents the aggregate effect of all gas-surface interaction mechanisms through a single operator, without distinguishing the physical scales at which these mechanisms occur. For gas scattering from a rough surface, however, it is advantageous to separate geometric effects associated with distinct roughness scales from the underlying thermochemical processes occurring at the atomic scale. We therefore introduce a roughness-based extension of the scattering-kernel formalism, in which a local kernel is successively lifted to larger scales via single- and multi-reflection operators associated with statistically defined surface morphologies. We derive sufficient conditions under which the resulting global kernels preserve reciprocity, normalisation, and non-negativity whenever these properties hold for the smallest-scale kernel. We further show that these constructions define operators on the space of scattering kernels, and establish the associated multi-scale composition laws that allow independent roughness contributions to be combined recursively. The resulting framework provides a general basis for modelling gas-surface scattering on rough surfaces with arbitrary scale decompositions.
 
}


\maketitle

\section{Introduction}\label{sec:Introduction}

Scattering kernels provide a standard mathematical framework for modelling gas–surface interaction (GSI) in orbital aerodynamics \cite{Livadiotti2020}. In their most general formulation, they are defined as non-compact operators that map an incident particle state $(\mathbf{r_i}, \mathbf{v_i}, t_i)$ onto a reflected particle state $(\mathbf{r_r}, \mathbf{v_r}, t_r)$ \cite{Kuer1975, Livadiotti2020, Cercignani1971}, i.e.
\begin{equation} \label{eq:general_scattering_kernel}
    \left|\mathbf{v_r} \cdot \mathbf{n}\right| f_r(\mathbf{r_r}, \mathbf{v_r}, t_r) = \int_{\mathbb{R}^3}\int_{\mathbf{v_i}\cdot\mathbf{n} < 0}\int_{\mathbb{R}} \mathcal{K}\left(\mathbf{r_i} \rightarrow \mathbf{r_r}, \mathbf{v_i} \rightarrow \mathbf{v_r}, t_i \rightarrow t_r \right) \left|\mathbf{v_i} \cdot \mathbf{n}\right| f_i(\mathbf{r_i}, \mathbf{v_i}, t_i) \, \dif t_i \dif \mathbf{v_i} \dif \mathbf{r_i},
\end{equation}
where $\mathbf{r_i}$ and $\mathbf{v_i}$ are the position and velocity, respectively, of a gas particle colliding with the surface at time $t_i$, written in an arbitrary inertial coordinate system, and $\mathbf{r_r}$ and $\mathbf{v_r}$ are the position and velocity, respectively, of the particle exiting the surface at time $t_r$, defined in the same system. Furthermore, $f_i(\mathbf{r_i}, \mathbf{v_i}, t_i)$ and $f_r(\mathbf{r_r}, \mathbf{v_r}, t_r)$ represent the incident and reflected gas particle probability density functions (PDF), respectively. Many scattering kernels are formulated under the assumption that gas–surface collisions are instantaneous and spatially local, i.e. $t_i \approx t_r$ and $\mathbf{r_i} \approx \mathbf {r_r}$ \cite{Livadiotti2020, Cercignani1971, Sentman1961FREEMF}. Under these assumptions, the scattering process reduces to a mapping in velocity space, leading to the simplified definition of
\begin{equation} \label{eq:velocity_scattering_kernel}
    \left|\mathbf{v_r} \cdot \mathbf{n}\right| f_r(\mathbf{v_r}) = \int_{\mathbf{v_i}\cdot\mathbf{n} < 0} \mathcal{K}\left(\mathbf{v_i} \rightarrow \mathbf{v_r}\right) \left|\mathbf{v_i} \cdot \mathbf{n}\right| f_i(\mathbf{v_i}) \, \dif \mathbf{v_i}.
\end{equation}
In contrast, recent advances in physically resolved GSI modelling \cite{Anton2025, Jorge2025, Schtte2025, Endo2025} indicate that gas scattering is intrinsically multi-scale. The interaction is governed not only by atomic-scale thermochemical processes, but also by the geometric morphology of engineering surfaces, whose roughness spans length scales from nanometres to millimetres. In fact, it was shown in \cite{Anton2025} that GSI processes can be separated into two scale-based categories: local interactions, which comprise thermochemical and corrugation-induced processes occurring at the length scale of the surface potential well, and global interactions, which comprise scattering effects induced by geometric roughness. These two classes differ in both characteristic length and characteristic time scales. Consequently, the assumptions of spatial and temporal locality in \cref{eq:velocity_scattering_kernel} are no longer valid, and \cref{eq:general_scattering_kernel} must be used. Nevertheless, this formulation is not well suited for modelling scale-dependent interactions separately, as it does not explicitly isolate the contribution of each individual scale. Moreover, it was shown in \cite{Anton2025} that, in the classical limit, surface roughness acts as a rotational operator applied to its corresponding local kernel, subject to shadowing and masking constraints in both the incident and reflected variables. Consequently, \cref{eq:general_scattering_kernel} may be reformulated as a recursive relation, constrained by shadowing, that links a local kernel and the roughness statistics at a given length scale to the corresponding global kernel. Although such an approach was explored in \cite{Liang2018} for a specific choice of Gaussian roughness statistics, shadowing function and local kernel, a general multi-scale formulation, valid for arbitrary roughness statistics, shadowing functions and admissible local kernels, has not yet been established.

Herein, we propose a new class of scattering kernels that decomposes the GSI process into an arbitrary number of scale-resolved operators. Each operator represents the action of a statistically defined roughness scale on the composite kernel formed by all smaller scales, with an arbitrary local kernel prescribing the gas dynamics at the smallest scale. We then establish sufficient conditions under which this class of kernels satisfies reciprocity, non-negativity, and normalisation. Finally, we introduce a scattering operator based on this formalism and prove that it defines an action of the abelian group composed of all surface morphologies with the summation operation onto the set of all scattering kernels. We therefore show that this operator provides a systematic means of approximating scattering from surface morphologies with arbitrary power spectral densities.

The paper is organised as follows. In \cref{sec:Preliminaries}, we introduce the mathematical definitions used throughout this work, including the notions of a surface and a gas adopted here, as well as the properties required of a physically admissible scattering kernel. In \cref{sec:Roughness_kernel}, we formulate single-reflection and multi-reflection kernels for a rough surface with characteristic scale $R$, and derive sufficient conditions under which these kernels satisfy reciprocity, normalisation, and non-negativity. In \cref{sec:Multiscale_formalism}, we introduce the associated scattering operator $\circ$, which acts between a surface of a given scale and a kernel local to that scale. We then show that this operator induces a homomorphism between the summation of surfaces and the composition of their corresponding scattering kernels. Finally, in \cref{sec:Study_case}, we verify the proposed multi-scale formalism through a study case involving scattering from a two-scale poly-Gaussian surface, using the scattering model in \cite{Anton2025}. Throughout this work, key derived expressions are highlighted using boxed equations.

\section{Preliminaries}\label{sec:Preliminaries}

\subsection{Surface and Gas Definitions} \label{subsec:surface_gas_definitions}

We begin by formalising the notion of a rough surface. Classical gas-surface scattering kernels are generally formulated for smooth, planar surfaces, for which the geometry is fully characterised by a single surface normal vector, $\mathbf{n}$ \cite{Maxwell1879, Sentman1961FREEMF, Cercignani1971, Cercignani2001, Livadiotti2020}. Here, we extend this framework by considering any surface that can be represented by a continuous and differentiable height profile $\xi(x,y,t)$, where $x$ and $y$ denote the in-plane spatial coordinates and $t$ denotes time. To this height profile, we associate a surface object $\Psi$, together with a characteristic length scale $R(t)$ derived from the autocovariance function of $\xi$, $C(\mathbf{r},t)$. In this sense, $R(t)$ plays a role analogous to the Taylor microscale in turbulence modelling, as it extracts a characteristic spatial scale from a two-point correlation function \cite{Taylor1935}. We therefore define $\Psi$ as follows.

\begin{definition} \label{def:surface_general}
    Let $\Gamma : \mathbb{R}^3 \times \mathbb{R} \rightarrow \mathbb{R}$,
    $\Gamma = \Gamma(x,y,z,t)$, be a $\mathcal{C}^2$ continuous level-set function, where $x,y,z$ denote Cartesian spatial coordinates and $t$ denotes time. We further assume that
    \begin{equation}
        \nabla_{\mathbf{r}} \Gamma(x,y,z,t) \neq \mathbf{0},
        \qquad
        \forall \, (x,y,z,t) \in \mathbb{R}^3 \times \mathbb{R}
        \text{ such that } \Gamma(x,y,z,t)=0,
    \end{equation}
    where
    \begin{equation}
        \nabla_{\mathbf{r}} \Gamma
        =
        \begin{bmatrix}
            \partial_x \Gamma &
            \partial_y \Gamma &
            \partial_z \Gamma
        \end{bmatrix}^T.
    \end{equation}
    Then, for each fixed $t \in \mathbb{R}$, we define the surface
    $\Psi(t) \subset \mathbb{R}^3$ as the zero level set of $\Gamma$, i.e.
    \begin{equation}
        \Psi(t)
        =
        \left\{
        (x,y,z) \in \mathbb{R}^3
        \;\middle|\;
        \Gamma(x,y,z,t)=0
        \right\}.
    \end{equation}
    Equivalently, the space-time surface associated with $\Gamma$ may be written as
    \begin{equation}
        \Psi
        =
        \left\{
        (x,y,z,t) \in \mathbb{R}^4
        \;\middle|\;
        \Gamma(x,y,z,t)=0
        \right\}.
    \end{equation}
    We further define the local surface normal field
    $\mathbf{n_L} : \mathbb{R}^3 \times \mathbb{R} \rightarrow \mathbb{S}^2$,
    $\mathbf{n_L} = \mathbf{n_L}(x,y,z,t)$, on $\Psi$ by
    \begin{equation}
        \mathbf{n_L}(x,y,z,t)
        =
        \frac{\nabla_{\mathbf{r}} \Gamma(x,y,z,t)}
        {\left\lVert \nabla_{\mathbf{r}} \Gamma(x,y,z,t) \right\rVert},
        \qquad
        \forall \, (x,y,z,t) \in \Psi.
    \end{equation}
    If $\Psi$ is regarded as a random surface, or as a member of an ensemble of
    surfaces, we define the associated probability density of local normals as
    \begin{equation}
        p_n : \mathbb{S}^2 \rightarrow [0,\infty),
        \qquad
        p_n = p_n(\mathbf{n_L}),
    \end{equation}
    where $p_n(\mathbf{n_L}) \, \dif \mathbf{n_L}$ gives the probability of
    observing a local normal in an infinitesimal neighbourhood of
    $\mathbf{n_L} \in \mathbb{S}^2$. We assume this distribution
    to be stationary and homogeneous.
    
    Finally, we define $\mathcal{P}$ as the set of all admissible surfaces, i.e.,
    \begin{equation}
        \begin{aligned}
            \mathcal{P}
            =
            \biggl\{
            \Psi
            \;\biggm|\;
            &\exists \, \Gamma \in \mathcal{C}^2(\mathbb{R}^3 \times \mathbb{R})
            \text{ satisfying the conditions above,} \\
            & \text{such that }
            \Psi
            =
            \left\{
            (x,y,z,t) \in \mathbb{R}^4
            \;\middle|\;
            \Gamma(x,y,z,t)=0
            \right\} \biggr\}.
        \end{aligned}
    \end{equation}
\end{definition}

\begin{definition} \label{def:surface_height}
    Let $\Psi$ be a surface in the sense of \cref{def:surface_general}, with associated level-set function
    $\Gamma : \mathbb{R}^3 \times \mathbb{R} \rightarrow \mathbb{R}$,
    $\Gamma = \Gamma(x,y,z,t)$, assumed to be of class $\mathcal{C}^2$. Let the local surface normal
    $\mathbf{n_L} : \mathbb{R}^3 \times \mathbb{R} \rightarrow \mathbb{S}^2$ be defined as in
    \cref{def:surface_general}, and let $p_n : \mathbb{S}^2 \rightarrow [0,\infty)$,
    $p_n = p_n(\mathbf{n_L})$, denote the corresponding PDF of the local normals.

    We say that $\Psi$ admits a \textbf{height representation} if the function
    $z \mapsto \Gamma(x,y,z,t)$ satisfies
    \begin{equation}
        \partial_z \Gamma(x,y,z,t) \neq 0,
        \qquad
        \forall \, (x,y,z,t) \in \mathbb{R}^3 \times \mathbb{R},
    \end{equation}
    and
    \begin{equation}
        \lim_{z\to\infty}\Gamma(x,y,z,t) > 0 \text{ and }
        \lim_{z\to-\infty}\Gamma(x,y,z,t) < 0,
        \qquad
        \forall \, (x,y,t) \in \mathbb{R}^2 \times \mathbb{R}.
    \end{equation}
    Under these conditions, there exists a unique function
    $\xi : \mathbb{R}^2 \times \mathbb{R} \rightarrow \mathbb{R}$,
    $\xi = \xi(x,y,t)$, of class $\mathcal{C}^2$, such that
    \begin{equation}
        \Gamma(x,y,t) = 0,
        \qquad
        \forall \, (x,y,t) \in \mathbb{R}^2 \times \mathbb{R},
    \end{equation}
    and therefore
    \begin{equation}
        \Psi
        =
        \left\{
        (x,y,z,t) \in \mathbb{R}^3 \times \mathbb{R}
        \;\middle|\;
        z = \xi(x,y,t)
        \right\}.
    \end{equation}
    Then, the local normal may be written equivalently as
    \begin{equation}
        \mathbf{n_L}(x,y,t)
        =
        \frac{1}{\sqrt{1 + (\partial_x \xi)^2 + (\partial_y \xi)^2}}
        \begin{bmatrix}
            -\partial_x \xi &
            -\partial_y \xi &
            1
        \end{bmatrix}^T,
        \qquad
        \forall \, (x,y,t) \in \mathbb{R}^2 \times \mathbb{R}.
    \end{equation}
    Finally, we define $\mathcal{P}^h \subset \mathcal{P}$ as the set of all admissible surfaces admitting such a height representation, i.e.,
    \begin{equation}
        \begin{aligned}
            \mathcal{P}^h
            =
            \biggl\{
            \Psi
            \;\biggm|\;
            &\exists \, \Gamma \in \mathcal{C}^2(\mathbb{R}^3 \times \mathbb{R})
            \text{ satisfying the conditions above, and } \\
            &\exists \, \xi \in \mathcal{C}^2(\mathbb{R}^2 \times \mathbb{R})
            \text{ such that }
            \Psi =
            \left\{
            (x,y,z,t) \in \mathbb{R}^3 \times \mathbb{R}
            \;\middle|\;
            z = \xi(x,y,t)
            \right\}
            \biggr\}.
        \end{aligned}
    \end{equation}
\end{definition}

\begin{definition} \label{def:length_scale}
    Let $\Psi \in \mathcal{P}$ be a surface with a height representation in the sense of \cref{def:surface_height}, with associated height profile $\xi : \mathbb{R}^2 \times \mathbb{R} \rightarrow \mathbb{R}$, $\xi = \xi(x,y,t)$. Assume that for each fixed $t \in \mathbb{R}$, $\xi(\cdot,\cdot,t)$ is a real-valued, stationary, second-order random field on $\mathbb{R}^2$ with finite variance. Then, its autocovariance function $C : \mathbb{R}^2 \times \mathbb{R} \rightarrow \mathbb{R}$ is defined by
    \begin{equation}
        C(\mathbf{r},t)
        =
        \mathbb{E}\!\left[
        \left(\xi(\mathbf{x},t)-\mathbb{E}[\xi(\mathbf{x},t)]\right)
        \left(\xi(\mathbf{x}+\mathbf{r},t)-\mathbb{E}[\xi(\mathbf{x}+\mathbf{r},t)]\right)
        \right],
    \end{equation}
    where $\mathbf{x} = \begin{bmatrix}x & y\end{bmatrix}^T$ and $\mathbf{r} = \begin{bmatrix}r_x & r_y\end{bmatrix}^T \in \mathbb{R}^2$. By stationarity, $C$ depends only on the lag $\mathbf{r}$ and time $t$. The power spectral density $\Phi : \mathbb{R}^2 \times \mathbb{R} \rightarrow [0,\infty)$ associated with $\xi$ is defined as
    \begin{equation}
        \Phi(\pmb{\kappa},t)
        =
        \frac{1}{(2\pi)^2}
        \int_{\mathbb{R}^2}
        C(\mathbf{r},t)
        \exp\!\left[- i \pmb{\kappa}\cdot\mathbf{r}\right]
        \,\dif \mathbf{r},
    \end{equation}
    where $\pmb{\kappa} = \begin{bmatrix}\kappa_x & \kappa_y\end{bmatrix}^T \in \mathbb{R}^2$.
    Assuming that $\Phi$ admits finite zeroth and second spectral moments, we define \textbf{the characteristic length scale of $\Psi$}, $R : \mathbb{R} \rightarrow [0,\infty)$, $R = R(t)$, as
    \begin{equation}
        R(t)
        =
        \sqrt{
        \frac{
        2\int_{\mathbb{R}^2}\Phi(\pmb{\kappa},t)\,\dif\pmb{\kappa}
        }{
        \int_{\mathbb{R}^2}\|\pmb{\kappa}\|^2\,\Phi(\pmb{\kappa},t)\,\dif\pmb{\kappa}
        }
        }.
    \end{equation}
\end{definition}
An immediate consequence of the definition of $\Psi$ is that, unlike for a planar surface, not every point on the surface is necessarily exposed to an incident particle stream with velocity $\mathbf{v_i}$, except in the special case in which the stream is aligned with the global normal $\mathbf{n_G}$, i.e., when $\mathbf{v_i} \times \mathbf{n_G} = 0$, and when no other parts of the surface occlude the stream. Similarly, for particles reflected from a point $\mathbf{r_r}$ on the surface, not every reflection direction leads to escape from the surface. Here, escape is defined by the condition $\lim_{t\to\infty} \left\lVert \mathbf{r}(t)-\mathbf{r_r} \right\rVert = \infty$, where $\mathbf{r}(t)$ denotes the trajectory of the reflected particle. We therefore introduce the one-point masking and shadowing functions, $\mathcal{M}(\mathbf{r_r}, \mathbf{v_r}, t_r)$ and $\mathcal{S}(\mathbf{r_r}, \mathbf{v_r}, t_r)$, respectively, which describe the probabilities that a particle either reaches the surface or escapes from it under the prescribed kinematic conditions. In addition, we define the two-point shadowing function $\mathcal{S}(\mathbf{r_i}, \mathbf{v_i}, t_i , | , \mathbf{r_r}, \mathbf{v_r}, t_r)$, which describes the probability that a particle departing from $\mathbf{r_r}$ at time $t_r$ with velocity $\mathbf{v_r}$ reaches the point $\mathbf{r_i}$ at time $t_i$ with velocity $\mathbf{v_i}$. Finally, we introduce the visible particle flux $\mathcal{Q}\left(\mathbf{v_i} , | , \mathbf{n_G} \right)$, encoding the average fraction of incident particles with velocity $\mathbf{v_i}$ that impinge on surface elements whose local normals are visible from the incoming direction defined with respect to the global normal $\mathbf{n_G}$.

\begin{definition} \label{def:shadowing_single_scale}
Let $\Psi$ be a surface in the sense of \cref{def:surface_general}, and let $R : \mathbb{R} \to [0,\infty)$, $R = R(t)$, be its characteristic length scale in the sense of \cref{def:length_scale}. We define the normal-averaged two-point shadowing function
\begin{equation}
    \mathcal{S} :
    \left(\mathbb{R}^3 \times \mathbb{R}^3 \times \mathbb{R}\right)
    \times
    \left(\mathbb{R}^3 \times \mathbb{R}^3 \times \mathbb{R}\right)
    \to [0,1],
\end{equation}
written as $\mathcal{S}(\mathbf{r_i},\mathbf{v_i},t_i \mid \mathbf{r_r},\mathbf{v_r},t_r)$, such that $\mathcal{S}(\mathbf{r_i},\mathbf{v_i},t_i \mid \mathbf{r_r},\mathbf{v_r},t_r)$ is the probability that a particle departing from the surface at position $\mathbf{r_r}$, time $t_r$, and velocity $\mathbf{v_r}$ does not collide again with $\Psi$ before reaching position $\mathbf{r_i}$ at time $t_i$ with velocity $\mathbf{v_i}$. 

For fixed $(\mathbf{r_r},\mathbf{v_r},t_r)$, let $\Delta r = \lVert \mathbf{r_i} - \mathbf{r_r} \rVert$ and $\Delta t = t_i - t_r$. Whenever the limit exists, we define the one-point shadowing function $\mathcal{S} : \mathbb{R}^3 \times \mathbb{R}^3 \times \mathbb{R} \to [0,1]$
by 
\begin{equation}
    \mathcal{S}(\mathbf{r_r},\mathbf{v_r},t_r) = \lim_{\substack{\Delta r \to \infty \\ \Delta t \to \infty}} \mathcal{S}(\mathbf{r_i},\mathbf{v_i},t_i \mid \mathbf{r_r},\mathbf{v_r},t_r).
\end{equation}
 Thus, $\mathcal{S} (\mathbf{r_r},\mathbf{v_r},t_r)$ is the probability that a particle leaving the surface from $(\mathbf{r_r},\mathbf{v_r},t_r)$ escapes without any subsequent collision. Similarly, for fixed $(\mathbf{r_i},\mathbf{v_i},t_i)$, whenever the limit exists, we define the one-point masking function $\mathcal{M} : \mathbb{R}^3 \times \mathbb{R}^3 \times \mathbb{R} \to [0,1],$ by
\begin{equation}
    \mathcal{M}(\mathbf{r_i},\mathbf{v_i},t_i)
    =
    \lim_{\substack{\Delta r \to \infty \\ \Delta t \to \infty}}
    \mathcal{S}(\mathbf{r_r},-\mathbf{v_r},t_r \mid \mathbf{r_i},-\mathbf{v_i},t_i).
\end{equation}
Thus, $\mathcal{M}(\mathbf{r_i},\mathbf{v_i},t_i)$ is the probability that a particle arriving from far away with velocity $\mathbf{v_i}$ reaches the surface and collides with it at position $\mathbf{r_i}$ and time $t_i$. If the surface evolution is time-reversal symmetric, then the shadowing and masking functions satisfy
\begin{equation}
    \mathcal{S}(\mathbf{r},-\mathbf{v},-t)
    =
    \mathcal{M}(\mathbf{r},\mathbf{v},t),
    \qquad
    \forall\,(\mathbf{r},\mathbf{v},t)
    \in
    \mathbb{R}^3 \times \mathbb{R}^3 \times \mathbb{R}.
\end{equation}
Taking the derivative of the two-point shadowing function with respect to time, we obtain \textbf{the probability of visible states} $(\mathbf{r_i}, \mathbf{v_i}, t_i)$, given the initial state $(\mathbf{r_r}, \mathbf{v_r}, t_r)$, as
\begin{equation}
    \Dot{\mathcal{S}}(\mathbf{r_i},\mathbf{v_i},t_i \mid \mathbf{r_r},\mathbf{v_r},t_r) = \frac{\dif}{\dif t} \mathcal{S}(\mathbf{r_i},\mathbf{v_i},t_i \mid \mathbf{r_r},\mathbf{v_r},t_r) = \mathcal{S}(\mathbf{r_i},\mathbf{v_i},t_i \mid \mathbf{r_r},\mathbf{v_r},t_r) f_{\varphi}\left(\mathbf{r_i}, \mathbf{v_i}, t_i \right),
\end{equation}
where $f_{\varphi} : \mathbb{R}^3 \times \mathbb{R}^3 \times \mathbb{R} \rightarrow [0, \infty)$, $f_{\varphi} = f_{\varphi}(\mathbf{r_i}, \mathbf{v_i}, t_i)$ is \textbf{the hazard probability} of state $(\mathbf{r_i}, \mathbf{v_i}, t_i)$, and is given as
\begin{equation}
    f_{\varphi}(\mathbf{r_i}, \mathbf{v_i}, t_i) = \frac{\dif}{\dif t}\ln\left[\mathcal{S}(\mathbf{r_i},\mathbf{v_i},t_i \mid \mathbf{r_r},\mathbf{v_r},t_r)\right].
\end{equation}
If we assume the previous state $(\mathbf{r_r}, \mathbf{v_r}, t_r)$ to reside at an infinite distance away from the surface, we may define \textbf{the probability of incident visible states} as
\begin{equation}
    \Dot{\mathcal{M}}(\mathbf{r_i},\mathbf{v_i},t_i) = \lim_{\substack{\Delta r \to \infty \\ \Delta t \to \infty}}\left[\frac{\dif}{\dif t} \mathcal{S}(\mathbf{r_i},\mathbf{v_i},t_i \mid \mathbf{r_r},\mathbf{v_r},t_r)\right] = \mathcal{M}(\mathbf{r_i},\mathbf{v_i},t_i) f_{\varphi}\left(\mathbf{r_i}, \mathbf{v_i}, t_i \right).
\end{equation}
\end{definition}

\begin{definition} \label{def:visible_flux}
    Let $\Psi \in \mathcal{P}$ be a surface in the sense of \cref{def:surface_general}, and let $\mathbf{v_i} \in \mathbb{R}^3$ be the incident velocity of a particle colliding with $\Psi$. Let $\mathbf{n_L} \in \mathbb{S}^2$ and $p_n : \mathbb{S}^2 \rightarrow [0, \infty)$, $p_n = p_n(\mathbf{n_L} \, | \, \mathbf{n_G})$ be the local normal field and its corresponding PDF associated with $\Psi$, where $\mathbf{n_G} \in \mathbb{S}^2$ is the global normal of $\Psi$. Finally, let $G$ and $L$ be the global and local frames associated with $\Psi$, in the sense of \cref{def:global_frame,def:local_frame}, such that $\mathbf{n_L}$ is defined in $L$. Then, we define \textbf{the particle visible flux} as
    \begin{equation}
        \mathcal{Q}(\mathbf{v_i} \, | \, \mathbf{n_G}) = \int_{\mathbf{v_i} \cdot \mathbf{n_L} < 0} \left|\mathbf{v_i} \cdot \mathbf{n_L} \right| \, p_n\left(\mathbf{n_L} \, | \, \mathbf{n_G} \right) \, \dif \mathbf{n_L}.
    \end{equation}
    If $\Psi$ is smooth in $G$ in the sense of \cref{def:local_smoothness}, then $\mathcal{Q}$ simplifies to
    \begin{equation}
        \mathcal{Q}(\mathbf{v_i} \, | \, \mathbf{n_G}) = \int_{\mathbf{v_i} \cdot \mathbf{n_L} < 0} \left|\mathbf{v_i} \cdot \mathbf{n_L} \right| \, \delta \left(\mathbf{n_L} - \mathbf{n_G} \right) \, \dif \mathbf{n_L} = \left|\mathbf{v_i} \cdot \mathbf{n_G} \right|.
    \end{equation}
\end{definition}
The reader should note that the notation adopted here for shadowing and masking is reversed relative to the convention commonly used in optics and wave-scattering literature \cite{Smith1967, Beckmann1965, Brown1980}. In those fields, shadowing refers to the obstruction of the incident flux by the surface, whereas masking refers to the obstruction of the reflected flux with respect to the observer. Here, however, we adopt the opposite convention, defining both quantities from the perspective of the reflection point rather than that of the observer.

We conclude this subsection by defining the assumed equilibrium velocity PDF of the free-stream particles, namely the Maxwell-Boltzmann distribution \cite{Maxwell1879}, given below.

\begin{definition} \label{def:MaxwellBoltzmann}
    Let $\mathbf{v}$ be a vector of dimension 3 $\times$ 1 representing the velocity of a gas particle in the global frame as given by \cref{def:global_frame}. Let $T$, $\mu$, and $\mathcal{R}$ be real numbers representing the gas particle temperature and molar mass, and the ideal gas constant, respectively. Then, $\mathbf{v}$ follows the \textbf{Maxwell-Boltzmann} distribution \cite{Maxwell1879}, $f_0 : \mathbb{R}^3 \rightarrow [0, \infty)$, defined as
    \begin{equation}
        f_0(\mathbf{v} \, | \, T, \mu) = \left[\frac{\mu}{2\pi \mathcal{R} T}\right]^{\frac{3}{2}} \exp\left[- \frac{\mu v^2}{2 \mathcal{R} T} \right].
    \end{equation}
\end{definition}

\subsection{Coordinate system definitions} \label{subsec:coordinate_systems_definitions}

From the specification of a rough surface $\Psi$ in \cref{def:surface_general}, two types of coordinate systems naturally arise in the development of a multi-scale scattering framework: global frames and local frames. A global frame is aligned with the global normal of $\Psi$ (given by the average of all local normals), whereas a local frame is aligned with a specific local normal $\mathbf{n_L}$ at a given point $\mathbf{r_i} = \begin{bmatrix}
    x_i & y_i & z_i
\end{bmatrix}^T$ and time $t_i$. They are defined as follows. 

\begin{definition} \label{def:global_frame}
    Let $\Psi$ be a surface in the sense of \cref{def:surface_general}, with associated local surface normal field
    $\mathbf{n_L} : \mathbb{R}^3 \times \mathbb{R} \to \mathbb{S}^2$,
    $\mathbf{n_L} = \mathbf{n_L}(\mathbf{r},t)$, defined on $\Psi$, where
    $\mathbf{r} = (x,y,z)^T$. For each $L > 0$, let
    \begin{equation}
        B_L
        =
        \left\{
        \mathbf{r} \in \mathbb{R}^3
        \;\middle|\;
        \lVert \mathbf{r} \rVert \leq L
        \right\},
        \qquad
        \Psi_L(t)
        =
        \Psi(t) \cap B_L,
    \end{equation}
    and let $\dif A$ denote the surface-area measure on $\Psi(t)$. Let $\mathbf{v_i}(t)$ be a $3 \times 1$ real vector defining the instantaneous incident velocity of particles at time $t$, and define the time-averaged incident velocity by
    \begin{equation}
        \mathbf{\Bar{v}_i}
        =
        \lim_{T\to\infty}
        \frac{1}{T}
        \int_0^T \mathbf{v_i}(t) \, \dif t,
    \end{equation}
    provided the limit exists and $\mathbf{\Bar{v}_i} \neq \mathbf{0}$. Assuming that the surface-area and time average of $\mathbf{n_L}$ exists and is non-zero, we define
    \begin{equation}
        \mathbf{n_G}
        =
        \lim_{L,T\to\infty}
        \frac{
        \displaystyle
        \int_0^T \int_{\Psi_L(t)}
        \mathbf{n_L}(\mathbf{r},t)
        \, \dif A \, \dif t
        }{
        \displaystyle
        \int_0^T \int_{\Psi_L(t)}
        1
        \, \dif A \, \dif t
        }, \text{ and }  \mathbf{z_G}
        = \mathbf{n_G}.
    \end{equation}
    Next, define
    \begin{equation}
        \mathbf{e}
        =
        \begin{cases}
            \begin{bmatrix}
                1 & 0 & 0
            \end{bmatrix}^T,
            & \text{if }
            \mathbf{z_G}
            \times
            \begin{bmatrix}
                1 & 0 & 0
            \end{bmatrix}^T
            \neq \mathbf{0},
            \\[10pt]
            \begin{bmatrix}
                0 & 1 & 0
            \end{bmatrix}^T,
            & \text{otherwise.}
        \end{cases}
    \end{equation}
    Then we define
    \begin{equation}
        \mathbf{x_G}
        =
        \begin{cases}
            \dfrac{
            \mathbf{\Bar{v}_i}
            -
            \left(
            \mathbf{\Bar{v}_i} \cdot \mathbf{z_G}
            \right)\mathbf{z_G}
            }{
            \left\lVert
            \mathbf{\Bar{v}_i}
            -
            \left(
            \mathbf{\Bar{v}_i} \cdot \mathbf{z_G}
            \right)\mathbf{z_G}
            \right\rVert
            },
            & \text{if }
            \mathbf{\Bar{v}_i} \times \mathbf{z_G} \neq \mathbf{0},
            \\[14pt]
            \dfrac{
            \mathbf{e}
            -
            \left(
            \mathbf{e} \cdot \mathbf{z_G}
            \right)\mathbf{z_G}
            }{
            \left\lVert
            \mathbf{e}
            -
            \left(
            \mathbf{e} \cdot \mathbf{z_G}
            \right)\mathbf{z_G}
            \right\rVert
            },
            & \text{otherwise.}
        \end{cases}
    \end{equation}
    Finally, we define
    \begin{equation}
        \mathbf{y_G}
        =
        \mathbf{z_G} \times \mathbf{x_G}.
    \end{equation}
    The ordered set $\left\{\mathbf{x_G}, \mathbf{y_G}, \mathbf{z_G}\right\}$ forms a right-handed orthonormal basis of $\mathbb{R}^3$, referred to as the global coordinate system associated with $\Psi$ and $\mathbf{\Bar{v}_i}$.
\end{definition}

\begin{definition} \label{def:local_frame}
    Let $\{\mathbf{x_G}, \mathbf{y_G}, \mathbf{z_G}\}$ denote the global reference frame as in \cref{def:global_frame}, associated with the surface $\Psi$, as given by \cref{def:surface_general}. Let $\mathbf{n_L} : \mathbb{R} \rightarrow \mathbb{R}^3, \, \mathbf{n_L} = \mathbf{n_L}(t)$ be a real vector of dimension $3 \times 1$ describing the local surface normal of $\Psi$ at time $t$, as in \cref{def:global_frame}. We assume the interactions in this scale to occur over time spans much shorter than in the global frame of \cref{def:global_frame}. Hence, $\mathbf{n_L}(t) \approx \mathbf{n_L}(t_c), \, \forall\, t \in \mathbb{R}$, where $t_c$ is the time where the local interaction begins. Let $\theta_{n_1} \in [0, \frac{\pi}{2})$ and $\theta_{n_2} \in [0, 2\pi)$ be two angles describing the direction of $\mathbf{n_L}$ as
    \begin{equation}
        \mathbf{n_L} = \begin{bmatrix}
            \sin(\theta_{n_1})\cos(\theta_{n_2}) & \sin(\theta_{n_1})\sin(\theta_{n_2}) & \cos(\theta_{n_1})
        \end{bmatrix}^T.
    \end{equation}
    Then, we define $\mathbf{z_L} = \mathbf{n_L}$, and $\mathbf{x_L}$ as
    \begin{equation}
        \mathbf{x_L} = \begin{cases}
            \frac{\mathbf{z_G} - \left(\mathbf{n_L} \cdot \mathbf{z_G} \right) \mathbf{n_L}}{\lVert\mathbf{z_G} - \left(\mathbf{n_L} \cdot \mathbf{z_G} \right) \mathbf{n_L} \rVert} & \text{if } \mathbf{n_L} \times \mathbf{z_G} \neq 0, \\
            \begin{bmatrix}
                1 & 0 & 0
            \end{bmatrix}^T & \text{otherwise}.
        \end{cases}
    \end{equation}
    Finally, we define $\mathbf{y_L} = \mathbf{x_L} \times \mathbf{z_L}$. The ordered set $\left\{\mathbf{x_L}, \mathbf{y_L}, \mathbf{z_L}\right\}$, shown in \cref{fig:coordinate_systems} in bronze, forms a right-handed orthonormal basis of $\mathbb{R}^3$, referred to as the local coordinate system associated with $\Psi$ and $\mathbf{n_L}$. For any vector $\left.\mathbf{r}\right|^G$ written in the global frame, a transformation to the local frame is achieved as
    \begin{equation} \label{pp2:eq:global_local_vector_trans}
        \left.\mathbf{r}\right|^L = T_{y}(\theta_{n_1})T_{z}(\theta_{n_2}) \left.\mathbf{r}\right|^G, \quad \text{where} \quad T_z(\theta) = \begin{bmatrix}
            \cos(\theta) & \sin(\theta) & 0 \\
            -\sin(\theta) & \cos(\theta) & 0 \\
            0 & 0 & 1
        \end{bmatrix}, \quad T_y(\theta) = \begin{bmatrix}
            \cos(\theta) & 0 & -\sin(\theta) \\
            0 & 1 & 0 \\
            \sin(\theta) & 0 & \cos(\theta)
        \end{bmatrix}.
    \end{equation}
\end{definition}
\begin{figure}[H]
    \centering
    \includegraphics[width=0.8\linewidth]{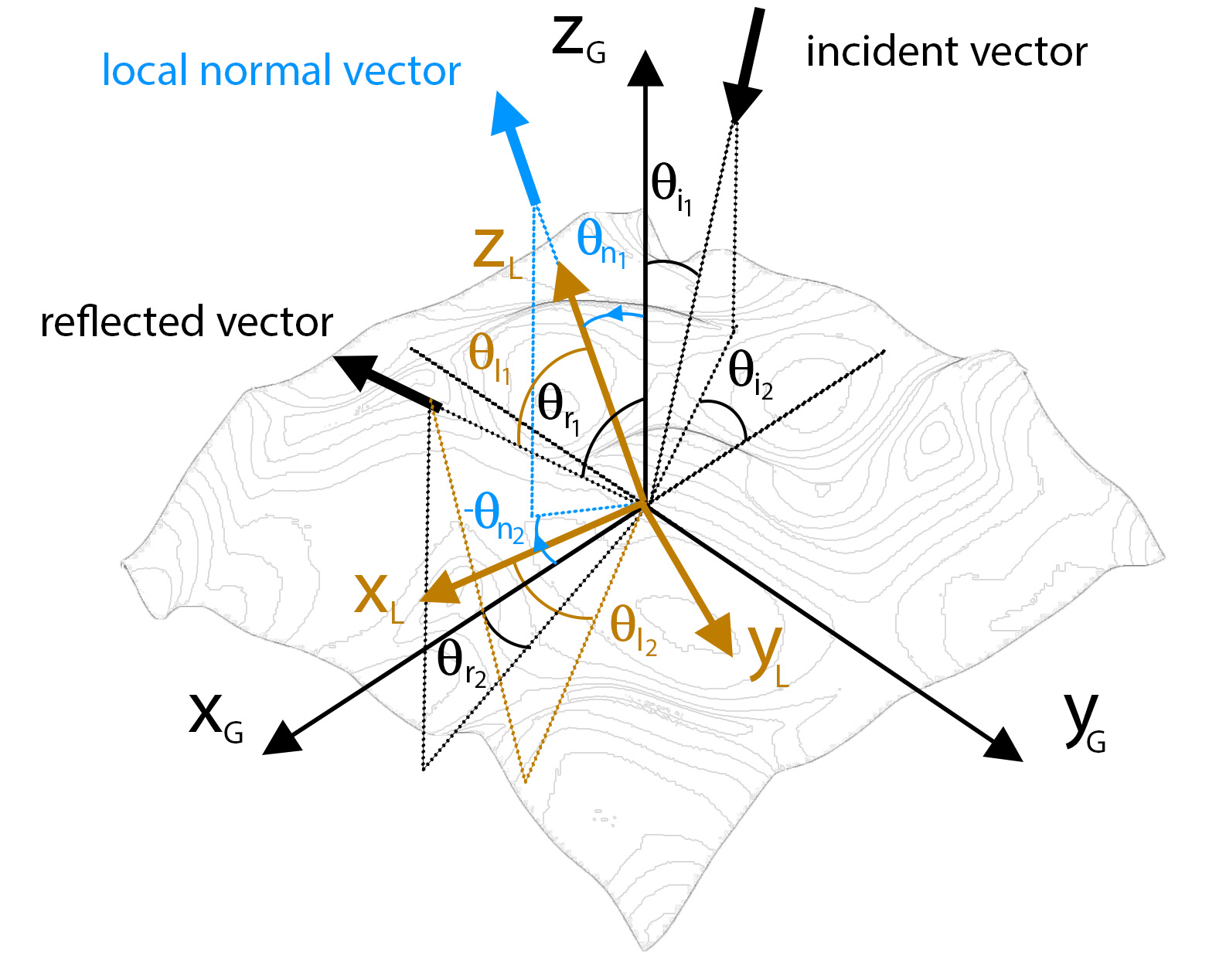}
    \caption{Schematic of the local and global coordinate systems with the associated angular definitions indicated.}
    \label{fig:coordinate_systems}
\end{figure}

\subsection{Multi-scale surface definitions} \label{subsec:multiscale_def}

We now generalise the definition of a rough surface to account for roughness over an arbitrary number of length scales. We first define a multi-scale surface $\Psi^{(m)}$ as an object composed of $m$ single-scale surfaces in the sense of \cref{def:surface_general}. We then introduce the special case of multi-scale surfaces admitting a height representation, for which the resulting boundary can be described by a height profile $\xi(x,y,t)$ in the global frame. Finally, we extend the shadowing, masking and visible flux functions of \cref{def:shadowing_single_scale,def:visible_flux} to the multi-scale setting.

\begin{definition} \label{def:surface_multiscale}
    Let $\Psi_1, \Psi_2, \dots, \Psi_m$ be $m$ surfaces in the sense of \cref{def:surface_general}, where $m \in \mathbb{Z}^+$. We further define the reference frames $F_0, F_1, \dots, F_m$ such that $F_0$ is a global frame in the sense of \cref{def:global_frame}, and, for each $k \in \overline{1,m}$, $F_k$ is the local frame in the sense of \cref{def:local_frame} associated with $\Psi_k$. By convention, any quantity $a^{(k)}$ is defined in frame $F_k$ and expressed in the global frame $F_0$. For each $k \in \overline{1,m}$, let $\mathbf{n}^{(k)} : \mathbb{R}^3 \times \mathbb{R} \rightarrow \mathbb{S}^2$ denote the normal field associated with $\Psi_k$, expressed in frame $F_{k-1}$. Moreover, let $\mathbf{n}^{(0)} : \mathbb{R}^3 \times \mathbb{R} \rightarrow \mathbb{S}^2$ be the deterministic global normal field, i.e.
    \begin{equation}
        p_{n_0}\Bigl(\mathbf{n}^{(0)}\Bigr)
        =
        \delta\Bigl(\mathbf{n}^{(0)} - \mathbf{n_G}\Bigr),
    \end{equation}
    for a given global normal $\mathbf{n_G}$.

    For each $k \in \overline{1,m}$, let
    \begin{equation}
        p_{n_k} : \mathbb{S}^2 \times \mathbb{S}^2 \rightarrow [0,\infty),
        \qquad
        p_{n_k}
        =
        p_{n_k}\Bigl(\mathbf{n}^{(k)} \,\Big|\, \mathbf{n}^{(k-1)}\Bigr),
    \end{equation}
    denote the conditional PDF of the normal field associated with $\Psi_k$ and defined in frame $F_k$, conditioned on the normal at the immediately coarser scale and expressed in $F_{0}$.

    We then define the \textbf{multi-scale surface} $\Psi^{(m)}$ associated with $\Psi_1,\dots,\Psi_m$ as any surface in the sense of \cref{def:surface_general} whose local normal distribution $p_n^{(m)} : \mathbb{S}^2 \rightarrow [0,\infty)$ is given by
    \begin{equation}
        \begin{aligned}
            p_n^{(m)}(\mathbf{n})
            &=
            \int_{\mathbb{S}^2} \cdots \int_{\mathbb{S}^2}
            p_{n_m}\Bigl(\mathbf{n}^{(m)} \,\Big|\, \mathbf{n}^{(m-1)}\Bigr)
            \prod_{k=1}^{m-1}
            p_{n_k}\Bigl(\mathbf{n}^{(k)} \,\Big|\, \mathbf{n}^{(k-1)}\Bigr)
            \,
            \dif \mathbf{n}^{(m-1)} \cdots \dif \mathbf{n}^{(1)},
        \end{aligned}
    \end{equation}
    for all $\mathbf{n} \in \mathbb{S}^2$, where $\mathbf{n}^{(0)} = \mathbf{n_G}$. 
\end{definition}

\begin{definition} \label{def:local_smoothness}
    Let $\Psi^{(m)} \in \mathcal{P}$ be a multi-scale surface in the sense of \cref{def:surface_multiscale}, where $m \in \mathbb{Z}^+$ is the number of scales, and let $\Psi_1, \Psi_2, \dots \Psi_m \in \mathcal{P}$ be its associated single-scale components in the sense of \cref{def:surface_general}, with $p_n^{(1)}, p_n^{(2)}, \dots, p_n^{(m)}$ being their corresponding normal field PDFs, defined in frames $F_1, F_2, \dots F_m$. 

    We say $\Psi^{(m)}$ is \textbf{smooth} in frame $F_m$ if 
    \begin{equation}
        p_n^{(m)}\Bigl(\mathbf{n}^{(m)} \, \Big| \, \mathbf{n}^{(m-1)}\Bigr) = \delta\Bigl(\mathbf{n}^{(m)} - \mathbf{n}^{(m-1)}\Bigr), \, \forall \, \Bigl(\mathbf{n}^{(m)}, \mathbf{n}^{(m-1)} \Bigr) \in \mathbb{S}^2 \times \mathbb{S}^2.
    \end{equation}
\end{definition}

\begin{definition} \label{def:surface_multiscale_height}
    Let $\Psi^{(m)}$ be a multi-scale surface in the sense of \cref{def:surface_multiscale}, associated with the component surfaces $\Psi_1,\Psi_2,\dots,\Psi_m$ and the conditional normal PDFs
    \begin{equation}
        p_{n_{k}}\Bigl(\mathbf{n}^{(k)} \,\Big|\, \mathbf{n}^{(k-1)}\Bigr),
        \qquad
        k \in \overline{1,m},
        \qquad
        \mathbf{n}^{(0)} = \mathbf{n_G}.
    \end{equation}
    We say that $\Psi^{(m)}$ admits a \textbf{multi-scale height representation} if, for every $k \in \overline{1,m}$, the surface $\Psi_k$ admits a height representation in the same global frame $F_0$ in the sense of \cref{def:surface_height}. That is, for each $k \in \overline{1,m}$, there exists a unique function
    \begin{equation}
        \xi_k^{(0)} : \mathbb{R}^2 \times \mathbb{R} \rightarrow \mathbb{R},
        \qquad
        \xi_k^{(0)} = \xi_k^{(0)}(x,y,t),
        \qquad
        \xi_k^{(0)} \in \mathcal{C}^2(\mathbb{R}^2 \times \mathbb{R}),
    \end{equation}
    such that
    \begin{equation}
        \Psi_k
        =
        \left\{
        (x,y,z,t) \in \mathbb{R}^3 \times \mathbb{R}
        \;\middle|\;
        z = \xi_k^{(0)}(x,y,t)
        \right\}.
    \end{equation}
    Then, defining the PDF of the slope vector $\Dot{\mathbf{\xi}}_k^{(0)} : \mathbb{R}^2 \times \mathbb{R} \rightarrow \mathbb{R}^2, \, \Dot{\mathbf{\xi}}_k^{(0)}(x, y, t) = \begin{bmatrix}
        \partial_x \xi_k^{(0)}(x, y, t) & \partial_y \xi_k^{(0)}(x, y, t)
    \end{bmatrix}^T$ as $p_{s_{k}} : \mathbb{R}^2 \rightarrow [0, \infty), \, p_{s_{k}} = p_{s_{k}}\Bigl(\Dot{\mathbf{\xi}}_k^{(0)}\Bigr)$, we establish the following relation with the local normal distribution $p_n(\mathbf{n})$ associated with $\Psi$,
    \begin{equation}
    \begin{aligned}
        p_{n}\left(\mathbf{n}\right) \, \dif \mathbf{n}
        &=
        \int_{\mathbb{S}^2}\cdots\int_{\mathbb{S}^2}
        \prod_{k=1}^{m}
        p_{n_k}\Bigl(
        \mathbf{n}^{(k)}\left(\Dot{\mathbf{\xi}}_k^{(0)}+\dots+\Dot{\mathbf{\xi}}_1^{(0)}\right)
        \,\Big|\,
        \mathbf{n}^{(k-1)}\left(\Dot{\mathbf{\xi}}_{k-1}^{(0)}+\dots+\Dot{\mathbf{\xi}}_1^{(0)}\right)
        \Bigr)
        \,
        \dif \mathbf{n}^{(m-1)} \cdots \dif \mathbf{n}^{(1)} \, \dif \mathbf{n} \\
        &=
        \int_{\mathbb{R}^2}\cdots\int_{\mathbb{R}^2}
        \prod_{k=1}^{m}
        p_{s_{k1}}\Bigl(
        \Dot{\mathbf{\xi}}_k^{(0)}+\dots+\Dot{\mathbf{\xi}}_1^{(0)}
        \,\Big|\,
        \Dot{\mathbf{\xi}}_{k-1}^{(0)}+\dots+\Dot{\mathbf{\xi}}_1^{(0)}
        \Bigr) \, \dif\!\Bigl(\Dot{\mathbf{\xi}}_m^{(0)}+\dots+\Dot{\mathbf{\xi}}_1^{(0)}\Bigr)
        \dif\!\Bigl(\Dot{\mathbf{\xi}}_{m-1}^{(0)}+\dots+\Dot{\mathbf{\xi}}_1^{(0)}\Bigr)
        \cdots
        \dif \Dot{\mathbf{\xi}}_1^{(0)} \\
        &=
        p_{s_m}\Bigl(\Dot{\mathbf{\xi}}_m^{(0)}+\dots+\Dot{\mathbf{\xi}}_1^{(0)}\Bigr)
        \,
        \dif\!\Bigl(\Dot{\mathbf{\xi}}_m^{(0)}+\dots+\Dot{\mathbf{\xi}}_1^{(0)}\Bigr) \\
        &=
        p_{s_m}\Bigl(\Dot{\mathbf{\xi}}^{(0)}\Bigr)\,\dif \Dot{\mathbf{\xi}}^{(0)},
    \end{aligned}
    \end{equation}
    where the slope field associated with $\Psi^{(m)}$, $\Dot{\mathbf{\xi}}^{(0)} = \sum_{k=1}^m \Dot{\mathbf{\xi}}_k^{(0)}$ is simply the sum of the slopes of the single-scale surfaces $\Psi_k$. Thus, we may define the addition operation between surfaces in the sense of \cref{def:surface_height} by
    \begin{equation}
        + : \mathcal{P}^h \times \mathcal{P}^h \rightarrow \mathcal{P}^h,
        \qquad
        (\Psi_a,\Psi_b) \mapsto \Psi_a + \Psi_b,
    \end{equation}
    where
    \begin{equation}
        \Psi_a + \Psi_b
        :=
        \left\{
        (x,y,z,t) \in \mathbb{R}^3 \times \mathbb{R}
        \;\middle|\;
        z = \xi_a^{(0)}(x,y,t) + \xi_b^{(0)}(x,y,t)
        \right\},
    \end{equation}
    for any $\Psi_a,\Psi_b \in \mathcal{P}^h$ with height representations $\xi_a^{(0)}$ and $\xi_b^{(0)}$ in the same global frame $F_0$. In particular,
    \begin{equation}
        \Psi^{(m)} = \sum_{k=1}^m \Psi_k.
    \end{equation}
    Since the sum of two $\mathcal{C}^2$ height functions is again of class $\mathcal{C}^2$, the set $\mathcal{P}^h$ is closed under $+$. Moreover, $+$ is associative and commutative, the identity element is the flat surface
    \begin{equation}
        0_{\mathcal{P}}
        :=
        \left\{
        (x,y,z,t) \in \mathbb{R}^3 \times \mathbb{R}
        \;\middle|\;
        z = 0
        \right\},
    \end{equation}
    and the inverse of any $\Psi \in \mathcal{P}^h$ with height representation $\xi^{(0)}$ is
    \begin{equation}
        -\Psi
        :=
        \left\{
        (x,y,z,t) \in \mathbb{R}^3 \times \mathbb{R}
        \;\middle|\;
        z = -\xi^{(0)}(x,y,t)
        \right\}.
    \end{equation}
    Hence, $(\mathcal{P}^h,+)$ is an abelian group.
\end{definition}

\begin{definition} \label{def:shadowing_multi_scale}
    Let $\Psi^{(m)}$ be a multi-scale surface, in the sense of \cref{def:surface_multiscale}, where $m \in \mathbb{Z}+$ represents the number of scales. Furthermore, let $\Psi_1, \Psi_2, \dots \Psi_m \in \mathcal{P}$ be its associated single-scale surfaces in the sense of \cref{def:surface_general}, and let $\mathbf{n}^{(1)}, \mathbf{n}^{(2)}, \dots \mathbf{n}^{(m)} : \mathbb{R}^2 \times \mathbb{R} \rightarrow \mathbb{S}^2$ be their corresponding normal fields. Finally, let $\mathbf{n}^{(0)} = \mathbf{n_G} \in \mathbb{S}^2$ be the global normal field associated with $\Psi^{(m)}$. Finally, we define  $\mathcal{S}_1, \mathcal{S}_2, \dots \mathcal{S}_m$ as the two-point shadowing functions associated with these surfaces in the sense of \cref{def:shadowing_single_scale}. Then, we define the two-point shadowing function of $\Psi^{(m)}$ between a previous reflected state $\pmb{\varphi}^{(0)}_{r_{k-1}}$ at reflection $k-1$ and a new incident state $\pmb{\varphi}^{(0)}_{i_k}$ at reflection $k$ as
    \begin{equation}
    \begin{aligned}
        &\mathcal{S}^{(m)}\Bigl(
        \pmb{\varphi}^{(0)}_{i_k},
        \mathbf{n}^{(m-1)}_k,\dots,\mathbf{n}^{(1)}_k,\mathbf{n}^{(0)}_k
        \mid
        \pmb{\varphi}^{(0)}_{r_{k-1}},
        \mathbf{n}^{(m-1)}_{k-1},\dots,\mathbf{n}^{(1)}_{k-1},\mathbf{n}^{(0)}_{k-1}
        \Bigr) =
        \\
        &\qquad =
        \sum_{j=1}^{m}
        \Biggl\{
        \Biggl[
        \prod_{\ell=j+1}^{m}
        \mathcal{S}_{\ell}\Bigl(
        \pmb{\varphi}^{(\ell-1)}_{i_k}
        \,\Big|\,
        \pmb{\varphi}^{(\ell-1)}_{r_{k-1}}, \mathbf{n}^{(\ell-1)}_{k-1}
        \Bigr)
        \Biggr]
        \mathcal{S}_{j}\Bigl(
        \pmb{\varphi}^{(j-1)}_{i_k}
        \,\Big|\,
        \pmb{\varphi}^{(j-1)}_{r_{k-1}}
        \Bigr)
        \\
        &\hspace{28mm}\cdot
        \Biggl[
        \prod_{\ell=j+1}^{m}
        \mathcal{S}_{\ell}\Bigl(
        \pmb{\varphi}^{(\ell-1)}_{i_k}
        \,\Big|\,
        \pmb{\varphi}^{(\ell-1)}_{r_{k-1}}, \mathbf{n}^{(\ell-1)}_{k}
        \Bigr)
        \,
        p_{n_k}^{(\ell-1)}\Bigl(
        \mathbf{n}^{(\ell-1)}_{k}
        \,\Big|\,
        \mathbf{n}^{(\ell-2)}_{k}
        \Bigr)
        \Biggr]
        \Biggl[
        \prod_{\ell=1}^{j-1}
        \delta\Bigl(
        \mathbf{n}^{(\ell)}_{k}
        -
        \mathbf{n}^{(\ell)}_{k-1}
        \Bigr)
        \Biggr]
        \Biggr\},
    \end{aligned}
\end{equation}
where $p_{n_k}^{(\ell)}$ represents the normal field PDF at reflection $k$ and scale $l$. If we further assume the length scales of $\Psi_1, \Psi_2, \dots \Psi_m$ in the sense of \cref{def:length_scale}, $R_1, R_2, \dots R_m$ to be very different, i.e. $R_1 \gg R_2 \gg \dots \gg R_m$, then the expression for $\mathcal{S}^{(m)}$ reduces to
    \begin{equation} \label{eq:shadowing_multiscale}
    \begin{aligned}
        &\mathcal{S}^{(m)}\Bigl(
        \pmb{\varphi}^{(0)}_{i_k},
        \mathbf{n}^{(m-1)}_k,\dots,\mathbf{n}^{(1)}_k,\mathbf{n}^{(0)}_k
        \mid
        \pmb{\varphi}^{(0)}_{r_{k-1}},
        \mathbf{n}^{(m-1)}_{k-1},\dots,\mathbf{n}^{(1)}_{k-1},\mathbf{n}^{(0)}_{k-1}
        \Bigr) =
        \\
        &\qquad =
        \sum_{j=1}^{m}
        \Biggl\{
        \Biggl[
        \prod_{\ell=j+1}^{m}
        \mathcal{S}_{\ell}\Bigl(
        \pmb{\varphi}^{(\ell-1)}_{r_{k-1}} \, | \, \mathbf{n}^{(\ell-1)}_{k-1}
        \Bigr)
        \Biggr]
        \mathcal{S}_{j}\Bigl(
        \pmb{\varphi}^{(j-1)}_{i_k}
        \,\Big|\,
        \pmb{\varphi}^{(j-1)}_{r_{k-1}}
        \Bigr)
        \\
        &\hspace{28mm}\cdot
        \Biggl[
        \prod_{\ell=j+1}^{m}
        \mathcal{M}_{\ell}\Bigl(
        \pmb{\varphi}^{(\ell-1)}_{i_k} \, | \, \mathbf{n}^{(\ell-1)}_{k}
        \Bigr)
        \,
        p_{n_k}^{(\ell-1)}\Bigl(
        \mathbf{n}^{(\ell-1)}_{k}
        \,\Big|\,
        \mathbf{n}^{(\ell-2)}_{k}
        \Bigr)
        \Biggr]
        \Biggl[
        \prod_{\ell=1}^{j-1}
        \delta\Bigl(
        \mathbf{n}^{(\ell)}_{k}
        -
        \mathbf{n}^{(\ell)}_{k-1}
        \Bigr)
        \Biggr]
        \Biggr\}.
    \end{aligned}
\end{equation}
Taking the derivative with respect to time of either of the two expressions above, we may define the hazard probability of the multi-scale surface as
\begin{equation}
\begin{aligned}
    &f_{\varphi}^{(m)}\Bigl(
        \pmb{\varphi}^{(0)}_{i_k},
        \mathbf{n}^{(m-1)}_k,\dots,\mathbf{n}^{(1)}_k,\mathbf{n}^{(0)}_k
        \mid
        \pmb{\varphi}^{(0)}_{r_{k-1}},
        \mathbf{n}^{(m-1)}_{k-1},\dots,\mathbf{n}^{(1)}_{k-1},\mathbf{n}^{(0)}_{k-1}
        \Bigr) = \\
        & \qquad = \frac{\dif }{\dif t} \ln\left[\mathcal{S}^{(m)}\Bigl(
        \pmb{\varphi}^{(0)}_{i_k},
        \mathbf{n}^{(m-1)}_k,\dots,\mathbf{n}^{(1)}_k,\mathbf{n}^{(0)}_k
        \mid
        \pmb{\varphi}^{(0)}_{r_{k-1}},
        \mathbf{n}^{(m-1)}_{k-1},\dots,\mathbf{n}^{(1)}_{k-1},\mathbf{n}^{(0)}_{k-1}
        \Bigr) \right].
\end{aligned}
\end{equation}
Moreover, assuming the initial state $\pmb{\varphi}^{(0)}_{r_{k-1}}$ to be at an infinite distance away from the surface $\Psi^{(m)}$, we derive the multi-scale masking function
\begin{equation}
    \begin{aligned}
        \mathcal{M}^{(m)}\Bigl(
        \pmb{\varphi}^{(0)}_{i_k} \, | \, 
        \mathbf{n}^{(m-1)}_k,\dots,\mathbf{n}^{(1)}_k,\mathbf{n}^{(0)}_k\Bigr) = \mathcal{M}_{1}\Bigl(
        \pmb{\varphi}^{(0)}_{i_k}
        \Bigr)\,
        \Biggl[
        \prod_{\ell=2}^{m}
        \mathcal{M}_{\ell}\Bigl(
        \pmb{\varphi}^{(\ell-1)}_{i_k} \, | \, \mathbf{n}^{(\ell-1)}_{k}
        \Bigr)
        \Biggr].
    \end{aligned}
\end{equation}
Similarly, we may define the multi-scale one-point shadowing function by assuming the new incident state $\pmb{\varphi}^{(0)}_{i_{k}}$ to be at an infinite distance away from the surface, to obtain
\begin{equation}
    \begin{aligned}
        \mathcal{S}^{(m)}\Bigl(
        \pmb{\varphi}^{(0)}_{r_{k-1}} \, | \, 
        \mathbf{n}^{(m-1)}_{k-1},\dots,\mathbf{n}^{(1)}_{k-1},\mathbf{n}^{(0)}_{k-1}\Bigr) = \mathcal{S}_{1}\Bigl(
        \pmb{\varphi}^{(0)}_{r_{k-1}}
        \Bigr)\,
        \Biggl[
        \prod_{\ell=2}^{m}
        \mathcal{S}_{\ell}\Bigl(
        \pmb{\varphi}^{(\ell-1)}_{r_{k-1}} \, | \, \mathbf{n}^{(\ell-1)}_{k-1}
        \Bigr)
        \Biggr].
    \end{aligned}
\end{equation}
Finally, we further define the normal-averaged masking function as
\begin{equation}
    \left\langle\mathcal{M}^{(m)}\Bigl(\pmb{\varphi}^{(0)}_{i_k}\Bigr)\right\rangle = \int_{\mathbb{S}^2}\dots\int_{\mathbb{S}^2}\mathcal{M}^{(m)}\Bigl(
        \pmb{\varphi}^{(0)}_{i_k} \, | \, 
        \mathbf{n}^{(m-1)}_k,\dots,\mathbf{n}^{(1)}_k,\mathbf{n}^{(0)}_k\Bigr) \, \prod_{\ell = 2}^m p_{n_k}^{(\ell-1)}\left( \mathbf{n}_k^{(\ell-1)} \, | \, \mathbf{n}_k^{(\ell-2)}\right) \, \dif \mathbf{n}_k^{(m-1)} \dif \mathbf{n}_k^{(m-2)} \dots \dif \mathbf{n}_k^{(1)}.
\end{equation}
\end{definition}

\begin{definition} \label{def:visible_flux_multiscale}
    Let $\Psi^{(m)} \in \mathcal{P}$ be a multi-scale surface in the sense of \cref{def:surface_multiscale}, where $m \in \mathbb{Z}^+$ is the number of scales, and let $\mathbf{v_i} \in \mathbb{R}^3$ be the incident velocity of a particle colliding with $\Psi^{(m)}$. Let $\Psi_1, \Psi_2, \dots \Psi_m \in \mathcal{P}$ be its associated single scale surfaces in the sense of \cref{def:surface_general}, and let $\mathbf{n}^{(1)}, \mathbf{n}^{(2)}, \dots \mathbf{n}^{(m)} : \mathbb{R}^2 \times \mathbb{R} \rightarrow \mathbb{S}^2$ be their corresponding normal fields. Furthermore, let $p_{n}^{(1)}, p_{n}^{(2), \dots, p_{n}^{(m)}}$ be the local normal PDFs, and let $\mathbf{n}^{(0)} = \mathbf{n_G} \in \mathbb{S}^2$ be the global normal field associated with $\Psi^{(m)}$. Finally, let $F_1 \dots F_m$ be the frames associated with $\Psi_1, \Psi_2, \dots \Psi_m$, in the sense of \cref{def:global_frame}, and let $F_0$ be the global frame of $\Psi^{(m)}$. Then, we define \textbf{the multi-scale particle visible flux} of $\Psi^{(m)}$ recursively as
    \begin{equation}
        \mathcal{Q}^{(m)}(\mathbf{v_i}^{(0)} \, | \, \mathbf{n}^{(0)}) = \int_{\mathbb{S}^2} \mathcal{Q}^{(m-1)}\left(\mathbf{v_i}^{(1)} \, | \, \mathbf{n}^{(1)} \right) \, p_n\left(\mathbf{n}^{(1)} \, | \, \mathbf{n}^{(0)} \right) \, \dif \mathbf{n}^{(1)},
    \end{equation}
    with $\mathcal{Q}^{(0)}(\mathbf{v_i}^{(m)} \, | \, \mathbf{n}^{(m)})$ taking the form
    \begin{equation}
        \mathcal{Q}^{(0)}(\mathbf{v_i}^{(m)} \, | \, \mathbf{n}^{(m)}) = \begin{cases}
            \left|\mathbf{v_i}^{(m)} \cdot \mathbf{n}^{(m)} \right|, & \mathbf{v_i}^{(m)} \cdot \mathbf{n}^{(m)} < 0, \\
            0 & \mathbf{v_i}^{(m)} \cdot \mathbf{n}^{(m)} \geq 0.
        \end{cases}
    \end{equation}
\end{definition}

\subsection{Scattering kernel properties} \label{subsec:scattering_kernel_properties}

We now turn to the notion of scattering kernels, as introduced in \cite{Cercignani1971}, and reformulate the admissibility conditions for rough surface morphologies. In addition to reciprocity, normalisation, and non-negativity, we also impose impermeability (which is intrinsically satisfied by kernels defined on smooth surfaces). Since the fraction of the surface visible to an incident particle flux depends on the incident velocity $\mathbf{v_i}$, this dependence must enter the reciprocity relation as a weighting factor to correctly represent the detailed-balance condition of the Boltzmann equation. A point $\mathbf{r}$ on a surface $\Psi$ may be inaccessible to the incident flux either because it is masked by the surrounding morphology, or because its local normal $\mathbf{n_L}$ is not exposed to the flux, i.e., $\mathbf{v_i}\cdot\mathbf{n_L} > 0$. We therefore distinguish between two reciprocity conditions: a global condition for particles that complete an interaction with the rough surface by entering and leaving it, and a pointwise condition for particles undergoing a single collision at an unmasked point. These definitions are given below, together with the corresponding conditions for impermeability, normalisation, and non-negativity.

\begin{definition} \label{def:impermeability}
    Let  $\mathbf{r_i}$, $\mathbf{r_r}$, $\mathbf{v_i}$ and $\mathbf{v_r}$ be real vectors of dimension 3 $\times$ 1,  representing the incident and reflected position and velocity of a particle, defined in an arbitrary inertial coordinate system, and let $t_i, t_r \in \mathbb{R}$ be the respective entry and exit times. Let $\mathbf{n_G}$ represent a 3 $\times$ 1 global normal unit vector of a given surface $\Psi$ as given by \cref{def:surface_general}, defined in the same frame. Furthermore, let $\mathcal{K}\left(\mathbf{r_i} \rightarrow \mathbf{r_r}, \mathbf{v_i} \rightarrow \mathbf{v_r}, t_i \rightarrow t_r\right)$. Then, $\mathcal{K}$ satisfies \textbf{impermeability} if 
    \begin{equation}
    \begin{aligned}
        \mathcal{K}\left(\mathbf{r_i} \rightarrow \mathbf{r_r}, \mathbf{v_i} \rightarrow \mathbf{v_r}, t_i \rightarrow t_r\right) = 0 \quad & \text{when} \quad \sgn\left(\mathbf{v_i} \cdot \mathbf{n_G} \right) \neq \sgn\left(-\mathbf{v_r} \cdot \mathbf{n_G} \right) \\
        & \forall \, (\mathbf{n_G}, \mathbf{r_i}, \mathbf{r_r}, \mathbf{v_i}, \mathbf{v_r}, t_i, t_r) \in \mathbb{R}^3 \times \mathbb{R}^3 \times \mathbb{R}^3 \times \mathbb{R}^3 \times \mathbb{R}^3 \times \mathbb{R} \times \mathbb{R}.
    \end{aligned}
    \end{equation}
    
\end{definition}

\begin{definition} \label{def:reciprocity}
Let $\mathcal{K} : \mathbb{R}^3 \times \mathbb{R}^3 \times \mathbb{R}
\times \mathbb{R}^3 \times \mathbb{R}^3 \times \mathbb{R}
\rightarrow [0, \infty)$, written as $\mathcal{K} = \mathcal{K}(\mathbf{r_i} \rightarrow \mathbf{r_r}, \mathbf{v_i} \rightarrow \mathbf{v_r}, t_i \rightarrow t_r)$, be a Lebesgue-integrable function defined in the global frame as in \cref{def:global_frame} and associated with a surface $\Psi$ in the sense of \cref{def:surface_general}. Here, the vectors $\mathbf{r_i}$, $\mathbf{r_r}$, $\mathbf{v_i}$ and $\mathbf{v_r}$ of dimension $3 \times 1$ represent the incident and reflected particle positions and velocities, while the real values $t_i$ and $t_r$ represent the corresponding surface entry and exit times. Let $\mathbf{n_G}$ denote the global surface normal unit vector of dimension $3 \times 1$, defined in the same frame. Then, the function $\mathcal{K}$ satisfies \textbf{reciprocity} if
\begin{multline} \label{eq:reciprocity_ultra_general}
    \mathcal{K}\left(\mathbf{r_i} \rightarrow \mathbf{r_r}, \mathbf{v_i} \rightarrow \mathbf{v_r}, t_i \rightarrow t_r \right) \mathcal{Q}(\mathbf{v_i} \, | \,  \mathbf{n_G}) \, f_0(\mathbf{v_i}) \, \left\langle\Dot{\mathcal{M}}(\mathbf{r_i}, \mathbf{v_i}, t_i)\right\rangle \\
    = \mathcal{K}\left(\mathbf{r_r} \rightarrow \mathbf{r_i}, -\mathbf{v_r} \rightarrow -\mathbf{v_i}, -t_r \rightarrow -t_i \right) \mathcal{Q}( -\mathbf{v_r} \, | \,  \mathbf{n_G}) \, f_0(-\mathbf{v_r}) \, \left\langle\Dot{\mathcal{M}}(\mathbf{r_r}, -\mathbf{v_r}, -t_r)\right\rangle,
\end{multline}
where $\Dot{\mathcal{M}} : \mathbb{R}^3 \times \mathbb{R}^3 \times \mathbb{R} \rightarrow [0, \infty)$, $\Dot{\mathcal{M}} = \Dot{\mathcal{M}}(\mathbf{r_i}, \mathbf{v_i}, t_i)$, is the PDF of incident visible states $(\mathbf{r_i}, \mathbf{v_i}, t_i)$ as given by \cref{def:shadowing_single_scale}, $\left\langle \Dot{\mathcal{M}} \right\rangle$ implies averaging over all local surface normals in the sense of \cref{def:shadowing_multi_scale}, and $f_0(\mathbf{v}) = f_0(\mathbf{v} \, | \, T_S, \mu)$ is the Maxwell-Boltzmann distribution as in \cref{def:MaxwellBoltzmann}, with $T_S, \mu \in \mathbb{R}$ representing the surface temperature and gas particle molar mass, respectively. $\mathcal{Q}$, on the other hand, is the incident visible flux, as given by \cref{def:visible_flux}. Expanding $\left\langle \Dot{\mathcal{M}} \right \rangle$ as in \cref{def:shadowing_single_scale}, the reciprocity condition becomes
\begin{multline} \label{eq:reciprocity_general}
    \mathcal{K}\left(\mathbf{r_i} \rightarrow \mathbf{r_r}, \mathbf{v_i} \rightarrow \mathbf{v_r}, t_i \rightarrow t_r \right) \mathcal{Q}( \mathbf{v_i}\, | \,  \mathbf{n_G}) \left\langle f_{\varphi_0}(\mathbf{r_i}, \mathbf{v_i}, t_i) \, \mathcal{M}(\mathbf{r_i}, \mathbf{v_i}, t_i)\right\rangle \\
    = \mathcal{K}\left(\mathbf{r_r} \rightarrow \mathbf{r_i}, -\mathbf{v_r} \rightarrow -\mathbf{v_i}, -t_r \rightarrow -t_i \right)  \mathcal{Q}(-\mathbf{v_r} \, | \,  \mathbf{n_G}) \left\langle f_{\varphi_0}(\mathbf{r_r}, -\mathbf{v_r}, -t_r) \, \mathcal{M}(\mathbf{r_r}, -\mathbf{v_r}, -t_r)\right\rangle,
\end{multline}
where the equilibrium state PDF $f_{\varphi_0} : \mathbb{R}^3 \times \mathbb{R}^3 \times \mathbb{R} \rightarrow [0, \infty)$, $f_{\varphi_0} = f_{\varphi}(\mathbf{r}, \mathbf{v}, t) \, f_0(\mathbf{v})$, with the hazard probability $f_{\varphi}$ given as in \cref{def:shadowing_single_scale}. We may further generalise the above expression for a multi-scale surface $\Psi^{(m)}$ in the sense of \cref{def:surface_multiscale}, as
\begin{multline} \label{eq:reciprocity_general_multi-scale}
    \mathcal{K}\left(\mathbf{r_i}^{(0)} \rightarrow \mathbf{r_r}^{(0)}, \mathbf{v_i}^{(0)} \rightarrow \mathbf{v_r}^{(0)}, t_i \rightarrow t_r \right) \, \mathcal{Q}^{(m)}\Bigl( \mathbf{v_i}^{(0)}\, | \,  \mathbf{n}^{(0)}\Bigr) \, \left\langle f_{\varphi_0}^{(m)}\Bigl(\mathbf{r_i}^{(0)}, \mathbf{v_i}^{(0)}, t_i\Bigr) \, \mathcal{M}^{(m)}(\mathbf{r_i}, \mathbf{v_i}, t_i) \right\rangle \\
    = \mathcal{K}\left(\mathbf{r_r}^{(0)} \rightarrow \mathbf{r_i}^{(0)}, -\mathbf{v_r}^{(0)} \rightarrow -\mathbf{v_i}^{(0)}, -t_r \rightarrow -t_i \right) \, \mathcal{Q}^{(m)}\Bigl( -\mathbf{v_r}^{(0)}\, | \,  \mathbf{n}^{(0)}\Bigr) \, \left\langle f_{\varphi_0}^{(m)}\Bigl(\mathbf{r_r}^{(0)}, -\mathbf{v_r}^{(0)}, -t_r\Bigr) \, \mathcal{M}^{(m)}(\mathbf{r_r}, -\mathbf{v_r}, -t_r)\right\rangle,
\end{multline}
where $m \in \mathbb{Z}^+$ represents the number of scales of $\Psi^{(m)}$. It is trivial to show that \cref{eq:reciprocity_general,eq:reciprocity_general_multi-scale} reduce, for a smooth surface $\Psi$ with $m=1$ and instantaneous collisions ($t_i \approx t_r, \mathbf{r_i} \approx \mathbf{r_r})$, to the expression given in earlier studies \cite{Cercignani1971, Livadiotti2020}, namely
\begin{equation} \label{eq:reciprocity_velocity}
     \mathcal{K}\left(\mathbf{v_i} \rightarrow \mathbf{v_r}\right) \left|\mathbf{v_i} \cdot \mathbf{n_G} \right| f_0(\mathbf{v_i}) = \mathcal{K}\left( -\mathbf{v_r} \rightarrow -\mathbf{v_i}\right)  \left|-\mathbf{v_r} \cdot \mathbf{n_G} \right|  f_0(-\mathbf{v_r}).
\end{equation}
\end{definition}

\begin{definition} \label{def:pointwise_reciprocity}
    Let $\mathcal{K} : \mathbb{R}^3 \times \mathbb{R}^3 \times \mathbb{R}
    \times \mathbb{R}^3 \times \mathbb{R}^3 \times \mathbb{R}
    \rightarrow [0,\infty)$, written as $\mathcal{K} = \mathcal{K}(\mathbf{r_i} \rightarrow \mathbf{r_r}, \mathbf{v_i} \rightarrow \mathbf{v_r}, t_i \rightarrow t_r)$, be a Lebesgue-integrable function defined in the global frame as in \cref{def:global_frame}, and associated with a surface $\Psi$ in the sense of \cref{def:surface_general}. Here, the vectors $\mathbf{r_i}$, $\mathbf{r_r}$, $\mathbf{v_i}$ and $\mathbf{v_r}$, of dimension $3 \times 1$, denote the incident and reflected particle positions and velocities, while $t_i$ and $t_r$ denote the corresponding surface entry and exit times. Let $\mathbf{n_G}$ denote the global surface normal unit vector, also of dimension $3 \times 1$, defined in the same frame. Then, the function $\mathcal{K}$ satisfies \textbf{pointwise reciprocity} if
    \begin{multline} \label{eq:reciprocity_pointwise}
        \mathcal{K}\left(\mathbf{r_i} \rightarrow \mathbf{r_r}, \mathbf{v_i} \rightarrow \mathbf{v_r}, t_i \rightarrow t_r \right)
        \mathcal{Q}\left(\mathbf{v_i} \,\middle|\, \mathbf{n_G}\right)
        f_{\varphi_0}(\mathbf{r_i}, \mathbf{v_i}, t_i)
        \\
        =
        \mathcal{K}\left(\mathbf{r_r} \rightarrow \mathbf{r_i}, -\mathbf{v_r} \rightarrow -\mathbf{v_i}, -t_r \rightarrow -t_i \right)
        \mathcal{Q}\left(-\mathbf{v_r} \,\middle|\, \mathbf{n_G}\right)
        f_{\varphi_0}(\mathbf{r_r}, -\mathbf{v_r}, -t_r),
    \end{multline}
    where $\mathcal{Q}$ and $f_{\varphi_0}$ are given as in \cref{def:visible_flux} and \cref{def:reciprocity}, respectively. It follows directly that this reciprocity relation is satisfied by velocity-reciprocal kernels, in the sense of \cref{eq:reciprocity_velocity}, of the form
    \begin{equation}
        \mathcal{K}(\mathbf{r_i} \rightarrow \mathbf{r_r}, \mathbf{v_i} \rightarrow \mathbf{v_r}, t_i \rightarrow t_r)
        =
        \mathcal{K}(\mathbf{v_i} \rightarrow \mathbf{v_r})
        \delta\left(\mathbf{r_r} - \mathbf{r_i}\right)
        \delta(t_r - t_i),
    \end{equation}
    which are defined on a perfectly smooth surface $\Psi = 0_{\mathcal{P}}$, in the sense of \cref{def:surface_multiscale_height}.
\end{definition}

\begin{definition} \label{def:normalisation}
    Let $\mathcal{K} : \mathbb{R}^3 \times \mathbb{R}^3 \times \mathbb{R}
\times \mathbb{R}^3 \times \mathbb{R}^3 \times \mathbb{R}
\rightarrow [0, \infty)$, written as $\mathcal{K} = \mathcal{K}(\mathbf{r_i} \rightarrow \mathbf{r_r}, \mathbf{v_i} \rightarrow \mathbf{v_r}, t_i \rightarrow t_r)$, be a Lebesgue-integrable function defined in the global frame of \cref{def:global_frame}, where the vectors $\mathbf{r_i}$, $\mathbf{r_r}$, $\mathbf{v_i}$ and $\mathbf{v_r}$ of dimension $3 \times 1$ represent the incident and reflected particle positions and velocities, respectively, and the real values $t_i$ and $t_r$ represent the corresponding surface entry and exit times, respectively. The function $\mathcal{K}$ is said to satisfy the \textbf{normalisation} condition if
    \begin{equation} \label{eq:normalisation}
        \int_{\mathbb{R}^3}\int_{\mathbb{R}^3}\int_{\mathbb{R}}
        \mathcal{K}(\mathbf{r_i} \rightarrow \mathbf{r_r}, \mathbf{v_i} \rightarrow \mathbf{v_r}, t_i \rightarrow t_r) \, \dif t_r \dif \mathbf{v_r} \dif \mathbf{r_r}
        = 1,
        \quad \forall \, \left(\mathbf{r_i}, \mathbf{v_i}, t_i\right) \in \mathbb{R}^3 \times \mathbb{R}^3 \times \mathbb{R}.
    \end{equation}
\end{definition}

\begin{definition} \label{def:nonnegativity}
    Let $\mathcal{K} : \mathbb{R}^3 \times \mathbb{R}^3 \times \mathbb{R}
\times \mathbb{R}^3 \times \mathbb{R}^3 \times \mathbb{R}
\rightarrow [0, \infty)$, written as $\mathcal{K} = \mathcal{K}(\mathbf{r_i} \rightarrow \mathbf{r_r}, \mathbf{v_i} \rightarrow \mathbf{v_r}, t_i \rightarrow t_r)$, be a Lebesgue-integrable function defined in the global frame of \cref{def:global_frame}, where the vectors $\mathbf{r_i}$, $\mathbf{r_r}$, $\mathbf{v_i}$ and $\mathbf{v_r}$ of dimension $3 \times 1$ represent the incident and reflected particle positions and velocities, respectively, and the real values $t_i$ and $t_r$ represent the corresponding surface entry and exit times, respectively. The function $\mathcal{K}$ is said to satisfy the non-negativity condition if
    \begin{equation} \label{eq:nonnegativity}
        \mathcal{K}(\mathbf{r_i} \rightarrow \mathbf{r_r}, \mathbf{v_i} \rightarrow \mathbf{v_r}, t_i \rightarrow t_r) \geq 0
        \quad \forall \, \left(\mathbf{r_i}, \mathbf{r_r}, \mathbf{v_i}, \mathbf{v_r}, t_i, t_r\right) \in \mathbb{R}^3 \times \mathbb{R}^3 \times \mathbb{R}^3 \times \mathbb{R}^3 \times \mathbb{R} \times \mathbb{R}.
    \end{equation}
\end{definition}
\begin{definition} \label{def:kernel_set}
    We denote with $\mathcal{T}$ the set of all kernels as
    \begin{multline}
        \mathcal{T} = \Bigl\{\exists \, \mathcal{K} : \mathbb{R}^3 \times \mathbb{R}^3 \times \mathbb{R} \times \mathbb{R}^3 \times \mathbb{R}^3 \times \mathbb{R} \rightarrow [0, \infty), \, \mathcal{K} = \mathcal{K}(\mathbf{r_i} \rightarrow \mathbf{r_r}, \mathbf{v_i} \rightarrow \mathbf{v_r}, t_i \rightarrow t_r) \, \Big| \, \mathcal{K}  \text{ satisfies  } \\ \text{\textbf{normalisation} and \textbf{non-negativity} in the sense of \cref{def:normalisation,def:nonnegativity}}.\Bigr\}
    \end{multline}
\end{definition}
\begin{definition} \label{def:scattering_kernel_set}
    We denote with $\mathcal{T}^s$ the set of all admissible scattering kernels as
    \begin{multline}
        \mathcal{T}^s = \Bigl\{\exists \, \mathcal{K} : \mathbb{R}^3 \times \mathbb{R}^3 \times \mathbb{R} \times \mathbb{R}^3 \times \mathbb{R}^3 \times \mathbb{R} \rightarrow [0, \infty), \, \mathcal{K} = \mathcal{K}(\mathbf{r_i} \rightarrow \mathbf{r_r}, \mathbf{v_i} \rightarrow \mathbf{v_r}, t_i \rightarrow t_r) \, \Big| \, \mathcal{K}  \text{ satisfies  } \\ \text{\textbf{impermeability}, \textbf{reciprocity}, \textbf{normalisation} and \textbf{non-negativity} in the sense of \cref{def:impermeability,def:reciprocity,def:normalisation,def:nonnegativity}}.\Bigr\}
    \end{multline}
    By definition, $\mathcal{T}^s \subset \mathcal{T}$.
\end{definition}
\begin{definition} \label{def:point_wise_scattering_kernel_set}
    We denote with $\mathcal{T}^p$ the set of all admissible pointwise scattering kernels as
    \begin{multline}
        \mathcal{T}^p = \Bigl\{\exists \, \mathcal{K} : \mathbb{R}^3 \times \mathbb{R}^3 \times \mathbb{R} \times \mathbb{R}^3 \times \mathbb{R}^3 \times \mathbb{R} \rightarrow [0, \infty), \, \mathcal{K} = \mathcal{K}(\mathbf{r_i} \rightarrow \mathbf{r_r}, \mathbf{v_i} \rightarrow \mathbf{v_r}, t_i \rightarrow t_r) \, \Big| \, \mathcal{K}  \text{ satisfies  } \\ \text{\textbf{pointwise reciprocity}, \textbf{normalisation} and \textbf{non-negativity} in the sense of \cref{def:pointwise_reciprocity,def:normalisation,def:nonnegativity}}.\Bigr\}
    \end{multline}
    By definition, $\mathcal{T}^p \subset \mathcal{T}$.
\end{definition}

\section{Roughness-based kernel definition and admissibility} \label{sec:Roughness_kernel}

\subsection{Single-reflection kernel} \label{subsec:single_reflection_kernel}

We now have all the ingredients required to extend the scattering kernel definition in \cref{eq:general_scattering_kernel} to a rough surface $\Psi$ with associated length scale $R$. We also derive sufficient conditions under which the resulting formulation satisfies the three admissibility criteria introduced above. This construction leads naturally to the definition of scattering operators that map a local kernel $\mathcal{K_L}$ to its corresponding global counterparts by incorporating the morphological effects of $\Psi$.

We introduce two classes of kernels. The first maps an incident particle distribution to a reflected one by combining the roughness of $\Psi$ with a local scattering law described by $\mathcal{K_L}$, under the assumption that each particle undergoes a single reflection at the global scale. The second extends this construction by accounting for the possibility of multiple reflections induced by the morphology of $\Psi$. In this subsection, we define the single-reflection kernel, $\mathcal{K_K}$, under one key assumption: during each local interaction described by $\mathcal{K_L}$, the local surface normal $\mathbf{n_L}$ remains constant. Physically, this assumption is justified only when $\Psi$ is locally smooth in the local reference frame of \cref{def:local_frame}. Equivalently, the characteristic length scale $R$ associated with $\Psi$ must be much larger than the distance travelled by a particle during the local interaction, that is, $R \gg \left\lVert \mathbf{r_{r_L}} - \mathbf{r_{i_L}} \right\rVert$. Below, we formally define the single-reflection kernel and show that it satisfies pointwise reciprocity, normalisation, and non-negativity whenever $\mathcal{K_L}$ satisfies the same properties.

\begin{definition} \label{def:single_reflection_scattering_kernel}
Let $\mathcal{K_L} : \mathbb{R}^3 \times \mathbb{R}^3 \times \mathbb{R}
\times
\mathbb{R}^3 \times \mathbb{R}^3 \times \mathbb{R}
\rightarrow [0, \infty)$, written as $\mathcal{K_L} =
\mathcal{K_L}(
\mathbf{r_{i_L}} \rightarrow \mathbf{r_{r_L}},
\mathbf{v_{i_L}} \rightarrow \mathbf{v_{r_L}},
t_i \rightarrow t_r)$, be a Lebesgue-integrable function defined in the local reference frame of \cref{def:local_frame}, representing a local scattering kernel associated with the surface $\Psi$ as given by \cref{def:surface_general}. Here, $\mathbf{r_{i_L}}, \mathbf{r_{r_L}}, \mathbf{v_{i_L}}, \mathbf{v_{r_L}} \in \mathbb{R}^3$ denote the incident and reflected particle positions and velocities, respectively, expressed in the local frame, and $t_i, t_r \in \mathbb{R}$ denote the corresponding interaction times, respectively. Further, let $\mathcal{K_K} :
\mathbb{R}^3 \times \mathbb{R}^3 \times \mathbb{R}
\times
\mathbb{R}^3 \times \mathbb{R}^3 \times \mathbb{R}
\rightarrow [0, \infty)$, written as $\mathcal{K_K} = \mathcal{K_K}(
\mathbf{r_i} \rightarrow \mathbf{r_r},
\mathbf{v_i} \rightarrow \mathbf{v_r},
t_i \rightarrow t_r)$, be a Lebesgue integrable function defined in the global reference frame of \cref{def:global_frame}, where $\mathbf{r_i}, \mathbf{r_r}, \mathbf{v_i}, \mathbf{v_r} \in \mathbb{R}^3$ are expressed in the global frame. Let $\mathbf{n_G}, \mathbf{n_L} \in \mathbb{R}^3$ be unit vectors describing the local and global surface normals of $\Psi$, respectively, expressed in the global frame. Let $p_n : \mathbb{S}^2 \rightarrow [0,\infty)$, $p_n = p_n(\mathbf{n_L} \, | \, \mathbf{n_G}, \mathbf{v_i})$, be the probability density function of the local surface normals on $\Psi$ for particles with velocity $\mathbf{v_i}$. 

Then $\mathcal{K_K}$ is said to be the \textbf{single-reflection global kernel} associated with $\mathcal{K_L}$ if
\begin{multline} \label{eq:single_reflection_kernel}
\boxed{
    \mathcal{K_K}(\mathbf{r_{i}} \rightarrow \mathbf{r_{r}}, \mathbf{v_{i}} \rightarrow \mathbf{v_{r}}, t_i \rightarrow t_r)
        =
        \int_{\mathbb{S}^2}
        \mathcal{K_L}(\mathbf{r_{i_L}} \rightarrow \mathbf{r_{r_L}}, \mathbf{v_{i_L}} \rightarrow \mathbf{v_{r_L}}, t_i \rightarrow t_r) 
         \frac{\mathcal{Q}(\mathbf{v_{i_L}}\, | \,  \mathbf{n_L})}
    {\mathcal{Q}( \mathbf{v_{i}} \, | \, \mathbf{n_G})} \,
        p_n(\mathbf{n_L} \, | \, \mathbf{n_G}, \mathbf{v_i})
        \, \dif \mathbf{n_L}.
    }
\end{multline}
\end{definition}

\begin{lemma} \label{lemma:global_reciprocity}
    Let $\mathcal{K_K} : \mathbb{R}^3 \times \mathbb{R}^3 \times \mathbb{R} \times \mathbb{R}^3 \times \mathbb{R}^3 \times \mathbb{R} \rightarrow [0,\infty)$, written as $\mathcal{K_K} = \mathcal{K_K}\left(\mathbf{r_i} \rightarrow \mathbf{r_r}, \mathbf{v_i} \rightarrow \mathbf{v_r}, t_i \rightarrow t_r\right)$, denote a Lebesgue-integrable function associated with the global reference frame $G$ of a surface $\Psi$ in the sense of \cref{def:global_frame,def:surface_general}, and defined as in \cref{def:single_reflection_scattering_kernel}, such that
    \begin{multline}
         \mathcal{K_K}(\mathbf{r_{i}} \rightarrow \mathbf{r_{r}}, \mathbf{v_{i}} \rightarrow \mathbf{v_{r}}, t_i \rightarrow t_r)
        =
        \int_{\mathbb{S}^2}
        \mathcal{K_L}(\mathbf{r_{i_L}} \rightarrow \mathbf{r_{r_L}}, \mathbf{v_{i_L}} \rightarrow \mathbf{v_{r_L}}, t_i \rightarrow t_r)  
         \frac{\mathcal{Q}( \mathbf{v_{i_L}} \, | \, \mathbf{n_L})}
    {\mathcal{Q}( \mathbf{v_{i}} \, | \, \mathbf{n_G})} \,
        p_n(\mathbf{n_L} \, | \, \mathbf{n_G}, \mathbf{v_i})
        \, \dif \mathbf{n_L}.
    \end{multline}
    where $\mathcal{K_L}: \mathbb{R}^3 \times \mathbb{R}^3 \times \mathbb{R}
\times \mathbb{R}^3 \times \mathbb{R}^3 \times \mathbb{R} \rightarrow[0, \infty)$, written as $\mathcal{K_L}= \mathcal{K_L}(\mathbf{r_{i_L}} \rightarrow \mathbf{r_{r_L}}, \mathbf{v_{i_L}} \rightarrow \mathbf{v_{r_L}}, t_i \rightarrow t_r) \in \mathcal{T}$ is a Lebesgue-integrable local scattering kernel defined in a local reference frame $L$, as given by \cref{def:local_frame}. Here, $\mathcal{T}$ is the set of admissible scattering kernels in the sense of \cref{def:scattering_kernel_set}, $\mathbf{r_{i_L}}$, $\mathbf{r_{r_L}}$, $\mathbf{v_{i_L}}$, $\mathbf{v_{r_L}}$ are $3\times 1$ real vectors denoting the incident and reflected particle positions and velocities, respectively, expressed in that frame, $t_i$ and $t_r$ represent the corresponding entry and exit times, respectively, and $\mathbf{n_L}$ and $\mathbf{n_G}$ are $3\times 1$ real unit vectors denoting the local and global surface normals, respectively. Then $\mathcal{K_K}$ satisfies pointwise reciprocity in the sense of \cref{def:reciprocity} under the assumption of smoothness of $\Psi$ in $L$ in the sense of \cref{def:local_smoothness} and local invariance in the surface normal with the travelled particle trajectory under time reversal, i.e. 
\begin{equation} \label{eq:normal_time_reversal}
    p_n(\mathbf{n_L} \, | \, \mathbf{n_G}, \mathbf{v_i}) = p_n(\mathbf{n_L} \, | \, \mathbf{n_G}, -\mathbf{v_r}),
\end{equation}
whenever the local kernel $\mathcal{K_L}$ satisfies pointwise reciprocity as well.
\end{lemma}
\begin{proof}
    We begin by substituting the expression of the single-reflection global kernel from \cref{eq:single_reflection_kernel} into the left-hand side of \cref{eq:reciprocity_pointwise}, obtaining
    \begin{equation}
    \begin{aligned}
          \mathcal{K_K}\left(\mathbf{r_{i}} \rightarrow \mathbf{r_{r}}, \mathbf{v_{i}} \rightarrow \mathbf{v_{r}}, t_i \rightarrow t_r\right)  \mathcal{Q}(\mathbf{v_i} \, | \,  \mathbf{n_G}) f_{\varphi_0}(\mathbf{r_i}, \mathbf{v_i}, t_i) & = \int_{\mathbb{S}^2} \mathcal{K_L}\left(\mathbf{r_{i_L}} \rightarrow \mathbf{r_{r_L}}, \mathbf{v_{i_L}} \rightarrow \mathbf{v_{r_L}}, t_i \rightarrow t_r\right) \\
          & \cdot  p_n(\mathbf{n_L} \, | \, \mathbf{n_G}, \mathbf{v_i})  \mathcal{Q}( \mathbf{v_i} \, | \,  \mathbf{n_G}) \frac{\mathcal{Q}( \mathbf{v_{i_L}}\, | \, \mathbf{n_L})}
    {\mathcal{Q}( \mathbf{v_{i}} \, | \, \mathbf{n_G})} \\
        & \cdot \dif \mathbf{n_L} \, f_{\varphi_0}(\mathbf{r_i}, \mathbf{v_i}, t_i) \\
          & = \int_{\mathbb{S}^2} \mathcal{K_L}\left(\mathbf{r_{i_L}} \rightarrow \mathbf{r_{r_L}}, \mathbf{v_{i_L}} \rightarrow \mathbf{v_{r_L}}, t_i \rightarrow t_r\right)
          \\ & \cdot \mathcal{Q}( \mathbf{v_{i_L}} \, | \, \mathbf{n_L}) p_n(\mathbf{n_L} \, | \, \mathbf{n_G}, \mathbf{v_i}) \, \dif \mathbf{n_L} \, f_{\varphi_0}(\mathbf{r_i}, \mathbf{v_i}, t_i).
    \end{aligned}
    \end{equation}
    Since $f_{\varphi_0}(\mathbf{r_i}, \mathbf{v_i}, t_i)$ is invariant in the velocity direction $\mathbf{u_i}$, i.e. $f_{\varphi_0}(\mathbf{r_i}, \mathbf{v_i}, t_i) = f_{\varphi_0}(\mathbf{r_i}, v_i, t_i)$, and the transformation from the global to the local reference frame $\mathbf{v_{i_L}} = T_{G\rightarrow L}(\mathbf{v_i}) = T_y(\theta_{n_1}) T_z(\theta_{n_2}) \mathbf{v_i}$ is invariant in magnitude, i.e. $v_{i_L} = v_i$, and has a unity Jacobian determinant, we can rewrite the left-hand side expression as
    \begin{equation}
    \begin{aligned}
        \mathcal{K_K}\left(\mathbf{r_{i}} \rightarrow \mathbf{r_{r}}, \mathbf{v_{i}} \rightarrow \mathbf{v_{r}}, t_i \rightarrow t_r\right)  \mathcal{Q}(\mathbf{v_i} \, | \,  \mathbf{n_G}) f_{\varphi_0}(\mathbf{r_i}, \mathbf{v_i}, t_i) & = \int_{\mathbb{S}^2} \mathcal{K_L}\left(\mathbf{r_{i_L}} \rightarrow \mathbf{r_{r_L}}, \mathbf{v_{i_L}} \rightarrow \mathbf{v_{r_L}}, t_i \rightarrow t_r\right) \\
        & \cdot \mathcal{Q}( \mathbf{v_{i_L}} \, | \, \mathbf{n_L}) \, f_{{\varphi_L}}(\mathbf{r_{i_L}}, \mathbf{v_{i_L}}, t_i) \, p_n(\mathbf{n_L} \, | \, \mathbf{n_G}, \mathbf{v_i}) \, \dif \mathbf{n_L},
    \end{aligned}
    \end{equation}
    where we have also introduced the PDF $f_{\varphi_{0_L}} : \mathbb{R}^3 \times \mathbb{R}^3 \times \mathbb{R} \rightarrow [0, \infty)$, $f_{\varphi_{0_L}}(\mathbf{r_{i_L}}, \mathbf{v_{i_L}}, t_i) =  f_{\varphi_0}(\mathbf{r_i}, \mathbf{v_i}, t_i) \, \forall (\mathbf{r_i}, \mathbf{v_i}, t_i) \in \mathbb{R}^3 \times \mathbb{R}^3 \times \mathbb{R}$ defining the position probability in the local frame at time $t_i$. Next, we employ the fact that $\Psi$ is smooth in the $L$ frame in the sense of \cref{def:local_smoothness}, and hence $\mathcal{Q}( \mathbf{v_{i_L}} \, | \, \mathbf{n_L}) = \left| \mathbf{v_{i_L}} \cdot \left.\mathbf{n_L}\right|^L\right|$ to rewrite the equality as
    \begin{equation}
        \begin{aligned}
        \mathcal{K_K}\left(\mathbf{r_{i}} \rightarrow \mathbf{r_{r}}, \mathbf{v_{i}} \rightarrow \mathbf{v_{r}}, t_i \rightarrow t_r\right)  \mathcal{Q}( \mathbf{v_i} \, | \,  \mathbf{n_G}) f_{\varphi_0}(\mathbf{r_i}, \mathbf{v_i}, t_i) & = \int_{\mathbb{S}^2} \mathcal{K_L}\left(\mathbf{r_{i_L}} \rightarrow \mathbf{r_{r_L}}, \mathbf{v_{i_L}} \rightarrow \mathbf{v_{r_L}}, t_i \rightarrow t_r\right) \\
        & \cdot \mathcal{Q}( \mathbf{v_{i_L}}\, | \, \mathbf{n_L}) \, f_{\varphi_{0_L}}(\mathbf{r_{i_L}}, \mathbf{v_{i_L}}, t_i) \, p_n(\mathbf{n_L} \, | \, \mathbf{n_G}, \mathbf{v_i}) \, \dif \mathbf{n_L} \\
            & = \int_{\mathbf{v_i} \cdot \mathbf{n_L} < 0} \mathcal{K_L}\left(\mathbf{r_{i_L}} \rightarrow \mathbf{r_{r_L}}, \mathbf{v_{i_L}} \rightarrow \mathbf{v_{r_L}}, t_i \rightarrow t_r\right) \\
        & \cdot \left| \mathbf{v_{i_L}} \cdot \left.\mathbf{n_L}\right|^L\right| \, f_{\varphi_{0_L}}(\mathbf{r_{i_L}}, \mathbf{v_{i_L}}, t_i) \, p_n(\mathbf{n_L} \, | \, \mathbf{n_G}, \mathbf{v_i}) \, \dif \mathbf{n_L},
        \end{aligned}
    \end{equation}
    where $\left.\mathbf{n_{L}}\right|^L = \begin{bmatrix}
        0 & 0 & 1
    \end{bmatrix}^T$ is just the local normal vector, expressed in the local reference frame. We now employ the pointwise reciprocity property of the local kernel $\mathcal{K_L}\left(\mathbf{r_{i_L}} \rightarrow \mathbf{r_{r_L}}, \mathbf{v_{i_L}} \rightarrow \mathbf{v_{r_L}}, t_i \rightarrow t_r\right)$, as given in \cref{def:reciprocity}, to obtain
    \begin{equation}
        \begin{aligned}
           \mathcal{K_K}\left(\mathbf{r_{i}} \rightarrow \mathbf{r_{r}}, \mathbf{v_{i}} \rightarrow \mathbf{v_{r}}, t_i \rightarrow t_r\right)  \mathcal{Q}(\mathbf{v_i} \, | \,  \mathbf{n_G}) f_{\varphi_0}(\mathbf{r_i}, \mathbf{v_i}, t_i) & = \int_{-\mathbf{v_r} \cdot \mathbf{n_L} < 0} \mathcal{K_L}\left(\mathbf{r_{r_L}} \rightarrow \mathbf{r_{i_L}}, -\mathbf{v_{r_L}} \rightarrow -\mathbf{v_{i_L}}, -t_r \rightarrow -t_i\right) \\
            & \qquad \cdot  \left|-\mathbf{v_{r_L}} \cdot \left.\mathbf{n_{L}}\right|^L\right| f_{\varphi_{0_L}}(\mathbf{r_{r_L}}, -\mathbf{v_{r_L}}, -t_r) \,p_n(\mathbf{n_L} \, | \, \mathbf{n_G}, -\mathbf{v_r}) \, \dif \mathbf{n_L},
        \end{aligned}
    \end{equation}
    where we have used the impermeability property of $\mathcal{K_L}$ as given in \cref{def:impermeability} to change the integration limits and the local normal invariance to substitute $p_n(\mathbf{n_L} \, | \, \mathbf{n_G}, -\mathbf{v_r}) = p_n(\mathbf{n_L} \, | \, \mathbf{n_G}, \mathbf{v_i})$. Finally, we use the invariance of the dot product under rotations to rewrite the terms inside the absolute value on the right-hand side in the global frame. We also express $f_{\varphi_{0_L}}(\mathbf{r_{r_L}}, -\mathbf{v_{r_L}}, t_r)$ in the global frame and move it outside the integral, yielding
    \begin{equation}
        \begin{aligned}
            \mathcal{K_K}\left(\mathbf{r_{i}} \rightarrow \mathbf{r_{r}}, \mathbf{v_{i}} \rightarrow \mathbf{v_{r}}, t_i \rightarrow t_r\right)  \mathcal{Q}(\mathbf{v_i} \, | \,  \mathbf{n_G}) f_{\varphi_0}(\mathbf{r_i}, \mathbf{v_i}, t_i) & =\int_{-\mathbf{v_r} \cdot \mathbf{n_L} < 0} \mathcal{K_L}\left(\mathbf{r_{r_L}} \rightarrow \mathbf{r_{i_L}}, -\mathbf{v_{r_L}} \rightarrow -\mathbf{v_{i_L}}, -t_r \rightarrow -t_i\right) \\
            & \qquad \cdot \left|-\mathbf{v_{r}} \cdot \mathbf{n_L}\right| f_{\varphi_0}(\mathbf{r_r}, -\mathbf{v_r}, -t_r) \, p_n(\mathbf{n_L} \, | \, \mathbf{n_G}, -\mathbf{v_r}) \, \dif \mathbf{n_L} \, \\
            & = \int_{\mathbb{S}^2} \mathcal{K_L}\left(\mathbf{r_{r_L}} \rightarrow \mathbf{r_{i_L}}, -\mathbf{v_{r_L}} \rightarrow -\mathbf{v_{i_L}}, -t_r \rightarrow -t_i\right) \\
            & \qquad \cdot p_n(\mathbf{n_L} \, | \, \mathbf{n_G}, -\mathbf{v_r}) \mathcal{Q}( -\mathbf{v_{r_L}} \, | \, \mathbf{n_L})  \, \dif \mathbf{n_L} \\
            & \qquad \cdot f_{\varphi_0}(\mathbf{r_r}, -\mathbf{v_r}, -t_r) \\
            & = \int_{\mathbb{S}^2} \mathcal{K_L}\left(\mathbf{r_{r_L}} \rightarrow \mathbf{r_{i_L}}, -\mathbf{v_{r_L}} \rightarrow -\mathbf{v_{i_L}}, -t_r \rightarrow -t_i\right) \\
            & \qquad \cdot p_n(\mathbf{n_L} \, | \, \mathbf{n_G}, -\mathbf{v_r})\frac{\mathcal{Q}( -\mathbf{v_{r_L}} \, | \, \mathbf{n_L})}{\mathcal{Q}( -\mathbf{v_{r}} \, | \, \mathbf{n_G})}  \, \dif \mathbf{n_L}  \\
            & \qquad \cdot \mathcal{Q}( -\mathbf{v_{r}} \, | \, \mathbf{n_G}) \, f_{\varphi_0}(\mathbf{r_r}, -\mathbf{v_r}, -t_r) \\
            & = \mathcal{K_K}\left(\mathbf{r_{r}} \rightarrow \mathbf{r_{i}}, -\mathbf{v_{r}} \rightarrow -\mathbf{v_{i}}, -t_r \rightarrow -t_i\right) \\
            & \qquad \cdot \mathcal{Q}( -\mathbf{v_r} \, | \,  \mathbf{n_G}) f_{\varphi_0}(\mathbf{r_r}, -\mathbf{v_r}, -t_r)
        \end{aligned}
    \end{equation}
\end{proof}
\begin{lemma} \label{lemma:global_normalisation}
    Let $\mathcal{K_K} : \mathbb{R}^3 \times \mathbb{R}^3 \times \mathbb{R} \times \mathbb{R}^3 \times \mathbb{R}^3 \times \mathbb{R} \rightarrow [0,\infty)$, written as $\mathcal{K_K} = \mathcal{K_K}\left(\mathbf{r_i} \rightarrow \mathbf{r_r}, \mathbf{v_i} \rightarrow \mathbf{v_r}, t_i \rightarrow t_r\right)$, denote a Lebesgue-integrable function associated with the global reference frame $G$ of a surface $\Psi$ in the sense of \cref{def:global_frame,def:surface_general}, and defined as in \cref{def:single_reflection_scattering_kernel}, such that
    \begin{multline}
         \mathcal{K_K}(\mathbf{r_{i}} \rightarrow \mathbf{r_{r}}, \mathbf{v_{i}} \rightarrow \mathbf{v_{r}}, t_i \rightarrow t_r)
        =
        \int_{\mathbb{S}^2}
        \mathcal{K_L}(\mathbf{r_{i_L}} \rightarrow \mathbf{r_{r_L}}, \mathbf{v_{i_L}} \rightarrow \mathbf{v_{r_L}}, t_i \rightarrow t_r)  
         \frac{\mathcal{Q}(\mathbf{v_{i_L}} \, | \, \mathbf{n_L})}
    {\mathcal{Q}(\mathbf{v_{i}} \, | \, \mathbf{n_G})} \,
        p_n(\mathbf{n_L} \, | \, \mathbf{n_G}, \mathbf{v_i})
        \, \dif \mathbf{n_L}.
    \end{multline}
 where $\mathcal{K_L}: \mathbb{R}^3 \times \mathbb{R}^3 \times \mathbb{R}
\times \mathbb{R}^3 \times \mathbb{R}^3 \times \mathbb{R} \rightarrow[0, \infty)$, written as $\mathcal{K_L}= \mathcal{K_L}(\mathbf{r_{i_L}} \rightarrow \mathbf{r_{r_L}}, \mathbf{v_{i_L}} \rightarrow \mathbf{v_{r_L}}, t_i \rightarrow t_r) \in \mathcal{T}$ is a Lebesgue-integrable local scattering kernel defined in a local reference frame $L$, as given by \cref{def:local_frame}. Here, $\mathcal{T}$ is the set of admissible scattering kernels in the sense of \cref{def:scattering_kernel_set}, $\mathbf{r_{i_L}}$, $\mathbf{r_{r_L}}$, $\mathbf{v_{i_L}}$, $\mathbf{v_{r_L}}$ are $3\times 1$ real vectors denoting the incident and reflected particle positions and velocities, respectively, expressed in that frame, $t_i$ and $t_r$ represent the corresponding entry and exit times, respectively, and $\mathbf{n_L}$ and $\mathbf{n_G}$ are $3\times 1$ real unit vectors denoting the local and global surface normals, respectively. Then, the function $\mathcal{K_K}$ satisfies the normalisation condition in the sense of \cref{def:normalisation} whenever $\Psi$ is smooth in $L$ in the sense of \cref{def:local_smoothness}, and the local kernel $\mathcal{K_L}$ satisfies the same condition. 
\end{lemma}
\begin{proof}
    We begin by substituting the single-reflection kernel expression in \cref{def:single_reflection_scattering_kernel} into the normalisation condition given by \cref{eq:normalisation}, and obtain
    \begin{equation}
    \begin{aligned}
        \int_{\mathbb{R}^3}\int_{\mathbb{R}^3}\int_{\mathbb{R}}\mathcal{K_K}(\mathbf{r_{i}} \rightarrow \mathbf{r_{r}}, \mathbf{v_{i}} \rightarrow \mathbf{v_{r}}, t_i \rightarrow t_r) \, \dif t_r \dif \mathbf{v_r} \dif \mathbf{r_r} & = \int_{\mathbb{R}^3}\int_{\mathbb{R}^3}\int_{\mathbb{R}}\int_{\mathbb{S}^2} \mathcal{K_L}(\mathbf{r_{i_L}} \rightarrow \mathbf{r_{r_L}}, \mathbf{v_{i_L}} \rightarrow \mathbf{v_{r_L}}, t_i \rightarrow t_r) 
        \\ & \cdot  \frac{\mathcal{Q}( \mathbf{v_{i_L}} \, | \, \mathbf{n_L})}
    {\mathcal{Q}( \mathbf{v_{i}}\, | \, \mathbf{n_G})} \,
        p_n(\mathbf{n_L} \, | \, \mathbf{n_G}, \mathbf{v_i}) \, \dif \mathbf{n_L}\, \dif t_r \dif \mathbf{v_r} \dif \mathbf{r_r}.
    \end{aligned}
    \end{equation}
    We can now change the order of integration between the local normal $\mathbf{n_L}$ and the normalisation variables $\mathbf{r_r}$, $\mathbf{v_r}$, and $t_r$ to obtain
    \begin{equation}
    \begin{aligned}
        \int_{\mathbb{R}^3}\int_{\mathbb{R}^3}\int_{\mathbb{R}}\mathcal{K_K}(\mathbf{r_{i}} \rightarrow \mathbf{r_{r}}, \mathbf{v_{i}} \rightarrow \mathbf{v_{r}}, t_i \rightarrow t_r) \,  \dif t_r \dif \mathbf{v_r} \dif \mathbf{r_r} & = \int_{\mathbb{S}^2} \int_{\mathbb{R}^3}\int_{\mathbb{R}^3}\int_{\mathbb{R}}\mathcal{K_L}(\mathbf{r_{i_L}} \rightarrow \mathbf{r_{r_L}}, \mathbf{v_{i_L}} \rightarrow \mathbf{v_{r_L}}, t_i \rightarrow t_r) \\
        & \cdot  \dif t_r \dif \mathbf{v_r} \dif \mathbf{r_r}  \frac{\mathcal{Q}( \mathbf{v_{i_L}} \, | \, \mathbf{n_L})}
    {\mathcal{Q}( \mathbf{v_{i}} \, | \, \mathbf{n_G})} \,
        p_n(\mathbf{n_L} \, | \, \mathbf{n_G}, \mathbf{v_i}) \, \dif \mathbf{n_L}.
    \end{aligned}
    \end{equation}
    Now, we change the integration variable of the inner integral, from $\mathbf{v_r}$ to $\mathbf{v_{r_L}}$, minding that a change in coordinate system results in a Jacobian equal to unity in both position and velocity,
      \begin{equation}
    \begin{aligned}
        \int_{\mathbb{R}^3}\int_{\mathbb{R}^3}\int_{\mathbb{R}}\mathcal{K_K}(\mathbf{r_{i}} \rightarrow \mathbf{r_{r}}, \mathbf{v_{i}} \rightarrow \mathbf{v_{r}}, t_i \rightarrow t_r) \,  \dif t_r \dif \mathbf{v_r} \dif \mathbf{r_r} & = \int_{\mathbb{S}^2} \int_{\mathbb{R}^3}\int_{\mathbb{R}^3}\int_{\mathbb{R}}\mathcal{K_L}(\mathbf{r_{i_L}} \rightarrow \mathbf{r_{r_L}}, \mathbf{v_{i_L}} \rightarrow \mathbf{v_{r_L}}, t_i \rightarrow t_r) \\
        & \cdot \dif t_r \dif \mathbf{v_{r_L}} \dif \mathbf{r_{r_L}}  \frac{\mathcal{Q}( \mathbf{v_{i_L}} \, | \, \mathbf{n_L})}
    {\mathcal{Q}( \mathbf{v_{i}} \, | \, \mathbf{n_G})} \,
        p_n(\mathbf{n_L} \, | \, \mathbf{n_G}, \mathbf{v_i}) \, \dif \mathbf{n_L}.
    \end{aligned}
    \end{equation}
    Finally, we make use of the normalisation property of the local kernel $\mathcal{K_L}(\mathbf{r_{i_L}} \rightarrow \mathbf{r_{r_L}}, \mathbf{v_{i_L}} \rightarrow \mathbf{v_{r_L}}, t_i \rightarrow t_r)$, to arrive at
    \begin{equation}
        \begin{aligned}
            \int_{\mathbb{R}^3}\int_{\mathbb{R}^3}\int_{\mathbb{R}}\mathcal{K_K}(\mathbf{r_{i}} \rightarrow \mathbf{r_{r}}, \mathbf{v_{i}} \rightarrow \mathbf{v_{r}}, t_i \rightarrow t_r) \,  \dif t_r \dif \mathbf{v_r} \dif \mathbf{r_r} & = \int_{\mathbb{S}^2} \int_{\mathbb{R}^3}\int_{\mathbb{R}^3}\int_{\mathbb{R}}\mathcal{K_L}(\mathbf{r_{i_L}} \rightarrow \mathbf{r_{r_L}}, \mathbf{v_{i_L}} \rightarrow \mathbf{v_{r_L}}, t_i \rightarrow t_r) \\
        & \cdot \dif t_r \dif \mathbf{v_{r_L}} \dif \mathbf{r_{r_L}}  \frac{\mathcal{Q}( \mathbf{v_{i_L}} \, | \, \mathbf{n_L})}
    {\mathcal{Q}(\mathbf{v_{i}} \, | \, \mathbf{n_G})} \,
        p_n(\mathbf{n_L} \, | \, \mathbf{n_G}, \mathbf{v_i}) \, \dif \mathbf{n_L} \\ 
            & = \int_{\mathbb{S}^2}  \frac{\mathcal{Q}( \mathbf{v_{i_L}} \, | \, \mathbf{n_L})}
    {\mathcal{Q}(\mathbf{v_{i}} \, | \, \mathbf{n_G})} \,
        p_n(\mathbf{n_L} \, | \, \mathbf{n_G}, \mathbf{v_i}) \, \dif \mathbf{n_L}.
        \end{aligned}
    \end{equation}
    We now remember that $\Psi$ is smooth in $L$ in the sense of \cref{def:local_smoothness}, and thus, $\mathcal{Q}( \mathbf{v_{i_L}} \, | \, \mathbf{n_L}) = \left|\mathbf{v_{i_L}} \cdot \left.\mathbf{n_L}\right|^L \right| = \left|\mathbf{v_{i}} \cdot \mathbf{n_L} \right|$. Then, under \cref{def:visible_flux},
    \begin{equation}
        \begin{aligned}
            \int_{\mathbb{R}^3}\int_{\mathbb{R}^3}\int_{\mathbb{R}}\mathcal{K_K}(\mathbf{r_{i}} \rightarrow \mathbf{r_{r}}, \mathbf{v_{i}} \rightarrow \mathbf{v_{r}}, t_i \rightarrow t_r) \,  \dif t_r \dif \mathbf{v_r} \dif \mathbf{r_r} & = \int_{\mathbf{v_i} \cdot \mathbf{n_L} < 0}  \frac{ \left|\mathbf{v_{i}} \cdot \mathbf{n_L} \right|}
    {\mathcal{Q}(\mathbf{v_{i}} \, | \, \mathbf{n_G})} \,
        p_n(\mathbf{n_L} \, | \, \mathbf{n_G}, \mathbf{v_i}) \, \dif \mathbf{n_L} = 1.
        \end{aligned}
    \end{equation}
\end{proof}
\begin{lemma}\label{lemma:nonnegativity}
    Let $\mathcal{K_K} : \mathbb{R}^3 \times \mathbb{R}^3 \times \mathbb{R} \times \mathbb{R}^3 \times \mathbb{R}^3 \times \mathbb{R} \rightarrow [0,\infty)$, written as $\mathcal{K_K} = \mathcal{K_K}\left(\mathbf{r_i} \rightarrow \mathbf{r_r}, \mathbf{v_i} \rightarrow \mathbf{v_r}, t_i \rightarrow t_r\right)$, denote a Lebesgue-integrable function associated with the global reference frame $G$ of a surface $\Psi$ in the sense of \cref{def:global_frame,def:surface_general}, and defined as in \cref{def:single_reflection_scattering_kernel}, such that
    \begin{multline}
         \mathcal{K_K}(\mathbf{r_{i}} \rightarrow \mathbf{r_{r}}, \mathbf{v_{i}} \rightarrow \mathbf{v_{r}}, t_i \rightarrow t_r)
        =
        \int_{\mathbb{S}^2}
        \mathcal{K_L}(\mathbf{r_{i_L}} \rightarrow \mathbf{r_{r_L}}, \mathbf{v_{i_L}} \rightarrow \mathbf{v_{r_L}}, t_i \rightarrow t_r)  
         \frac{\mathcal{Q}( \mathbf{v_{i_L}} \, | \, \mathbf{n_L})}
    {\mathcal{Q}( \mathbf{v_{i}} \, | \, \mathbf{n_G})} \,
        p_n(\mathbf{n_L} \, | \, \mathbf{n_G}, \mathbf{v_i})
        \, \dif \mathbf{n_L}.
    \end{multline}
where $\mathcal{K_L}: \mathbb{R}^3 \times \mathbb{R}^3 \times \mathbb{R}
\times \mathbb{R}^3 \times \mathbb{R}^3 \times \mathbb{R} \rightarrow[0, \infty)$, written as $\mathcal{K_L}= \mathcal{K_L}(\mathbf{r_{i_L}} \rightarrow \mathbf{r_{r_L}}, \mathbf{v_{i_L}} \rightarrow \mathbf{v_{r_L}}, t_i \rightarrow t_r) \in \mathcal{T}$ is a Lebesgue-integrable local scattering kernel defined in a local reference frame $L$, as given by \cref{def:local_frame}. Here, $\mathcal{T}$ is the set of admissible scattering kernels in the sense of \cref{def:scattering_kernel_set}, $\mathbf{r_{i_L}}$, $\mathbf{r_{r_L}}$, $\mathbf{v_{i_L}}$, $\mathbf{v_{r_L}}$ are $3\times 1$ real vectors denoting the incident and reflected particle positions and velocities, respectively, expressed in that frame, $t_i$ and $t_r$ represent the corresponding entry and exit times, respectively, and $\mathbf{n_L}$ and $\mathbf{n_G}$ are $3\times 1$ real unit vectors denoting the local and global surface normals, respectively. Then, the function  $\mathcal{K_K}$ satisfies the non-negativity condition in the sense of \cref{def:nonnegativity}, provided that $\mathcal{K_L}$ satisfies the same condition.
\end{lemma}
\begin{proof}
    The proof is trivial. We immediately see that $\frac{\mathcal{Q}( \mathbf{v_{i_L}} \, | \, \mathbf{n_L})}
    {\mathcal{Q}( \mathbf{v_{i}} \, | \, \mathbf{n_G})} \geq 0, \, \forall (\mathbf{r_i}, \mathbf{v_i}, t_i, \mathbf{n_L}, \mathbf{n_G}) \in \mathbb{R}^3 \times \mathbb{R}^3 \times \mathbb{R} \times \mathbb{R}^3 \times \mathbb{R}^3$, and $p_n(\mathbf{n_L} \, | \, \mathbf{n_G}, \mathbf{v_i}) \geq 0, \, \forall (\mathbf{n_L}, \mathbf{n_G}, \mathbf{r_i}, t_i) \in \mathbb{S}^2 \times \mathbb{S}^2 \times \mathbb{R}^3 \times \mathbb{R}$. We further know that $\mathcal{K_L}$ satisfies non-negativity, i.e. $\mathcal{K_L}(\mathbf{r_{i_L}} \rightarrow \mathbf{r_{r_L}}, \mathbf{v_{i_L}} \rightarrow \mathbf{v_{r_L}}, t_i \rightarrow t_r) \geq 0, \, \forall (\mathbf{r_{i_L}}, \mathbf{r_{r_L}}, \mathbf{v_{i_L}}, \mathbf{v_{r_{L}}}, t_i, t_r) \in \mathbb{R}^3 \times \mathbb{R}^3 \times \mathbb{R}^3 \times \mathbb{R}^3 \times \mathbb{R} \times \mathbb{R}$. Since all terms in the integral are non-negative, we can infer that $\mathcal{K_K}\left(\mathbf{r_i} \rightarrow \mathbf{r_r}, \mathbf{v_i} \rightarrow \mathbf{v_r}, t_i \rightarrow t_r\right)$ is non-negative for any $\mathbf{r_r}, \mathbf{v_r}, t_r, \mathbf{n_L} \in \mathbb{R}^3 \times \mathbb{R}^3 \times [0 , \infty) \times \mathbb{R}^3$, and thus satisfies the condition. \qedhere
\end{proof}

\subsection{Multi-reflection kernel} \label{subsec:multi_reflection_kernel}

We now use the single-reflection kernel, $\mathcal{K_K}$, to construct a multi-reflection kernel, $\mathcal{K_{MK}}$, describing the scattering of gas particles under the action of a rough surface $\Psi$ for which particles may undergo multiple successive reflections before escaping. The probability of re-collision is controlled by the two-point shadowing function, $\mathcal{S}(\mathbf{r_i}, \mathbf{v_i}, t_i \, | \, \mathbf{r_r}, \mathbf{v_r}, t_r)$ associated with $\Psi$. Below, we formally define this kernel and provide a sufficient set of conditions for $\mathcal{S}$ under which it preserves impermeability, reciprocity, normalisation, and non-negativity, provided that its associated local kernel, $\mathcal{K_L}$ satisfies impermeability, pointwise reciprocity, normalisation and non-negativity.

\begin{definition} \label{def:multi_reflection_scattering_kernel}
Let 
$\mathcal{K_K} :
\mathbb{R}^3 \times \mathbb{R}^3 \times \mathbb{R}
\times
\mathbb{R}^3 \times \mathbb{R}^3 \times \mathbb{R}
\rightarrow [0, \infty)$,
written as
$\mathcal{K_K}(
\mathbf{r_i} \rightarrow \mathbf{r_r},
\mathbf{v_i} \rightarrow \mathbf{v_r},
t_i \rightarrow t_r)$,
be a Lebesgue-integrable function defined in the global reference frame of \cref{def:global_frame}, representing the single-reflection global kernel of \cref{def:single_reflection_scattering_kernel} associated with a surface $\Psi$ as given by \cref{def:surface_general}, where
$\mathbf{r_i}, \mathbf{r_r}, \mathbf{v_i}, \mathbf{v_r} \in \mathbb{R}^3$
are expressed in the global frame. Further, let
$\mathcal{K_{MK}} :
\mathbb{R}^3 \times \mathbb{R}^3 \times \mathbb{R}
\times
\mathbb{R}^3 \times \mathbb{R}^3 \times \mathbb{R}
\rightarrow [0, \infty)$,
written as
$\mathcal{K_{MK}}(
\mathbf{r_i} \rightarrow \mathbf{r_r},
\mathbf{v_i} \rightarrow \mathbf{v_r},
t_i \rightarrow t_r)$,
be a Lebesgue-integrable function defined in the same global reference frame.

Let $\mathcal{S} : \mathbb{R}^3 \times \mathbb{R}^3 \times \mathbb{R} \times \mathbb{R}^3 \times \mathbb{R}^3 \times \mathbb{R} \rightarrow (0,1]$, written as $\mathcal{S}=\mathcal{S}(\mathbf{r_i},\mathbf{v_i},t_i \mid \mathbf{r_r},\mathbf{v_r},t_r)$, denote the two-point shadowing function associated with surface $\Psi$ and autocorrelation length $R$, as defined in \cref{def:shadowing_single_scale,def:length_scale}. Furthermore, let $\mathcal{S} : \mathbb{R}^3 \times \mathbb{R}^3 \times \mathbb{R} \rightarrow (0,1]$, $\mathcal{S}=\mathcal{S}(\mathbf{r_r},\mathbf{v_r},t_r)$, denotes the corresponding single-point one-point shadowing function associated with $\Psi$ and $R$. Then $\mathcal{K_{MK}}$ is said to be the \textbf{multi-reflection global kernel} associated with $\Psi$ if
\begin{equation} \label{eq:multi_reflection_kernel}
\boxed{
    \begin{aligned}
        \mathcal{K_{MK}}(
        \mathbf{r_i} \rightarrow \mathbf{r_r},
        \mathbf{v_i} \rightarrow \mathbf{v_r},
        t_i \rightarrow t_r)
        &=
        \sum_{n=1}^{\infty}
        \Bigg\{
        \underset{\mathbf{v_{r_k}},\mathbf{v_{i_k}}}{\int\cdots\int}
        \underset{\mathbf{r_{r_k}},\mathbf{r_{i_k}}}{\int\cdots\int}
        \underset{t_{r_k},t_{i_k}}{\int\cdots\int}
        \mathcal{K_K}(
        \mathbf{r_i}\rightarrow\mathbf{r_{r_1}},
        \mathbf{v_i}\rightarrow\mathbf{v_{r_1}},
        t_i\rightarrow t_{r_1})
        \\
        &\qquad\cdot
        \prod_{k=2}^{n}
        \Bigg[
        \mathcal{S}(
        \mathbf{r_{i_k}}, \mathbf{v_{i_k}}, t_{i_k}
        \mid
        \mathbf{r_{r_{k-1}}},\mathbf{v_{r_{k-1}}},t_{r_{k-1}})
        \,
        \\
        &\qquad\cdot \, f_{\varphi}(\mathbf{r_{i_k}}, \mathbf{v_{i_k}}, t_{i_k}) \,\mathcal{K_K}(
        \mathbf{r_{i_k}}\rightarrow\mathbf{r_{r_k}},
        \mathbf{v_{i_{k}}}\rightarrow\mathbf{v_{r_k}},
        t_{i_k}\rightarrow t_{r_k})
        \Bigg]
        \\
        &\qquad\cdot
        \dif\mathbf{v_{r_1}}\cdots\dif\mathbf{v_{r_{n-1}}}
        \dif\mathbf{v_{i_2}}\cdots\dif\mathbf{v_{i_{n-1}}}
        \dif\mathbf{r_{r_1}}\cdots\dif\mathbf{r_{r_{n-1}}}
        \dif\mathbf{r_{i_2}}\cdots\dif\mathbf{r_{i_{n-1}}}
        \\
        &\qquad\cdot
        \dif t_{r_1}\cdots\dif t_{r_{n-1}}
        \dif t_{i_2}\cdots\dif t_{i_{n-1}}
        \;
        \mathcal{S}(\mathbf{r_r},\mathbf{v_r},t_r)
        \Bigg\}.
    \end{aligned}
}
\end{equation}
\end{definition}

\begin{lemma} \label{lemma:reciprocity_multi_refl}
     Let $\mathcal{K_{MK}} : \mathbb{R}^3 \times \mathbb{R}^3 \times \mathbb{R}
\times \mathbb{R}^3 \times \mathbb{R}^3 \times \mathbb{R}
\rightarrow [0, \infty)$, $\mathcal{K_{MK}} = \mathcal{K_{MK}}\left(\mathbf{r_i} \rightarrow \mathbf{r_r}, \mathbf{v_i} \rightarrow \mathbf{v_r}, t_i \rightarrow t_r \right)$ be a Lebesgue-integrable function, defined in the global reference frame from \cref{def:global_frame} and associated with a surface $\Psi$ in the sense of \cref{def:surface_general}, whose expression is
\begin{equation}
        \begin{aligned}
        \mathcal{K_{MK}}(
        \mathbf{r_i} \rightarrow \mathbf{r_r},
        \mathbf{v_i} \rightarrow \mathbf{v_r},
        t_i \rightarrow t_r)
        &=
        \sum_{n=1}^{\infty}
        \Bigg\{
        \underset{\mathbf{v_{r_k}},\mathbf{v_{i_k}}}{\int\cdots\int}
        \underset{\mathbf{r_{r_k}},\mathbf{r_{i_k}}}{\int\cdots\int}
        \underset{t_{r_k},t_{i_k}}{\int\cdots\int}
        \mathcal{K_K}(
        \mathbf{r_i}\rightarrow\mathbf{r_{r_1}},
        \mathbf{v_i}\rightarrow\mathbf{v_{r_1}},
        t_i\rightarrow t_{r_1})
        \\
        &\qquad\cdot
        \prod_{k=2}^{n}
        \Bigg[
        \mathcal{S}(
        \mathbf{r_{i_k}}, \mathbf{v_{i_k}}, t_{i_k}
        \mid
        \mathbf{r_{r_{k-1}}},\mathbf{v_{r_{k-1}}},t_{r_{k-1}})
        \,
        \\
        &\qquad\cdot \, f_{\varphi}(\mathbf{r_{i_k}}, \mathbf{v_{i_k}}, t_{i_k}) \,\mathcal{K_K}(
        \mathbf{r_{i_k}}\rightarrow\mathbf{r_{r_k}},
        \mathbf{v_{i_{k}}}\rightarrow\mathbf{v_{r_k}},
        t_{i_k}\rightarrow t_{r_k})
        \Bigg]
        \\
        &\qquad\cdot
        \dif\mathbf{v_{r_1}}\cdots\dif\mathbf{v_{r_{n-1}}}
        \dif\mathbf{v_{i_2}}\cdots\dif\mathbf{v_{i_{n-1}}}
        \dif\mathbf{r_{r_1}}\cdots\dif\mathbf{r_{r_{n-1}}}
        \dif\mathbf{r_{i_2}}\cdots\dif\mathbf{r_{i_{n-1}}}
        \\
        &\qquad\cdot
        \dif t_{r_1}\cdots\dif t_{r_{n-1}}
        \dif t_{i_2}\cdots\dif t_{i_{n-1}}
        \;
        \mathcal{S}(\mathbf{r_r},\mathbf{v_r},t_r)
        \Bigg\},
    \end{aligned}
\end{equation}
as given by \cref{def:multi_reflection_scattering_kernel}, where $\mathcal{K_K}: \mathbb{R}^3 \times \mathbb{R}^3 \times \mathbb{R}
\times \mathbb{R}^3 \times \mathbb{R}^3 \times \mathbb{R} \rightarrow\mathbb{R}$, written as $\mathcal{K_K}= \mathcal{K_K}(\mathbf{r_{i}} \rightarrow \mathbf{r_{r}}, \mathbf{v_{i}} \rightarrow \mathbf{v_{r}}, t_i \rightarrow t_r)$, is a Lebesgue-integrable single-reflection global scattering kernel defined in the same frame. Here, $\mathbf{r_{i}}$, $\mathbf{r_{r}}$, $\mathbf{v_{i}}$, and $\mathbf{v_{r}}$ are $3\times 1$ real vectors denoting the incident and reflected particle positions and velocities, respectively, $t_i$ and $t_r$ represent the corresponding entry and exit times, respectively. $\mathcal{S}(\mathbf{r_r},\mathbf{v_r},t_r)$ and $\mathcal{S}(\mathbf{r_i}, \mathbf{v_i}, t_i \, | \, \mathbf{r_r},\mathbf{v_r},t_r)$ denote the one and two-point shadowing
functions, respectively, defined in \cref{def:shadowing_single_scale} and associated with the
surface $\Psi$. Then $\mathcal{K_{MK}}$ satisfies reciprocity in the sense of \cref{def:reciprocity} if the local kernel $\mathcal{K_L}$ associated with $\mathcal{K_K}$ satisfies pointwise reciprocity and impermeability in the sense of \cref{def:pointwise_reciprocity,def:impermeability}, and the two-point shadowing function $\mathcal{S}$ satisfies
\begin{multline} \label{eq:shadowing_condition_1}
    \mathcal{S}\!\left(
    \mathbf{r_i}, \mathbf{v_i}, t_i
    \mid
    \mathbf{r_r}, \mathbf{v_r}, t_r
    \right) 
    =
    \mathcal{S}\!\left(
    \mathbf{r_r}, -\mathbf{v_r}, -t_r
    \mid
    \mathbf{r_i}, -\mathbf{v_i}, -t_i
    \right),
    \forall\,(\mathbf{r_i},\mathbf{v_i},t_i,\mathbf{r_r},\mathbf{v_r},t_r)
    \in
    \mathbb{R}^3 \times \mathbb{R}^3 \times \mathbb{R}
    \times
    \mathbb{R}^3 \times \mathbb{R}^3 \times \mathbb{R}.
\end{multline}
\end{lemma}
\begin{proof}
    We first show that
    \begin{equation}
        f_{\varphi}(\mathbf{r_i},\mathbf{v_i},t_i)\,
        f_{\varphi_0}(\mathbf{r_r},-\mathbf{v_r},-t_r)
        =
        f_{\varphi_0}(\mathbf{r_i},\mathbf{v_i},t_i)\,
        f_{\varphi}(\mathbf{r_r},-\mathbf{v_r},-t_r),
    \end{equation}
    for all $(\mathbf{r_i}, \mathbf{v_i}, t_i, \mathbf{r_r}, \mathbf{v_r}, t_r)
        \in
        \mathbb{R}^3 \times \mathbb{R}^3 \times \mathbb{R}
        \times
        \mathbb{R}^3 \times \mathbb{R}^3 \times \mathbb{R}$. Expanding the left-hand side and using the facts that $\mathbf{v_i}=\mathbf{v_r}$ for any particle travelling on $\Psi$ and that by \cref{def:MaxwellBoltzmann}, $f_0(-\mathbf{v})=f_0(\mathbf{v})$ for all $\mathbf{v}\in\mathbb{R}^3$, we obtain
    \begin{equation}
        \begin{aligned}
        f_{\varphi}(\mathbf{r_i},\mathbf{v_i},t_i)\,
        f_{\varphi_0}(\mathbf{r_r},-\mathbf{v_r},-t_r)
        &=
        \delta(\mathbf{v_i}-\mathbf{v_r})\,
        f_{\varphi}(\mathbf{r_i},\mathbf{v_i},t_i)\,f_{\varphi}(\mathbf{r_r},-\mathbf{v_r},-t_r)\,
        f_0(-\mathbf{v_r})
        \\
        &=
        \delta(\mathbf{v_i}-\mathbf{v_r})\,
        f_{\varphi}(\mathbf{r_i},\mathbf{v_i},t_i)\,f_{\varphi}(\mathbf{r_r},-\mathbf{v_r},-t_r)\,
        f_0(\mathbf{v_i})
        \\
        &=
        f_{\varphi_0}(\mathbf{r_i},\mathbf{v_i},t_i)\,
        f_{\varphi}(\mathbf{r_r},-\mathbf{v_r},-t_r).
        \end{aligned}
    \end{equation}
    Now, as in \cref{lemma:global_reciprocity}, we substitute the multi-reflection kernel into the left-hand side of \cref{eq:reciprocity_general} and obtain
    \begin{equation}
        \begin{aligned}
        &\mathcal{K_{MK}}\!\left(
        \mathbf{r_i}\rightarrow\mathbf{r_r},
        \mathbf{v_i}\rightarrow\mathbf{v_r},
        t_i\rightarrow t_r
        \right)
         \mathcal{Q}( \mathbf{v_i} \, | \, \mathbf{n_G})
        f_{\varphi_0}(\mathbf{r_i}, \mathbf{v_i},t_i)\,
        \mathcal{M}(\mathbf{r_i},\mathbf{v_i},t_i)
        \\[6pt]
        &\quad=
        \sum_{n=1}^{\infty}
        \Bigg\{
        \mathcal{M}(\mathbf{r_i},\mathbf{v_i},t_i)
        \!\!\underset{\mathbf{v_{r_k}},\mathbf{v_{i_k}}}{\int\!\cdots\!\int}
        \!\!\underset{\mathbf{r_{r_k}},\mathbf{r_{i_k}}}{\int\!\cdots\!\int}
        \!\!\underset{t_{r_k},t_{i_k}}{\int\!\cdots\!\int}
        f_{\varphi_0}(\mathbf{r_i}, \mathbf{v_i},t_i) \, \mathcal{Q}( \mathbf{v_i} \, | \, \mathbf{n_G})
        \mathcal{K_K}\!\left(
        \mathbf{r_i}\rightarrow\mathbf{r_{r_1}},
        \mathbf{v_i}\rightarrow\mathbf{v_{r_1}},
        t_i\rightarrow t_{r_1}
        \right)
        \\
        &\qquad\cdot
        \prod_{k=2}^{n}
        \Bigg[
        \mathcal{S}\!\left(
        \mathbf{r_{i_k}},\mathbf{v_{i_k}},t_{i_k}
        \,\middle|\,
        \mathbf{r_{r_{k-1}}},\mathbf{v_{r_{k-1}}},t_{r_{k-1}}
        \right)
        \,f_{\varphi}(\mathbf{r_{i_k}}, \mathbf{v_{i_k}}, t_{i_k}) \mathcal{K_K}\!\left(
        \mathbf{r_{i_k}}\rightarrow\mathbf{r_{r_k}},
        \mathbf{v_{i_{k}}}\rightarrow\mathbf{v_{r_k}},
        t_{i_k}\rightarrow t_{r_k}
        \right)
        \Bigg]
        \\
        &\qquad\cdot
        \dif\mathbf{v_{r_1}}\cdots\dif\mathbf{v_{r_{n-1}}}\,
        \dif\mathbf{v_{i_2}}\cdots\dif\mathbf{v_{i_{n-1}}}\,
        \dif\mathbf{r_{r_1}}\cdots\dif\mathbf{r_{r_{n-1}}}\,
        \dif\mathbf{r_{i_2}}\cdots\dif\mathbf{r_{i_{n-1}}}\, \dif t_{r_1}\cdots\dif t_{r_{n-1}}\,
        \dif t_{i_2}\cdots\dif t_{i_{n-1}}\;
        \mathcal{S}(\mathbf{r_r},\mathbf{v_r},t_r)
        \Bigg\}.
        \end{aligned}
    \end{equation}
    By \cref{lemma:global_reciprocity}, we employ the pointwise reciprocity property onto the first single-reflection kernel $\mathcal{K_K}(\mathbf{r_{i}} \rightarrow \mathbf{r_{r_1}}, \mathbf{v_{i}} \rightarrow \mathbf{v_{r_1}}, t_i \rightarrow t_{r_1})$, while remembering that $\mathcal{M}(\mathbf{r_{i}}, \mathbf{v_{i}}, t_i) = \mathcal{S}(\mathbf{r_{i}}, -\mathbf{v_{i}}, t_i)$ and $\mathcal{Q}(\mathbf{v}\, | \, \mathbf{n}) = \mathcal{Q}( -\mathbf{v} \, | \, \mathbf{n})$ by \cref{def:shadowing_single_scale,def:visible_flux}, and get
    \begin{equation}
        \begin{aligned}
        &\mathcal{K_{MK}}\!\left(
        \mathbf{r_i}\rightarrow\mathbf{r_r},
        \mathbf{v_i}\rightarrow\mathbf{v_r},
        t_i\rightarrow t_r
        \right)
         \mathcal{Q}( \mathbf{v_i} \, | \, \mathbf{n_G})
        f_{\varphi_0}(\mathbf{r_i}, \mathbf{v_i},t_i)\,
        \mathcal{M}(\mathbf{r_i},\mathbf{v_i},t_i) =
        \\[6pt]
        &\quad=
        \sum_{n=1}^{\infty}
        \Bigg\{
        \mathcal{S}(\mathbf{r_i},-\mathbf{v_i},-t_i)
        \!\!\underset{\mathbf{v_{r_k}},\mathbf{v_{i_k}}}{\int\!\cdots\!\int}
        \!\!\underset{\mathbf{r_{r_k}},\mathbf{r_{i_k}}}{\int\!\cdots\!\int}
        \!\!\underset{t_{r_k},t_{i_k}}{\int\!\cdots\!\int} \mathcal{K_K}\!\left(
        \mathbf{r_{r_1}}\rightarrow\mathbf{r_i},
        -\mathbf{v_{r_1}}\rightarrow-\mathbf{v_i},
        -t_{r_1}\rightarrow -t_i
        \right)
        f_{\varphi_0}(\mathbf{r_{r_1}},-\mathbf{v_{r_1}},-t_{r_1}) \mathcal{Q}( \mathbf{v_{r_1}} \, | \, \mathbf{n_G})
        \\
        &\qquad\cdot
        \prod_{k=2}^{n}
        \Bigg[
        \mathcal{S}\!\left(
        \mathbf{r_{i_k}}, \mathbf{v_{i_k}},t_{i_k}
        \,\middle|\,
        \mathbf{r_{r_{k-1}}},\mathbf{v_{r_{k-1}}},t_{r_{k-1}}
        \right)
        \,f_{\varphi}(\mathbf{r_{i_k}}, \mathbf{v_{i_k}}, t_{i_k}) \mathcal{K_K}\!\left(
        \mathbf{r_{i_k}}\rightarrow\mathbf{r_{r_k}},
        \mathbf{v_{i_{k}}}\rightarrow\mathbf{v_{r_k}},
        t_{i_k}\rightarrow t_{r_k}
        \right)
        \Bigg]
        \\
        &\qquad\cdot
        \dif\mathbf{v_{r_1}}\cdots\dif\mathbf{v_{r_{n-1}}}\,
        \dif\mathbf{v_{i_2}}\cdots\dif\mathbf{v_{i_{n-1}}}\,
        \dif\mathbf{r_{r_1}}\cdots\dif\mathbf{r_{r_{n-1}}}\,
        \dif\mathbf{r_{i_2}}\cdots\dif\mathbf{r_{i_{n-1}}}\, \dif t_{r_1}\cdots\dif t_{r_{n-1}}\,
        \dif t_{i_2}\cdots\dif t_{i_{n-1}}\;
        \mathcal{S}(\mathbf{r_r},\mathbf{v_r},t_r)
        \Bigg\}.
        \end{aligned}
    \end{equation}
    We now use the previously derived property, $f_{\varphi}(\mathbf{r_{i_2}},\mathbf{v_{i_2}},t_{i_2}) \, f_{\varphi_0}(\mathbf{r_{r_1}},-\mathbf{v_{r_1}},-t_{r_1}) = f_{\varphi_0}(\mathbf{r_{i_2}},\mathbf{v_{i_2}},t_{i_2}) \, f_{\varphi}(\mathbf{r_{r_1}},-\mathbf{v_{r_1}},-t_{r_1})$ to obtain
    \begin{equation}
        \begin{aligned}
        &\mathcal{K_{MK}}\!\left(
        \mathbf{r_i}\rightarrow\mathbf{r_r},
        \mathbf{v_i}\rightarrow\mathbf{v_r},
        t_i\rightarrow t_r
        \right)
         \mathcal{Q}( \mathbf{v_i} \, | \, \mathbf{n_G})
        f_{\varphi_0}(\mathbf{r_i}, \mathbf{v_i},t_i)\,
        \mathcal{M}(\mathbf{r_i},\mathbf{v_i},t_i) =
        \\[6pt]
        &\quad=
        \sum_{n=1}^{\infty}
        \Bigg\{
        \mathcal{S}(\mathbf{r_i},-\mathbf{v_i},-t_i)
        \!\!\underset{\mathbf{v_{r_k}},\mathbf{v_{i_k}}}{\int\!\cdots\!\int}
        \!\!\underset{\mathbf{r_{r_k}},\mathbf{r_{i_k}}}{\int\!\cdots\!\int}
        \!\!\underset{t_{r_k},t_{i_k}}{\int\!\cdots\!\int} \mathcal{K_K}\!\left(
        \mathbf{r_{r_1}}\rightarrow\mathbf{r_i},
        -\mathbf{v_{r_1}}\rightarrow-\mathbf{v_i},
        -t_{r_1}\rightarrow -t_i
        \right) \, \mathcal{S}\!\left(
        \mathbf{r_{r_1}},-\mathbf{v_{r_1}},-t_{r_1}
        \,\middle|\,
        \mathbf{r_{i_2}},-\mathbf{v_{i_2}},-t_{i_2}
        \right)
        \\
        &\qquad\cdot
        \,f_{\varphi}(\mathbf{r_{r_1}},-\mathbf{v_{r_1}},-t_{r_1}) \, f_{\varphi_0}(\mathbf{r_{i_2}},\mathbf{v_{i_2}},t_{i_2}) \,  \mathcal{K_K}\!\left(
        \mathbf{r_{i_2}}\rightarrow\mathbf{r_{r_2}},
        \mathbf{v_{i_2}}\rightarrow\mathbf{v_{r_2}},
        t_{i_2}\rightarrow t_{r_2}
        \right)\,\mathcal{Q}( \mathbf{v_{r_2}} \, | \, \mathbf{n_G})
        \\
        &\qquad\cdot
        \prod_{k=3}^{n}
        \Bigg[
        \mathcal{S}\!\left(
        \mathbf{r_{i_k}},\mathbf{v_{i_k}},t_{i_k}
        \,\middle|\,
        \mathbf{r_{r_{k-1}}},\mathbf{v_{r_{k-1}}},t_{r_{k-1}}
        \right)
        \,f_{\varphi}(\mathbf{r_{i_k}},\mathbf{v_{i_k}},t_{i_k}) \mathcal{K_K}\!\left(
        \mathbf{r_{i_k}}\rightarrow\mathbf{r_{r_k}},
        \mathbf{v_{i_{k}}}\rightarrow\mathbf{v_{r_k}},
        t_{i_k}\rightarrow t_{r_k}
        \right)
        \Bigg]
        \\
        &\qquad\cdot
        \dif\mathbf{v_{r_1}}\cdots\dif\mathbf{v_{r_{n-1}}}\,
        \dif\mathbf{v_{i_2}}\cdots\dif\mathbf{v_{i_{n-1}}}\,
        \dif\mathbf{r_{r_1}}\cdots\dif\mathbf{r_{r_{n-1}}}\,
        \dif\mathbf{r_{i_2}}\cdots\dif\mathbf{r_{i_{n-1}}}\, \dif t_{r_1}\cdots\dif t_{r_{n-1}}\,
        \dif t_{i_2}\cdots\dif t_{i_{n-1}}\;
        \mathcal{S}(\mathbf{r_r},\mathbf{v_r},t_r)
        \Bigg\}.
        \end{aligned}
    \end{equation}
Next,  we employ pointwise reciprocity on the second single-reflection kernel $\mathcal{K_K}(\mathbf{r_{i_2}} \rightarrow \mathbf{r_{r_2}}, \mathbf{v_{i_2}} \rightarrow \mathbf{v_{r_2}}, t_{i_2} \rightarrow t_{r_2})$, which yields
    \begin{equation}
        \begin{aligned}
        &\mathcal{K_{MK}}\!\left(
        \mathbf{r_i}\rightarrow\mathbf{r_r},
        \mathbf{v_i}\rightarrow\mathbf{v_r},
        t_i\rightarrow t_r
        \right)
         \mathcal{Q}( \mathbf{v_i} \, | \, \mathbf{n_G})
        f_{\varphi_0}(\mathbf{r_i}, \mathbf{v_i},t_i)\,
        \mathcal{M}(\mathbf{r_i},\mathbf{v_i},t_i) =
        \\[6pt]
        &\quad=
        \sum_{n=1}^{\infty}
        \Bigg\{
        \mathcal{S}(\mathbf{r_i},-\mathbf{v_i},-t_i)
        \!\!\underset{\mathbf{v_{r_k}},\mathbf{v_{i_k}}}{\int\!\cdots\!\int}
        \!\!\underset{\mathbf{r_{r_k}},\mathbf{r_{i_k}}}{\int\!\cdots\!\int}
        \!\!\underset{t_{r_k},t_{i_k}}{\int\!\cdots\!\int} \mathcal{K_K}\!\left(
        \mathbf{r_{r_1}}\rightarrow\mathbf{r_i},
        -\mathbf{v_{r_1}}\rightarrow-\mathbf{v_i},
        -t_{r_1}\rightarrow -t_i
        \right) \mathcal{S}\!\left(
        \mathbf{r_{r_1}},-\mathbf{v_{r_1}},-t_{r_1}
        \,\middle|\,
        \mathbf{r_{i_2}},-\mathbf{v_{i_2}},-t_{i_2}
        \right)
        \\
        &\qquad\cdot
        \,f_{\varphi}(\mathbf{r_{r_1}},-\mathbf{v_{r_1}},-t_{r_1}) \mathcal{K_K}\!\left(
        \mathbf{r_{r_2}}\rightarrow\mathbf{r_{i_2}},
        -\mathbf{v_{r_2}}\rightarrow-\mathbf{v_{i_2}},
        -t_{r_2}\rightarrow -t_{i_2}
        \right)
        f_{\varphi_0}(\mathbf{r_{r_2}},-\mathbf{v_{r_2}},t_{r_2}) \,\mathcal{Q}( \mathbf{v_{r_2}} \, | \, \mathbf{n_G})
        \\
        &\qquad\cdot
        \prod_{k=3}^{n}
        \Bigg[
        \mathcal{S}\!\left(
        \mathbf{r_{i_k}},\mathbf{v_{i_k}},t_{i_k}
        \,\middle|\,
        \mathbf{r_{r_{k-1}}},\mathbf{v_{r_{k-1}}},t_{r_{k-1}}
        \right)
        \,f_{\varphi}(\mathbf{r_{i_k}},\mathbf{v_{i_k}},t_{i_k}) \mathcal{K_K}\!\left(
        \mathbf{r_{i_k}}\rightarrow\mathbf{r_{r_k}},
        \mathbf{v_{i_{k}}}\rightarrow\mathbf{v_{r_k}},
        t_{i_k}\rightarrow t_{r_k}
        \right)
        \Bigg]
        \\
        &\qquad\cdot
        \dif\mathbf{v_{r_1}}\cdots\dif\mathbf{v_{r_{n-1}}}\,
        \dif\mathbf{v_{i_2}}\cdots\dif\mathbf{v_{i_{n-1}}}\,
        \dif\mathbf{r_{r_1}}\cdots\dif\mathbf{r_{r_{n-1}}}\,
        \dif\mathbf{r_{i_2}}\cdots\dif\mathbf{r_{i_{n-1}}}\, \dif t_{r_1}\cdots\dif t_{r_{n-1}}\,
        \dif t_{i_2}\cdots\dif t_{i_{n-1}}\;
        \mathcal{S}(\mathbf{r_r},\mathbf{v_r},t_r)
        \Bigg\},
        \end{aligned}
    \end{equation}
    where we have used the properties: $ \mathcal{S}\!\left(
    \mathbf{r_{i_2}}, \mathbf{v_{i_2}}, t_{i_2}
    \mid
    \mathbf{r_{r_1}}, \mathbf{v_{r_1}}, t_{r_1}
    \right)
    =
    \mathcal{S}\!\left(
    \mathbf{r_{r_1}}, -\mathbf{v_{r_1}}, -t_{r_1}
    \mid
    \mathbf{r_{i_2}}, -\mathbf{v_{i_2}}, -t_{i_2}
    \right) $, $\mathcal{Q}( -\mathbf{v_{r_1}} \, | \, \mathbf{n_G}) = \mathcal{Q}( \mathbf{v_{r_1}} \, | \, \mathbf{n_G})$, and $\mathbf{v_{r_1}} = \mathbf{v_{i_2}}$. We now repeat the previous step for all remaining terms in the kernel products, and use the property that $\mathcal{S}(\mathbf{r_{r}}, \mathbf{v_{r}}, t_r) = \mathcal{M}(\mathbf{r_{r}}, -\mathbf{v_{r}}, -t_r)$ to obtain
    \begin{equation}
        \begin{aligned}
        &\mathcal{K_{MK}}\!\left(
        \mathbf{r_i}\rightarrow\mathbf{r_r},
        \mathbf{v_i}\rightarrow\mathbf{v_r},
        t_i\rightarrow t_r
        \right)
         \mathcal{Q}( \mathbf{v_i} \, | \, \mathbf{n_G})
        f_{\varphi_0}(\mathbf{r_i}, \mathbf{v_i},t_i)\,
        \mathcal{M}(\mathbf{r_i},\mathbf{v_i},t_i) =
        \\[6pt]
        &\quad=
        \sum_{n=1}^{\infty}
        \Bigg\{
        \mathcal{S}(\mathbf{r_i},-\mathbf{v_i},-t_i)
        \!\!\underset{\mathbf{v_{r_k}},\mathbf{v_{i_k}}}{\int\!\cdots\!\int}
        \!\!\underset{\mathbf{r_{r_k}},\mathbf{r_{i_k}}}{\int\!\cdots\!\int}
        \!\!\underset{t_{r_k},t_{i_k}}{\int\!\cdots\!\int} \prod_{k=1}^{n-1}
        \Bigg[
        \mathcal{S}\!\left(
        \mathbf{r_{r_k}},-\mathbf{v_{r_k}},-t_{r_k}
        \,\middle|\,
        \mathbf{r_{i_{k+1}}},-\mathbf{v_{i_{k+1}}},-t_{i_{k+1}}
        \right)
        \\
        &\qquad\cdot
        f_{\varphi}(\mathbf{r_{r_k}},-\mathbf{v_{r_k}},-t_{r_k}) \mathcal{K_K}\!\left(
        \mathbf{r_{r_k}}\rightarrow\mathbf{r_{i_k}},
        -\mathbf{v_{r_{k}}}\rightarrow-\mathbf{v_{i_k}},
        -t_{r_k}\rightarrow -t_{i_k}
        \right)
        \Bigg]
        \\
        &\qquad\cdot
        \mathcal{K_K}\!\left(
        \mathbf{r_r}\rightarrow\mathbf{r_{i_n}},
        -\mathbf{v_r}\rightarrow-\mathbf{v_{r_{n-1}}},
        -t_r\rightarrow -t_{i_n}
        \right)
        f_{\varphi_0}(\mathbf{r_r},-\mathbf{v_r},-t_r) \, \mathcal{Q}( -\mathbf{v_{r}} \, | \, \mathbf{n_G})
        \\
        &\qquad\cdot
        \dif\mathbf{v_{r_1}}\cdots\dif\mathbf{v_{r_{n-1}}}\,
        \dif\mathbf{v_{i_2}}\cdots\dif\mathbf{v_{i_{n-1}}}\,
        \dif\mathbf{r_{r_1}}\cdots\dif\mathbf{r_{r_{n-1}}}\,
        \dif\mathbf{r_{i_2}}\cdots\dif\mathbf{r_{i_{n-1}}}\, \dif t_{r_1}\cdots\dif t_{r_{n-1}}\,
        \dif t_{i_2}\cdots\dif t_{i_{n-1}}\;
        \mathcal{M}(\mathbf{r_r},-\mathbf{v_r},-t_r)
        \Bigg\}.
        \end{aligned}
    \end{equation}
    Finally, rearranging terms for each term in the sum and taking the flux term $\mathcal{Q}( -\mathbf{v_{r}} \, | \, \mathbf{n_G})$, and the reflected state PDF $f_{\varphi}(\mathbf{r_r},-\mathbf{v_r},t_r)$ out of the summation, we get
    \begin{equation}
        \begin{aligned}
        &\mathcal{K_{MK}}\!\left(
        \mathbf{r_i}\rightarrow\mathbf{r_r},
        \mathbf{v_i}\rightarrow\mathbf{v_r},
        t_i\rightarrow t_r
        \right)
         \mathcal{Q}( \mathbf{v_i} \, | \, \mathbf{n_G})
        f_{\varphi}(\mathbf{r_i}, \mathbf{v_i},t_i)\,
        \mathcal{M}(\mathbf{r_i},\mathbf{v_i},t_i) =
        \\[6pt]
        &\quad=
        \sum_{n=1}^{\infty}
        \Bigg\{
        \mathcal{M}(\mathbf{r_r},-\mathbf{v_r},-t_r)
        \!\!\underset{\mathbf{v_{r_k}},\mathbf{v_{i_k}}}{\int\!\cdots\!\int}
        \!\!\underset{\mathbf{r_{r_k}},\mathbf{r_{i_k}}}{\int\!\cdots\!\int}
        \!\!\underset{t_{r_k},t_{i_k}}{\int\!\cdots\!\int} \mathcal{K_K}\!\left(
        \mathbf{r_r}\rightarrow\mathbf{r_{i_n}},
        -\mathbf{v_r}\rightarrow-\mathbf{v_{r_{n-1}}},
        -t_r\rightarrow -t_{i_n}
        \right)
        \\
        &\qquad\cdot
        \prod_{k=n-1}^{1}
        \Bigg[
        \mathcal{S}\!\left(
        \mathbf{r_{r_k}},-\mathbf{v_{r_k}},-t_{r_k}
        \,\middle|\,
        \mathbf{r_{i_{k+1}}},-\mathbf{v_{i_{k+1}}},-t_{i_{k+1}}
        \right)
        f_{\varphi}(\mathbf{r_{r_k}},-\mathbf{v_{r_k}},-t_{r_k}) 
        \\
        &\qquad\cdot
        \mathcal{K_K}\!\left(
        \mathbf{r_{r_k}}\rightarrow\mathbf{r_{i_k}},
        -\mathbf{v_{r_{k}}}\rightarrow-\mathbf{v_{i_k}},
        -t_{r_k}\rightarrow -t_{i_k}
        \right)
        \Bigg]
        \dif\mathbf{v_{r_1}}\cdots\dif\mathbf{v_{r_{n-1}}}\,
        \dif\mathbf{v_{i_2}}\cdots\dif\mathbf{v_{i_{n-1}}}\,
        \dif\mathbf{r_{r_1}}\cdots\dif\mathbf{r_{r_{n-1}}}\,
        \dif\mathbf{r_{i_2}}\cdots\dif\mathbf{r_{i_{n-1}}}\,
        \\[6pt]
        &\qquad\cdot
        \dif t_{r_1}\cdots\dif t_{r_{n-1}}\,
        \dif t_{i_2}\cdots\dif t_{i_{n-1}}\;
        \mathcal{S}(\mathbf{r_i},-\mathbf{v_i},-t_i)
        \Bigg\} \,
        f_{\varphi_0}(\mathbf{r_r},-\mathbf{v_r},-t_r) \, \mathcal{Q}( -\mathbf{v_{r}} \, | \, \mathbf{n_G})
        \\[6pt]
        &\quad=
        \mathcal{K_{MK}}\!\left(
        \mathbf{r_r}\rightarrow\mathbf{r_i},
        -\mathbf{v_r}\rightarrow-\mathbf{v_i},
        -t_r\rightarrow -t_i
        \right)
       \mathcal{Q}( -\mathbf{v_{r}} \, | \, \mathbf{n_G})\,
        f_{\varphi_0}(\mathbf{r_r},-\mathbf{v_r},-t_r)\,
        \mathcal{M}(\mathbf{r_r},-\mathbf{v_r},-t_r).
        \end{aligned}
    \end{equation}
\end{proof}

\begin{lemma} \label{lemma:normalisation_multi_refl}
     Let $\mathcal{K_{MK}} : \mathbb{R}^3 \times \mathbb{R}^3 \times \mathbb{R}
\times \mathbb{R}^3 \times \mathbb{R}^3 \times \mathbb{R}
\rightarrow [0, \infty)$, $\mathcal{K_{MK}} = \mathcal{K_{MK}}\left(\mathbf{r_i} \rightarrow \mathbf{r_r}, \mathbf{v_i} \rightarrow \mathbf{v_r}, t_i \rightarrow t_r \right)$ be a Lebesgue-integrable function, defined in the global reference frame from \cref{def:global_frame} and associated with a surface $\Psi$ in the sense of \cref{def:surface_general}, whose expression is
\begin{equation}
        \begin{aligned}
        \mathcal{K_{MK}}(
        \mathbf{r_i} \rightarrow \mathbf{r_r},
        \mathbf{v_i} \rightarrow \mathbf{v_r},
        t_i \rightarrow t_r)
        &=
        \sum_{n=1}^{\infty}
        \Bigg\{
        \underset{\mathbf{v_{r_k}},\mathbf{v_{i_k}}}{\int\cdots\int}
        \underset{\mathbf{r_{r_k}},\mathbf{r_{i_k}}}{\int\cdots\int}
        \underset{t_{r_k},t_{i_k}}{\int\cdots\int}
        \mathcal{K_K}(
        \mathbf{r_i}\rightarrow\mathbf{r_{r_1}},
        \mathbf{v_i}\rightarrow\mathbf{v_{r_1}},
        t_i\rightarrow t_{r_1})
        \\
        &\qquad\cdot
        \prod_{k=2}^{n}
        \Bigg[
        \mathcal{S}(
        \mathbf{r_{i_k}}, \mathbf{v_{i_k}}, t_{i_k}
        \mid
        \mathbf{r_{r_{k-1}}},\mathbf{v_{r_{k-1}}},t_{r_{k-1}})
        \,
        \\
        &\qquad\cdot \, f_{\varphi}(\mathbf{r_{i_k}}, \mathbf{v_{i_k}}, t_{i_k}) \,\mathcal{K_K}(
        \mathbf{r_{i_k}}\rightarrow\mathbf{r_{r_k}},
        \mathbf{v_{i_{k}}}\rightarrow\mathbf{v_{r_k}},
        t_{i_k}\rightarrow t_{r_k})
        \Bigg]
        \\
        &\qquad\cdot
        \dif\mathbf{v_{r_1}}\cdots\dif\mathbf{v_{r_{n-1}}}
        \dif\mathbf{v_{i_2}}\cdots\dif\mathbf{v_{i_{n-1}}}
        \dif\mathbf{r_{r_1}}\cdots\dif\mathbf{r_{r_{n-1}}}
        \dif\mathbf{r_{i_2}}\cdots\dif\mathbf{r_{i_{n-1}}}
        \\
        &\qquad\cdot
        \dif t_{r_1}\cdots\dif t_{r_{n-1}}
        \dif t_{i_2}\cdots\dif t_{i_{n-1}}
        \;
        \mathcal{S}(\mathbf{r_r},\mathbf{v_r},t_r)
        \Bigg\},
    \end{aligned}
\end{equation}
as given by \cref{def:multi_reflection_scattering_kernel}, where $\mathcal{K_K}: \mathbb{R}^3 \times \mathbb{R}^3 \times \mathbb{R}
\times \mathbb{R}^3 \times \mathbb{R}^3 \times \mathbb{R} \rightarrow\mathbb{R}$, written as $\mathcal{K_K}= \mathcal{K_K}(\mathbf{r_{i}} \rightarrow \mathbf{r_{r}}, \mathbf{v_{i}} \rightarrow \mathbf{v_{r}}, t_i \rightarrow t_r)$, is a Lebesgue-integrable single-reflection global scattering kernel defined in the same frame. Here, $\mathbf{r_{i}}$, $\mathbf{r_{r}}$, $\mathbf{v_{i}}$, and $\mathbf{v_{r}}$ are $3\times 1$ real vectors denoting the incident and reflected particle positions and velocities, respectively, $t_i$ and $t_r$ represent the corresponding entry and exit times, respectively. $\mathcal{S}(\mathbf{r_r},\mathbf{v_r},t_r)$ and $\mathcal{S}(\mathbf{r_i},\mathbf{v_i}, t_i \, | \, \mathbf{r_r},\mathbf{v_r},t_r)$ denote the one and two-point shadowing functions, respectively, defined in \cref{def:shadowing_single_scale} and associated with the
surface $\Psi$. Then $\mathcal{K_{MK}}$ satisfies normalisation in the sense of \cref{def:normalisation} if the local kernel $\mathcal{K_L}$ associated with $\mathcal{K_K}$ satisfies normalisation, and the two-point shadowing function $\mathcal{S}$ satisfies
\begin{equation} \label{eq:shadowing_condition_2}
    \int_{0}^{\infty}\mathcal{S}\!\left(
    \mathbf{r_i},\mathbf{v_i}, t_i
    \mid
    \mathbf{r_r}, \mathbf{v_r}, t_r
    \right) f_{\varphi}(\mathbf{r_i},\mathbf{v_i}, t_i ) \, \dif \tau
    = 1 - \mathcal{S}(\mathbf{r_r}, \mathbf{v_r}, t_r),
    \quad
    \forall\,(\mathbf{r_r},\mathbf{v_r},t_r)
    \in
    \mathbb{R}^3 \times \mathbb{R}^3 \times \mathbb{R}, 
\end{equation}
where $\tau = t_i - t_r$. Furthermore, $\exists \, \varepsilon \in (0,1]$ such that, for every $(\mathbf{r_i}, \mathbf{v_i}, t_i) \in \mathbb{R}^3 \times \mathbb{R}^3 \times \mathbb{R}$ and $\mathcal{K_L}\left(\mathbf{r_{i_L}} \rightarrow \mathbf{r_{r_L}},\mathbf{v_{i_L}} \rightarrow \mathbf{v_{r_L}},t_i \rightarrow t_r\right) \in \mathcal{T}$,
\begin{multline} \label{eq:shadowing_single_scale_condition_3}
    \int_{\mathbb{R}^3}\int_{\mathbb{R}^3}\int_{\mathbb{R}} \int_{\mathbb{S}^2}
        \mathcal{K_L}(\mathbf{r_{i_L}} \rightarrow \mathbf{r_{r_L}}, \mathbf{v_{i_L}} \rightarrow \mathbf{v_{r_L}}, t_i \rightarrow t_r) 
         \frac{\mathcal{Q}(\mathbf{v_{i_L}}\, | \,  \mathbf{n_L})}
    {\mathcal{Q}( \mathbf{v_{i}} \, | \, \mathbf{n_G})} \cdot \\ \cdot \mathcal{S}_L\left(\mathbf{r_{r_L}}, \mathbf{v_{r_L}}, t_r \right) \,
        p_n(\mathbf{n_L} \, | \, \mathbf{n_G}, \mathbf{v_i})
        \, \dif \mathbf{n_L} \, \mathcal{S}\left(\mathbf{r_r}, \mathbf{v_r}, t_r\right) \, \dif t_r \dif \mathbf{v_r} \dif \mathbf{r_r} > \epsilon,
\end{multline}
where $\mathcal{S}_L\left(\mathbf{r_{r_L}}, \mathbf{v_{r_L}}, t_r \right)$ represents the one-point shadowing function associated with the local frame $L$. If $\Psi$ is locally smooth in the sense of \cref{def:local_smoothness}, then  $\mathcal{S}_L\left(\mathbf{r_{r_L}}, \mathbf{v_{r_L}}, t_r \right) = 1, \, \forall \, \left(\mathbf{r_{r_L}}, \mathbf{v_{r_L}}, t_r \right) \in \mathbb{R}^3 \times \mathbb{R}^3 \times \mathbb{R}$.
\end{lemma}
\begin{proof}
    For simplicity, we first introduce several new notations. We denote every reflected particle state after $k$ reflections by $\pmb{\varphi_{r_k}} = (\mathbf{r_{r_k}}, \mathbf{v_{r_k}}, t_{r_k})$, and every incident state by
    $\pmb{\varphi_{i_k}} = (\mathbf{r_{i_k}}, \mathbf{v_{i_k}}, t_{i_k})$.
    In what follows, for each fixed reflected state $\pmb{\varphi_r}$, the quantities $\mathbf{r_i}$ and $\mathbf{v_i}$ appearing in
    $\mathcal{S}(\mathbf{r_i},\mathbf{v_i},t_i \mid \pmb{\varphi_r})$ are understood to be those associated with the trajectory connecting $\pmb{\varphi_r}$ to the incident state at time $t_i$. We further define
    \begin{equation}
        \mathcal{O}(\pmb{\varphi_r})
        =
        \int_{0}^{\infty}
        \mathcal{S}\!\left(
        \mathbf{r_i},\mathbf{v_i},t_i
        \mid
        \pmb{\varphi_r}
        \right)
        f_{\varphi}(\mathbf{r_i},\mathbf{v_i},t_i)
        \, \dif \tau,
    \end{equation}
    so that, by \cref{eq:shadowing_condition_2},
    \begin{equation}
        \mathcal{O}(\pmb{\varphi_r}) + \mathcal{S}(\pmb{\varphi_r}) = 1,
        \qquad
        \forall\, \pmb{\varphi_r} \in \mathbb{R}^3 \times \mathbb{R}^3 \times \mathbb{R}.
    \end{equation}
    Then, we can rewrite \cref{eq:multi_reflection_kernel} as
    \begin{equation}
    \begin{aligned}
        \mathcal{K_{MK}}(
        \mathbf{r_i} \rightarrow \mathbf{r_r},
        \mathbf{v_i} \rightarrow \mathbf{v_r},
        t_i \rightarrow t_r)
        &=
        \sum_{n=1}^{\infty}
        \mathcal{K_{K}}_n(
        \mathbf{r_i} \rightarrow \mathbf{r_r},
        \mathbf{v_i} \rightarrow \mathbf{v_r},
        t_i \rightarrow t_r),
        \text{ with}
        \\
        \mathcal{K_{K}}_n(
        \mathbf{r_i} \rightarrow \mathbf{r_r},
        \mathbf{v_i} \rightarrow \mathbf{v_r},
        t_i \rightarrow t_r)
        &=
        \int_{\pmb{\varphi_{r_1}}}
        \mathcal{K_K}(\pmb{\varphi_i} \rightarrow \pmb{\varphi_{r_1}})
        \int_{\pmb{\varphi_{i_2}}}
        \mathcal{S}(\pmb{\varphi_{i_2}} \mid \pmb{\varphi_{r_1}})
        f_{\varphi}(\pmb{\varphi_{i_2}})
        \int_{\pmb{\varphi_{r_2}}}
        \mathcal{K_K}(\pmb{\varphi_{i_2}} \rightarrow \pmb{\varphi_{r_2}})
        \cdots
        \\
        &\qquad \cdot
        \int_{\pmb{\varphi_{i_n}}}
        \mathcal{S}(\pmb{\varphi_{i_n}} \mid \pmb{\varphi_{r_{n-1}}})
        f_{\varphi}(\pmb{\varphi_{i_n}})
        \mathcal{K_K}(\pmb{\varphi_{i_n}} \rightarrow \pmb{\varphi_r})
        \mathcal{S}(\pmb{\varphi_r})
        \, \dif \pmb{\varphi_{i_{n}}}
        \dif \pmb{\varphi_{r_{n-1}}}
        \dif \pmb{\varphi_{i_{n-1}}}
        \cdots
        \dif \pmb{\varphi_{i_{2}}}
        \dif \pmb{\varphi_{r_1}}.
    \end{aligned}
    \end{equation}
    Now, minding that the trajectory of a gas particle travelling on $\Psi$ is completely determined by its initial state and the time passed, i.e., $\pmb{\varphi_i} = \left(\mathbf{r_i}, \mathbf{v_i}, t_i\right) = \left( \mathbf{r_r} + \mathbf{v_r} \, \tau , \mathbf{v_r}, t_r + \tau\right)$, where $\tau = t_i - t_r$, we may switch variables from the incident states $\pmb{\varphi_i}$ to the times $\tau$, and get
        \begin{equation}
    \begin{aligned}
        \mathcal{K_{MK}}(
        \mathbf{r_i} \rightarrow \mathbf{r_r},
        \mathbf{v_i} \rightarrow \mathbf{v_r},
        t_i \rightarrow t_r)
        &=
        \sum_{n=1}^{\infty}
        \mathcal{K_{K}}_n(
        \mathbf{r_i} \rightarrow \mathbf{r_r},
        \mathbf{v_i} \rightarrow \mathbf{v_r},
        t_i \rightarrow t_r),
        \text{ with}
        \\
        \mathcal{K_{K}}_n(
        \mathbf{r_i} \rightarrow \mathbf{r_r},
        \mathbf{v_i} \rightarrow \mathbf{v_r},
        t_i \rightarrow t_r)
        &=
        \int_{\pmb{\varphi_{r_1}}}
        \mathcal{K_K}(\pmb{\varphi_i} \rightarrow \pmb{\varphi_{r_1}})
        \int_{0}^{\infty}
        \mathcal{S}(\pmb{\varphi_{i_2}} \mid \pmb{\varphi_{r_1}})
        f_{\varphi}(\pmb{\varphi_{i_2}})
        \int_{\pmb{\varphi_{r_2}}}
        \mathcal{K_K}(\pmb{\varphi_{i_2}} \rightarrow \pmb{\varphi_{r_2}})
        \cdots
        \\
        &\qquad \cdot
        \int_{0}^{\infty}
        \mathcal{S}(\pmb{\varphi_{i_n}} \mid \pmb{\varphi_{r_{n-1}}})
        f_{\varphi}(\pmb{\varphi_{i_n}})
        \mathcal{K_K}(\pmb{\varphi_{i_n}} \rightarrow \pmb{\varphi_r})
        \mathcal{S}(\pmb{\varphi_r})
        \, \dif \tau_{n}
        \dif \pmb{\varphi_{r_{n-1}}}
        \dif \tau_{n-1}
        \cdots
        \dif \tau_{2}
        \dif \pmb{\varphi_{r_1}}.
    \end{aligned}
    \end{equation}
    We now define the normalisation constants
    $\mathcal{N_{K}}_n : \mathbb{R}^3 \times \mathbb{R}^3 \times \mathbb{R} \rightarrow (0,\infty)$,
    $\mathcal{N_K}_n = \mathcal{N_K}_n(\mathbf{r_i}, \mathbf{v_i}, t_i)$, as
    \begin{equation}
    \begin{aligned}
        \mathcal{N_K}_n(\mathbf{r_i}, \mathbf{v_i}, t_i)
        &=
        \int_{\pmb{\varphi_r}}
        \mathcal{K_{K}}_n(
        \mathbf{r_i} \rightarrow \mathbf{r_r},
        \mathbf{v_i} \rightarrow \mathbf{v_r},
        t_i \rightarrow t_r)
        \, \dif \pmb{\varphi_r}
        \\
        &=
        \int_{\pmb{\varphi_{r_1}}}
        \mathcal{K_K}(\pmb{\varphi_i} \rightarrow \pmb{\varphi_{r_1}})
        \int_{0}^{\infty}
        \mathcal{S}(\pmb{\varphi_{i_2}} \mid \pmb{\varphi_{r_1}})
        f_{\varphi}(\pmb{\varphi_{i_2}})
        \int_{\pmb{\varphi_{r_2}}}
        \mathcal{K_K}(\pmb{\varphi_{i_2}} \rightarrow \pmb{\varphi_{r_2}})
        \cdots
        \\
        &\qquad \cdot
        \int_{0}^{\infty}
        \mathcal{S}(\pmb{\varphi_{i_n}} \mid \pmb{\varphi_{r_{n-1}}})
        f_{\varphi}(\pmb{\varphi_{i_n}})
        \int_{\pmb{\varphi_r}}
        \mathcal{K_K}(\pmb{\varphi_{i_n}} \rightarrow \pmb{\varphi_r})
        \mathcal{S}(\pmb{\varphi_r})
        \, \dif \pmb{\varphi_r}
        \dif \tau_n
        \dif \pmb{\varphi_{r_{n-1}}}
        \dif \tau_{n-1}
        \cdots
        \dif \tau_2
        \dif \pmb{\varphi_{r_1}}.
    \end{aligned}
    \end{equation}
    Finally, we define
    $\mathcal{R_{K}}_n : \mathbb{R}^3 \times \mathbb{R}^3 \times \mathbb{R} \rightarrow (0,\infty)$,
    $\mathcal{R_K}_n = \mathcal{R_K}_n(\mathbf{r_i}, \mathbf{v_i}, t_i)$, as
    \begin{equation}
    \begin{aligned}
        \mathcal{R_K}_n(\mathbf{r_i}, \mathbf{v_i}, t_i)
        &=
        \int_{\pmb{\varphi_{r_1}}}
        \mathcal{K_K}(\pmb{\varphi_i} \rightarrow \pmb{\varphi_{r_1}})
        \int_{0}^{\infty}
        \mathcal{S}(\pmb{\varphi_{i_2}} \mid \pmb{\varphi_{r_1}})
        f_{\varphi}(\pmb{\varphi_{i_2}})
        \int_{\pmb{\varphi_{r_2}}}
        \mathcal{K_K}(\pmb{\varphi_{i_2}} \rightarrow \pmb{\varphi_{r_2}})
        \cdots
        \\
        &\qquad \cdot
        \int_{0}^{\infty}
        \mathcal{S}(\pmb{\varphi_{i_n}} \mid \pmb{\varphi_{r_{n-1}}})
        f_{\varphi}(\pmb{\varphi_{i_n}})
        \int_{\pmb{\varphi_r}}
        \mathcal{K_K}(\pmb{\varphi_{i_n}} \rightarrow \pmb{\varphi_r})
        \mathcal{O}(\pmb{\varphi_r})
        \, \dif \pmb{\varphi_r}
        \dif \tau_n
        \dif \pmb{\varphi_{r_{n-1}}}
        \dif \tau_{n-1}
        \cdots
        \dif \tau_2
        \dif \pmb{\varphi_{r_1}}.
    \end{aligned}
    \end{equation}
    Next, we prove by induction that the normalisation sum of the kernels for the first $M$ reflections is
    \begin{equation} \label{eq:induction_normalisation}
    \sum_{n=1}^{M}
    \mathcal{N_K}_n(\mathbf{r_i}, \mathbf{v_i}, t_i)
    =
    1 - \mathcal{R_K}_M(\mathbf{r_i}, \mathbf{v_i}, t_i).
    \end{equation}
    This expression is immediately proven for the base case $M=1$ by substituting \cref{eq:shadowing_condition_2} into the one-reflection normalisation constant:
    \begin{equation}
        \begin{aligned}
            \mathcal{N_K}_1(\mathbf{r_i}, \mathbf{v_i}, t_i)
            &=
            \int_{\pmb{\varphi_r}}
            \mathcal{K_{K}}_1(
            \mathbf{r_i} \rightarrow \mathbf{r_r},
            \mathbf{v_i} \rightarrow \mathbf{v_r},
            t_i \rightarrow t_r)
            \, \dif \pmb{\varphi_r}
            \\
            &=
            \int_{\pmb{\varphi_r}}
            \mathcal{K_K}(\pmb{\varphi_i} \rightarrow \pmb{\varphi_r})
            \mathcal{S}(\pmb{\varphi_r})
            \, \dif \pmb{\varphi_r}
            \\
            &=
            \int_{\pmb{\varphi_r}}
            \mathcal{K_K}(\pmb{\varphi_i} \rightarrow \pmb{\varphi_r})
            \left[1 - \mathcal{O}(\pmb{\varphi_r})\right]
            \, \dif \pmb{\varphi_r}
            \\
            &=
            \int_{\pmb{\varphi_r}}
            \mathcal{K_K}(\pmb{\varphi_i} \rightarrow \pmb{\varphi_r})
            \, \dif \pmb{\varphi_r}
            -
            \int_{\pmb{\varphi_r}}
            \mathcal{K_K}(\pmb{\varphi_i} \rightarrow \pmb{\varphi_r})
            \mathcal{O}(\pmb{\varphi_r})
            \, \dif \pmb{\varphi_r}
            \\
            &=
            1 - \mathcal{R_K}_1(\mathbf{r_i}, \mathbf{v_i}, t_i),
        \end{aligned}
    \end{equation}
    where we have used the normalisation property of $\mathcal{K_K}$. We now prove that \cref{eq:induction_normalisation} holds for $M+1$ reflections, provided that it holds for $M$ reflections. We begin by rewriting the sum $\mathcal{R_K}_{n}(\mathbf{r_i}, \mathbf{v_i}, t_i) + \mathcal{N_K}_{n}(\mathbf{r_i}, \mathbf{v_i}, t_i)$ as
    \begin{equation}
        \begin{aligned}
            \mathcal{R_K}_{n}(\mathbf{r_i}, \mathbf{v_i}, t_i)
            +
            \mathcal{N_K}_{n}(\mathbf{r_i}, \mathbf{v_i}, t_i)
            &=
            \int_{\pmb{\varphi_{r_1}}}
            \mathcal{K_K}(\pmb{\varphi_i} \rightarrow \pmb{\varphi_{r_1}})
            \int_{0}^{\infty}
            \mathcal{S}(\pmb{\varphi_{i_2}} \mid \pmb{\varphi_{r_1}})
            f_{\varphi}(\pmb{\varphi_{i_2}})
            \int_{\pmb{\varphi_{r_2}}}
            \mathcal{K_K}(\pmb{\varphi_{i_2}} \rightarrow \pmb{\varphi_{r_2}})
            \cdots
            \\
            &\qquad \cdot
            \int_{0}^{\infty}
            \mathcal{S}(\pmb{\varphi_{i_n}} \mid \pmb{\varphi_{r_{n-1}}})
            f_{\varphi}(\pmb{\varphi_{i_n}})
            \int_{\pmb{\varphi_r}}
            \mathcal{K_K}(\pmb{\varphi_{i_n}} \rightarrow \pmb{\varphi_r})
            \left[
            \mathcal{S}(\pmb{\varphi_r})
            +
            \mathcal{O}(\pmb{\varphi_r})
            \right]
            \, \dif \pmb{\varphi_r}
            \\
            &\qquad \cdot
            \dif \tau_n
            \dif \pmb{\varphi_{r_{n-1}}}
            \dif \tau_{n-1}
            \cdots
            \dif \tau_2
            \dif \pmb{\varphi_{r_1}}
            \\
            &=
            \int_{\pmb{\varphi_{r_1}}}
            \mathcal{K_K}(\pmb{\varphi_i} \rightarrow \pmb{\varphi_{r_1}})
            \int_{0}^{\infty}
            \mathcal{S}(\pmb{\varphi_{i_2}} \mid \pmb{\varphi_{r_1}})
            f_{\varphi}(\pmb{\varphi_{i_2}})
            \int_{\pmb{\varphi_{r_2}}}
            \mathcal{K_K}(\pmb{\varphi_{i_2}} \rightarrow \pmb{\varphi_{r_2}})
            \cdots
            \\
            &\qquad \cdot
            \int_{0}^{\infty}
            \mathcal{S}(\pmb{\varphi_{i_n}} \mid \pmb{\varphi_{r_{n-1}}})
            f_{\varphi}(\pmb{\varphi_{i_n}})
            \int_{\pmb{\varphi_r}}
            \mathcal{K_K}(\pmb{\varphi_{i_n}} \rightarrow \pmb{\varphi_r})
            \, \dif \pmb{\varphi_r}
            \dif \tau_n
            \dif \pmb{\varphi_{r_{n-1}}}
            \dif \tau_{n-1}
            \cdots
            \dif \tau_2
            \dif \pmb{\varphi_{r_1}}
            \\
            &=
            \int_{\pmb{\varphi_{r_1}}}
            \mathcal{K_K}(\pmb{\varphi_i} \rightarrow \pmb{\varphi_{r_1}})
            \int_{0}^{\infty}
            \mathcal{S}(\pmb{\varphi_{i_2}} \mid \pmb{\varphi_{r_1}})
            f_{\varphi}(\pmb{\varphi_{i_2}})
            \int_{\pmb{\varphi_{r_2}}}
            \mathcal{K_K}(\pmb{\varphi_{i_2}} \rightarrow \pmb{\varphi_{r_2}})
            \cdots
            \\
            &\qquad \cdot
            \int_{0}^{\infty}
            \mathcal{S}(\pmb{\varphi_{i_n}} \mid \pmb{\varphi_{r_{n-1}}})
            f_{\varphi}(\pmb{\varphi_{i_n}})
            \, \dif \tau_n
            \dif \pmb{\varphi_{r_{n-1}}}
            \dif \tau_{n-1}
            \cdots
            \dif \tau_2
            \dif \pmb{\varphi_{r_1}}
            \\
            &=
            \int_{\pmb{\varphi_{r_1}}}
            \mathcal{K_K}(\pmb{\varphi_i} \rightarrow \pmb{\varphi_{r_1}})
            \int_{0}^{\infty}
            \mathcal{S}(\pmb{\varphi_{i_2}} \mid \pmb{\varphi_{r_1}})
            f_{\varphi}(\pmb{\varphi_{i_2}})
            \int_{\pmb{\varphi_{r_2}}}
            \mathcal{K_K}(\pmb{\varphi_{i_2}} \rightarrow \pmb{\varphi_{r_2}})
            \cdots
            \\
            &\qquad \cdot
            \int_{\pmb{\varphi_{r_{n-1}}}}
            \mathcal{K_K}(\pmb{\varphi_{i_{n-1}}} \rightarrow \pmb{\varphi_{r_{n-1}}})
            \mathcal{O}(\pmb{\varphi_{r_{n-1}}})
            \, \dif \pmb{\varphi_{r_{n-1}}}
            \dif \tau_{n-1}
            \cdots
            \dif \tau_2
            \dif \pmb{\varphi_{r_1}}
            \\
            &=
            \mathcal{R_K}_{n-1}(\mathbf{r_i}, \mathbf{v_i}, t_i),
        \end{aligned}
    \end{equation}
    where we have used the normalisation of $\mathcal{K_K}$ and the definition of $\mathcal{O}$. Then, we expand \cref{eq:induction_normalisation} for $M+1$ reflections as
    \begin{equation}
        \begin{aligned}
            \sum_{n=1}^{M+1}
            \mathcal{N_K}_n(\mathbf{r_i}, \mathbf{v_i}, t_i)
            &=
            \sum_{n=1}^{M}
            \mathcal{N_K}_n(\mathbf{r_i}, \mathbf{v_i}, t_i)
            +
            \mathcal{N_K}_{M+1}(\mathbf{r_i}, \mathbf{v_i}, t_i)
            \\
            &=
            \sum_{n=1}^{M}
            \mathcal{N_K}_n(\mathbf{r_i}, \mathbf{v_i}, t_i)
            +
            \mathcal{R_K}_{M}(\mathbf{r_i}, \mathbf{v_i}, t_i)
            -
            \mathcal{R_K}_{M+1}(\mathbf{r_i}, \mathbf{v_i}, t_i)
            \\
            &=
            1 -
            \mathcal{R_K}_{M}(\mathbf{r_i}, \mathbf{v_i}, t_i)
            +
            \mathcal{R_K}_{M}(\mathbf{r_i}, \mathbf{v_i}, t_i)
            -
            \mathcal{R_K}_{M+1}(\mathbf{r_i}, \mathbf{v_i}, t_i)
            \\
            &=
            1 - \mathcal{R_K}_{M+1}(\mathbf{r_i}, \mathbf{v_i}, t_i).
        \end{aligned}
    \end{equation}
    Thus, \cref{eq:induction_normalisation} holds for all $M \in \mathbb{N}$. We now take the limit as $M \to \infty$ to obtain the normalisation constant of $\mathcal{K_{MK}}$:
    \begin{equation} \label{eq:normalisation_final_step}
        \begin{aligned}
            \int_{\mathbb{R}^3}
            \int_{\mathbb{R}^3}
            \int_{\mathbb{R}}
            \mathcal{K_{MK}}(
            \mathbf{r_i} \rightarrow \mathbf{r_r},
            \mathbf{v_i} \rightarrow \mathbf{v_r},
            t_i \rightarrow t_r)
            \, \dif t_r \dif \mathbf{v_r} \dif \mathbf{r_r}
            &=
            \lim_{M\rightarrow\infty}
            \left[
            \sum_{n=1}^{M}
            \mathcal{N_K}_n(\mathbf{r_i}, \mathbf{v_i}, t_i)
            \right]
            \\
            &=
            \lim_{M\rightarrow\infty}
            \left[
            1 - \mathcal{R_K}_{M}(\mathbf{r_i}, \mathbf{v_i}, t_i)
            \right]
            \\
            &=
            1 - \mathcal{L}_R.
        \end{aligned}
    \end{equation}
    It is immediate that $\mathcal{L}_R \geq 0$. We employ the second condition in the statement of the lemma, namely that there exists $\varepsilon \in (0,1]$ such that, for every $\pmb{\varphi_i} \in \mathbb{R}^3 \times \mathbb{R}^3 \times \mathbb{R}$,
    \begin{equation}
        \int_{\pmb{\varphi_r}}
        \mathcal{K_K}(\pmb{\varphi_i} \rightarrow \pmb{\varphi_r})
        \mathcal{S}(\pmb{\varphi_r})
        \, \dif \pmb{\varphi_r}
        >
        \varepsilon.
    \end{equation}
    Since $\mathcal{O}(\pmb{\varphi_r}) = 1 - \mathcal{S}(\pmb{\varphi_r})$ and $\mathcal{K_K}$ is normalised, it follows that
    \begin{equation} \label{eq:inner_remainder_bound}
        \begin{aligned}
            \int_{\pmb{\varphi_r}}
            \mathcal{K_K}(\pmb{\varphi_i} \rightarrow \pmb{\varphi_r})
            \mathcal{O}(\pmb{\varphi_r})
            \, \dif \pmb{\varphi_r}
            &=
            \int_{\pmb{\varphi_r}}
            \mathcal{K_K}(\pmb{\varphi_i} \rightarrow \pmb{\varphi_r})
            \left[
            1 - \mathcal{S}(\pmb{\varphi_r})
            \right]
            \, \dif \pmb{\varphi_r}
            \\
            &=
            1 -
            \int_{\pmb{\varphi_r}}
            \mathcal{K_K}(\pmb{\varphi_i} \rightarrow \pmb{\varphi_r})
            \mathcal{S}(\pmb{\varphi_r})
            \, \dif \pmb{\varphi_r}
            \\
            &\leq
            1 - \varepsilon.
        \end{aligned}
    \end{equation}
    Therefore, for every $M \geq 1$,
    \begin{equation}
        \begin{aligned}
            \mathcal{R_K}_{M}(\mathbf{r_i}, \mathbf{v_i}, t_i)
            &=
            \int_{\pmb{\varphi_{r_1}}}
            \mathcal{K_K}(\pmb{\varphi_i} \rightarrow \pmb{\varphi_{r_1}})
            \int_{0}^{\infty}
            \mathcal{S}(\pmb{\varphi_{i_2}} \mid \pmb{\varphi_{r_1}})
            f_{\varphi}(\pmb{\varphi_{i_2}})
            \int_{\pmb{\varphi_{r_2}}}
            \mathcal{K_K}(\pmb{\varphi_{i_2}} \rightarrow \pmb{\varphi_{r_2}})
            \cdots
            \\
            &\qquad \cdot
            \int_{0}^{\infty}
            \mathcal{S}(\pmb{\varphi_{i_M}} \mid \pmb{\varphi_{r_{M-1}}})
            f_{\varphi}(\pmb{\varphi_{i_M}})
            \int_{\pmb{\varphi_r}}
            \mathcal{K_K}(\pmb{\varphi_{i_M}} \rightarrow \pmb{\varphi_r})
            \mathcal{O}(\pmb{\varphi_r})
            \, \dif \pmb{\varphi_r}
            \, \dif \tau_M
            \, \dif \pmb{\varphi_{r_{M-1}}}
            \, \dif \tau_{M-1}
            \cdots
            \dif \tau_2
            \dif \pmb{\varphi_{r_1}}
            \\
            &\leq
            (1-\varepsilon)
            \int_{\pmb{\varphi_{r_1}}}
            \mathcal{K_K}(\pmb{\varphi_i} \rightarrow \pmb{\varphi_{r_1}})
            \int_{0}^{\infty}
            \mathcal{S}(\pmb{\varphi_{i_2}} \mid \pmb{\varphi_{r_1}})
            f_{\varphi}(\pmb{\varphi_{i_2}})
            \int_{\pmb{\varphi_{r_2}}}
            \mathcal{K_K}(\pmb{\varphi_{i_2}} \rightarrow \pmb{\varphi_{r_2}})
            \cdots
            \\
            &\qquad \cdot
            \int_{0}^{\infty}
            \mathcal{S}(\pmb{\varphi_{i_M}} \mid \pmb{\varphi_{r_{M-1}}})
            f_{\varphi}(\pmb{\varphi_{i_M}})
            \, \dif \tau_M
            \, \dif \pmb{\varphi_{r_{M-1}}}
            \, \dif \tau_{M-1}
            \cdots
            \dif \tau_2
            \dif \pmb{\varphi_{r_1}}
            \\
            &=
            (1-\varepsilon)
            \int_{\pmb{\varphi_{r_1}}}
            \mathcal{K_K}(\pmb{\varphi_i} \rightarrow \pmb{\varphi_{r_1}})
            \int_{0}^{\infty}
            \mathcal{S}(\pmb{\varphi_{i_2}} \mid \pmb{\varphi_{r_1}})
            f_{\varphi}(\pmb{\varphi_{i_2}})
            \int_{\pmb{\varphi_{r_2}}}
            \mathcal{K_K}(\pmb{\varphi_{i_2}} \rightarrow \pmb{\varphi_{r_2}})
            \cdots
            \\
            &\qquad \cdot
            \mathcal{O}(\pmb{\varphi_{r_{M-1}}})
            \, \dif \pmb{\varphi_{r_{M-1}}}
            \, \dif \tau_{M-1}
            \cdots
            \dif \tau_2
            \dif \pmb{\varphi_{r_1}}
            \\
            &\leq \cdots \leq (1-\varepsilon)^M.
        \end{aligned}
    \end{equation}
    Hence,
    \begin{equation}
        \mathcal{L}_R
        =
        \lim_{M\to\infty}
        \mathcal{R_K}_M(\mathbf{r_i}, \mathbf{v_i}, t_i)
        \leq
        \lim_{M\to\infty}
        (1-\varepsilon)^M
        =
        0.
    \end{equation}
    Since also $\mathcal{L}_R \geq 0$, we conclude that $\mathcal{L}_R = 0$. Substituting this value into \cref{eq:normalisation_final_step}, we obtain
    \begin{equation}
        \begin{aligned}
            \int_{\mathbb{R}^3}
            \int_{\mathbb{R}^3}
            \int_{\mathbb{R}}
            \mathcal{K_{MK}}(
            \mathbf{r_i} \rightarrow \mathbf{r_r},
            \mathbf{v_i} \rightarrow \mathbf{v_r},
            t_i \rightarrow t_r)
            \, \dif t_r \dif \mathbf{v_r} \dif \mathbf{r_r}
            =
            1 - \mathcal{L}_R
            =
            1,
        \end{aligned}
    \end{equation}
    which proves that $\mathcal{K_{MK}}$ satisfies the normalisation condition.
\end{proof}

\begin{lemma} \label{lemma:nonnegativity_multi_refl}
    Let $\mathcal{K_{MK}} : \mathbb{R}^3 \times \mathbb{R}^3 \times \mathbb{R}
\times \mathbb{R}^3 \times \mathbb{R}^3 \times \mathbb{R}
\rightarrow [0, \infty)$, $\mathcal{K_{MK}} = \mathcal{K_{MK}}\left(\mathbf{r_i} \rightarrow \mathbf{r_r}, \mathbf{v_i} \rightarrow \mathbf{v_r}, t_i \rightarrow t_r \right)$ be a Lebesgue-integrable function, defined in the global reference frame from \cref{def:global_frame} and associated with a surface $\Psi$ in the sense of \cref{def:surface_general}, whose expression is
\begin{equation}
    \begin{aligned}
        \mathcal{K_{MK}}(
        \mathbf{r_i} \rightarrow \mathbf{r_r},
        \mathbf{v_i} \rightarrow \mathbf{v_r},
        t_i \rightarrow t_r)
        &=
        \sum_{n=1}^{\infty}
        \Bigg\{
        \underset{\mathbf{v_{r_k}},\mathbf{v_{i_k}}}{\int\cdots\int}
        \underset{\mathbf{r_{r_k}},\mathbf{r_{i_k}}}{\int\cdots\int}
        \underset{t_{r_k},t_{i_k}}{\int\cdots\int}
        \mathcal{K_K}(
        \mathbf{r_i}\rightarrow\mathbf{r_{r_1}},
        \mathbf{v_i}\rightarrow\mathbf{v_{r_1}},
        t_i\rightarrow t_{r_1})
        \\
        &\qquad\cdot
        \prod_{k=2}^{n}
        \Bigg[
        \mathcal{S}(
        \mathbf{r_{i_k}}, \mathbf{v_{i_k}}, t_{i_k}
        \mid
        \mathbf{r_{r_{k-1}}},\mathbf{v_{r_{k-1}}},t_{r_{k-1}})
        \,
        \\
        &\qquad\cdot \, f_{\varphi}(\mathbf{r_{i_k}}, \mathbf{v_{i_k}}, t_{i_k}) \,\mathcal{K_K}(
        \mathbf{r_{i_k}}\rightarrow\mathbf{r_{r_k}},
        \mathbf{v_{i_{k}}}\rightarrow\mathbf{v_{r_k}},
        t_{i_k}\rightarrow t_{r_k})
        \Bigg]
        \\
        &\qquad\cdot
        \dif\mathbf{v_{r_1}}\cdots\dif\mathbf{v_{r_{n-1}}}
        \dif\mathbf{v_{i_2}}\cdots\dif\mathbf{v_{i_{n-1}}}
        \dif\mathbf{r_{r_1}}\cdots\dif\mathbf{r_{r_{n-1}}}
        \dif\mathbf{r_{i_2}}\cdots\dif\mathbf{r_{i_{n-1}}}
        \\
        &\qquad\cdot
        \dif t_{r_1}\cdots\dif t_{r_{n-1}}
        \dif t_{i_2}\cdots\dif t_{i_{n-1}}
        \;
        \mathcal{S}(\mathbf{r_r},\mathbf{v_r},t_r)
        \Bigg\},
    \end{aligned}
\end{equation}
as given by \cref{def:multi_reflection_scattering_kernel}, where $\mathcal{K_K}: \mathbb{R}^3 \times \mathbb{R}^3 \times \mathbb{R}
\times \mathbb{R}^3 \times \mathbb{R}^3 \times \mathbb{R} \rightarrow\mathbb{R}$, written as $\mathcal{K_K}= \mathcal{K_K}(\mathbf{r_{i}} \rightarrow \mathbf{r_{r}}, \mathbf{v_{i}} \rightarrow \mathbf{v_{r}}, t_i \rightarrow t_r)$, is a Lebesgue-integrable single-reflection global scattering kernel defined in the same frame. Here, $\mathbf{r_{i}}$, $\mathbf{r_{r}}$, $\mathbf{v_{i}}$, and $\mathbf{v_{r}}$ are $3\times 1$ real vectors denoting the incident and reflected particle positions and velocities, respectively, $t_i$ and $t_r$ represent the corresponding entry and exit times, respectively. $\mathcal{S}(\mathbf{r_r},\mathbf{v_r},t_r)$ and $\mathcal{S}(\mathbf{r_i},\mathbf{v_i}, t_i \, | \, \mathbf{r_r},\mathbf{v_r},t_r)$ denote the one and two-point shadowing
functions, respectively, defined in \cref{def:shadowing_single_scale} and associated with the
surface $\Psi$. Then $\mathcal{K_{MK}}$ satisfies non-negativity in the sense of \cref{def:nonnegativity} if the local kernel $\mathcal{K_L}$ associated with $\mathcal{K_K}$ satisfies non-negativity.
\end{lemma}
\begin{proof}
     The proof is, again, trivial. By \cref{def:shadowing_single_scale}, we know that $\mathcal{S}(
        \mathbf{r_{i}},\mathbf{v_i},t_{i}
        \mid
        \mathbf{r_{r}},\mathbf{v_{r}},t_{r})$, $\mathcal{S}( \mathbf{r_{r}},\mathbf{v_{r}},t_{r})$ and $f_{\varphi}(\mathbf{r_{i}},\mathbf{v_i},t_{i}) \geq 0 \, \forall \, (\mathbf{r_i}, \mathbf{v_i}, t_i, \mathbf{r_r}, \mathbf{v_r}, t_r) \in \mathbb{R}^3 \times \mathbb{R}^3 \times \mathbb{R}
\times \mathbb{R}^3 \times \mathbb{R}^3 \times \mathbb{R}$. We further know that $\mathcal{K_K}$ satisfies non-negativity, i.e. $\mathcal{K_K}(\mathbf{r_{i}} \rightarrow \mathbf{r_{r}}, \mathbf{v_{i}} \rightarrow \mathbf{v_{r}}, t_i \rightarrow t_r) \geq 0, \, \forall (\mathbf{r_{i}}, \mathbf{r_{r}}, \mathbf{v_{i}}, \mathbf{v_{r}}, t_i, t_r) \in \mathbb{R}^3 \times \mathbb{R}^3 \times \mathbb{R} \times \mathbb{R}^3 \times \mathbb{R}^3 \times \mathbb{R}
\rightarrow [0, \infty)$. Since all terms in the integral are non-negative, we can infer that $\mathcal{K_{MK}}\left(\mathbf{r_i} \rightarrow \mathbf{r_r}, \mathbf{v_i} \rightarrow \mathbf{v_r, t_i \rightarrow t_r}\right)$ is non-negative for any $\mathbf{v_r}, \mathbf{n_L} \in \mathbb{R}^3 \times \mathbb{R}^3$, and thus satisfies the condition. \qedhere
\end{proof}

\section{Multi-scale scattering formalism} \label{sec:Multiscale_formalism}

So far, we have developed a scattering-kernel formulation for gas-particle reflection under the combined action of a rough surface morphology and a local scattering law $\mathcal{K_L}$. However, real surfaces generally contain roughness over multiple length scales, so that a single characteristic scale $R$, as introduced in \cref{def:length_scale}, is not sufficient to represent the full morphology. We therefore use the multi-scale surface framework of \cref{subsec:multiscale_def} to extend the scattering construction to surfaces with arbitrary power spectral densities, which we represent through a hierarchy of normal distributions. This leads to a recursive formulation in which a newly defined scattering operator $\circ$ is applied successively across scales. In particular, for multi-scale surfaces admitting a height representation in a common global frame, we show that repeated application of $\circ$ is equivalent to applying $\circ$ once to the sum of the corresponding surface height profiles. Thus, if a surface $\Psi^{(m)}$ is decomposed into components $\Psi^{(1)}_k$, with $k \in \overline{1,m}$, each associated with a distinct characteristic scale ordered such that $R_{k+1} \ll R_k$, the resulting multi-scale kernel can be constructed by successively applying the scattering operator to these components. In the following, we formalise this construction and prove the corresponding composition and summation properties of $\circ$.

\subsection{Construction of scattering operator} \label{subsec:scattering_operator}

Herein, we formalise \cref{eq:multi_reflection_kernel} through a scattering operator, $\circ$, defined on the set of all surfaces, $\mathcal{P}$, and the set of all kernels, $\mathcal{T}$, which maps a surface $\Psi$ and a local kernel $\mathcal{K_L}$ to the corresponding roughness-modified kernel. We further prove that the co-domain of $\circ$ is $\mathcal{T}^{s}$, provided that $\mathcal{K_L} \in \mathcal{T}^p$.

\begin{definition} \label{def:scattering_operator}
    Let $\mathcal{K_L} : \mathbb{R}^3 \times \mathbb{R}^3 \times \mathbb{R}
    \times \mathbb{R}^3 \times \mathbb{R}^3 \times \mathbb{R}
    \rightarrow [0, \infty)$ be a Lebesgue-integrable function, and let $\mathcal{K_L} = \mathcal{K_L}\left(\pmb{\varphi_i} \rightarrow \pmb{\varphi_r} \right)$ denote the corresponding kernel defined in the local reference frame of \cref{def:local_frame} and associated with a multi-scale surface $\Psi^{(m)}$ in the sense of \cref{def:surface_multiscale}. As convention, we write $a^{(q)}_k$ for a quantity $a$ associated with reflection $k$ and expressed in frame $q$, with $q \in \{0,1,\dots,m\}$ for a surface with $m$ scales, where frame $0$ is the most global frame and frame $m$ is the local frame where $\Psi$ becomes smooth. We then define the scattering operator $\circ : \mathcal{P} \times \mathcal{T} \rightarrow \mathcal{T}$, for a surface with $m$ scales, by
    \begin{equation} \label{eq:scattering_operator}
    \boxed{
    \begin{aligned}
        &\left[\Psi^{(m)} \circ \mathcal{K_L}\right]\left(
        \pmb{\varphi}^{(0)}_i \rightarrow \pmb{\varphi}^{(0)}_r
        \right)
        =
        \frac{1}{\left\langle\Dot{\mathcal{M}}^{(m)}\left(
        \pmb{\varphi}^{(0)}_i
        \right)\right\rangle}
        \sum_{n=1}^{\infty}
        \Biggl\{
        \underset{
        \pmb{\varphi}^{(0)}_{r_{1:n-1}},\,
        \pmb{\varphi}^{(0)}_{i_{2:n}},\,
        \mathbf{n}^{(1:m-1)}_{1:n}
        }{\int\cdots\int}
        \Dot{\mathcal{M}}^{(m)}\left(
        \pmb{\varphi}^{(0)}_i
        \,\Big|\,
        \mathbf{n}^{(m-1)}_{1},\dots,\mathbf{n}^{(0)}_{1}
        \right)
        \\
        &\quad\cdot
        p_{n_1}^{(m-1)}\left(\mathbf{n}^{(m-1)}_1, \dots, \mathbf{n}^{(1)}_1 \, | \, \mathbf{n}^{(0)}_1, \mathbf{v}^{(0)}_{i_1}\right) \, \mathcal{K_K}\left(
        \pmb{\varphi}^{(0)}_i \rightarrow \pmb{\varphi}^{(0)}_{r_1}
        \right)
        \prod_{k=2}^{n}
        \Biggl[
        \Dot{\mathcal{S}}^{(m)}\Bigl(
        \pmb{\varphi}^{(0)}_{i_{k}},
        \mathbf{n}^{(m-1)}_k,\dots,\mathbf{n}^{(0)}_k
        \,\Big|\,
        \pmb{\varphi}^{(0)}_{r_{k-1}},
        \mathbf{n}^{(m-1)}_{k-1},\dots,\mathbf{n}^{(0)}_{k-1}
        \Bigr)
        \\
        &\quad\cdot
        \mathcal{K_K}\left(
        \pmb{\varphi}^{(0)}_{i_k} \rightarrow \pmb{\varphi}^{(0)}_{r_k}
        \right)
        \Biggr]
        \mathcal{S}^{(m)}\left(
        \pmb{\varphi}^{(0)}_r
        \,\Big|\,
        \mathbf{n}^{(m-1)}_{n},\dots,\mathbf{n}^{(0)}_{n}
        \right) \, \dif\pmb{\varphi}^{(0)}_{r_1}\cdots\dif\pmb{\varphi}^{(0)}_{r_{n-1}}\,
        \dif\pmb{\varphi}^{(0)}_{i_2}\cdots\dif\pmb{\varphi}^{(0)}_{i_n}\,
        \prod_{\alpha=1}^{m-1}
        \prod_{k=1}^{n}
        \dif\mathbf{n}^{(\alpha)}_k
        \Biggr\}.
    \end{aligned}
    }
\end{equation}
    with the single-reflection kernel, $\mathcal{K_K}$, taking the form given in \cref{def:single_reflection_scattering_kernel},
    \begin{equation}
        \mathcal{K_K}\left(
        \pmb{\varphi}^{(0)}_{i_k} \rightarrow \pmb{\varphi}^{(0)}_{r_k}
        \right)
        =
        \int_{\mathbb{S}^2}
        \left\langle\Dot{\mathcal{M}}_L\left(
            \pmb{\varphi}^{(m)}_i
            \right)\right\rangle \, \mathcal{K_L}\left(
        \pmb{\varphi}^{(m)}_{i_k} \rightarrow \pmb{\varphi}^{(m)}_{r_k}
        \right)
        \cdot
        \frac{\mathcal{Q}\left(\mathbf{v}^{(m)}_{i_k} \, | \, \mathbf{n}^{(m)}_k\right)}
        {\mathcal{Q}\left(\mathbf{v}^{(0)}_{i_k} \, | \, \mathbf{n}^{(0)}_k\right)}
        \, p_{n_k}^{(m)}\left(\mathbf{n}^{(m)}_k \, | \, \mathbf{n}^{(m-1)}_k, \mathbf{v}^{(m-1)}_{i_k}\right)
        \, \dif \mathbf{n}^{(m)}_k.
    \end{equation}
    The two-point and one-point shadowing functions $\mathcal{S}^{(m)}\left(\pmb{\varphi}^{(0)}_{i_k}, \mathbf{n}^{(m-1)}_k,\dots,\mathbf{n}^{(0)}_k \mid \pmb{\varphi}^{(0)}_{r_{k-1}}, \mathbf{n}^{(m-1)}_{k-1},\dots,\mathbf{n}^{(0)}_{k-1}\right)$ and $\mathcal{S}^{(m)}\left(\pmb{\varphi}^{(0)}_{r} \mid \mathbf{n}^{(m-1)}_n,\dots,\mathbf{n}^{(0)}_n\right)$ and masking functions $\mathcal{M}^{(m)}\left(\pmb{\varphi}^{(0)}_{i} \mid \mathbf{n}^{(m-1)}_1,\dots,\mathbf{n}^{(0)}_1\right)$ are defined as in \cref{def:shadowing_multi_scale}. Furthermore, $\left\langle\mathcal{M}^{(m)}\right\rangle$ and $\left\langle\mathcal{M}_L\right\rangle$ represent the normal-averaged masking functions of the first $m$ scales and any local scales of $\mathcal{K_L}$, respectively. Finally, $\pmb{\varphi}^{(q)}_{i_k} = \left(\mathbf{r}^{(q)}_{i_k}, \mathbf{v}^{(q)}_{i_k}, t_{i_k}\right)$ and $\pmb{\varphi}^{(q)}_{r_k} = \left(\mathbf{r}^{(q)}_{r_k}, \mathbf{v}^{(q)}_{r_k}, t_{r_k}\right)$ denote the incident and reflected particle states, respectively, at reflection $k$, expressed in frame $q$. By convention, $\pmb{\varphi}^{(0)}_{i_1} = \pmb{\varphi}^{(0)}_i$ and $\pmb{\varphi}^{(0)}_{r_n} = \pmb{\varphi}^{(0)}_r$.
\end{definition}

We first establish that the scattering operator $\circ$ is closed under successive application across multi-scale surface hierarchies, i.e. applying it to two multi-scale surfaces in sequence is equivalent to applying it once to a single composite multi-scale surface.

\begin{theorem} \label{theorem:scattering_kernel_composition}
    Let $\circ$ denote the scattering operator in the sense of \cref{def:scattering_operator}, and let $\Psi^{(p)}_1, \Psi^{(q)}_2 \in \mathcal{P}$ be two multi-scale surfaces in the sense of \cref{def:surface_multiscale}, with scales $p,q\in\mathbb{Z}^+$. Furthermore, let $\mathcal{K_L}: \mathbb{R}^3 \times \mathbb{R}^3 \times \mathbb{R}
    \times \mathbb{R}^3 \times \mathbb{R}^3 \times \mathbb{R} \rightarrow [0, \infty)$ be a Lebesgue-integrable function defined in the most local reference frame associated with the composed surface, in the sense of \cref{def:local_frame}. Then, there exists a multi-scale surface $\Psi^{(p+q)} \in \mathcal{P}$ in the sense of \cref{def:surface_multiscale}, with $p+q$ scales, such that
    \begin{equation} \label{eq:scattering_operator_composition}
    \boxed{
        \left[\Psi^{(p)}_1 \circ \left(\Psi^{(q)}_2 \circ \mathcal{K_L}\right)\right]\left(
        \mathbf{r}^{(0)}_i \rightarrow \mathbf{r}^{(0)}_r,
        \mathbf{v}^{(0)}_i \rightarrow \mathbf{v}^{(0)}_r,
        t_i \rightarrow t_r
        \right)
        =
        \left[\Psi^{(p+q)} \circ \mathcal{K_L}\right]\left(
        \mathbf{r}^{(0)}_i \rightarrow \mathbf{r}^{(0)}_r,
        \mathbf{v}^{(0)}_i \rightarrow \mathbf{v}^{(0)}_r,
        t_i \rightarrow t_r
        \right),
    }
    \end{equation}
     $\forall \, \Psi^{(p)}_1, \Psi^{(q)}_2 \in \mathcal{P}$ and $ \forall \, \left(\mathbf{r}^{(0)}_i,\mathbf{v}^{(0)}_i,t_i,\mathbf{r}^{(0)}_r,\mathbf{v}^{(0)}_r,t_r\right)\in \mathbb{R}^3 \times \mathbb{R}^3 \times \mathbb{R} \times \mathbb{R}^3 \times \mathbb{R}^3 \times \mathbb{R}$. Furthermore, $\left[\Psi^{(p+q)} \circ \mathcal{K_L}\right]\left(
    \mathbf{r}^{(0)}_i \rightarrow \mathbf{r}^{(0)}_r,
    \mathbf{v}^{(0)}_i \rightarrow \mathbf{v}^{(0)}_r,
    t_i \rightarrow t_r
    \right)$ is a scattering kernel in the sense of \cref{def:reciprocity,def:normalisation,def:nonnegativity,def:impermeability} provided that $\mathcal{K_L}$ is a pointwise scattering kernel in the sense of \cref{def:pointwise_reciprocity,def:normalisation,def:nonnegativity}, and the associated two-point and one-point shadowing functions $\mathcal{S}^{(p+q)}\left(\pmb{\varphi}^{(0)}_{i_k}, \mathbf{n}^{(p+q-1)}_k,\dots,\mathbf{n}^{(0)}_k \mid \pmb{\varphi}^{(0)}_{r_{k-1}}, \mathbf{n}^{(p+q-1)}_{k-1},\dots,\mathbf{n}^{(0)}_{k-1}\right)$ and $\mathcal{S}^{(p+q)}\left(\pmb{\varphi}^{(0)}_{r} \mid \mathbf{n}^{(p+q-1)}_n,\dots,\mathbf{n}^{(0)}_n\right)$, respectively, satisfy the following three conditions. First, the two-point shadowing function is reciprocal, i.e.,
    \begin{multline} \label{eq:shadowing_multi_scale_condition_1}
        p^{(p+q-1)}_{n_{k-1}}\Bigl(
        \mathbf{n}^{(p+q-1)}_{k-1}, \dots, \mathbf{n}^{(1)}_{k-1}
        \,\Big|\,
        \mathbf{n}^{(0)}_{k-1}, \mathbf{v}^{(0)}_{i_{k-1}}
        \Bigr)
        \, \mathcal{S}^{(p+q)}\left(
        \pmb{\varphi}^{(0)}_{i_k},
        \mathbf{n}^{(p+q-1)}_k,\dots,\mathbf{n}^{(0)}_k
        \mid
        \pmb{\varphi}^{(0)}_{r_{k-1}},
        \mathbf{n}^{(p+q-1)}_{k-1},\dots,\mathbf{n}^{(0)}_{k-1}
        \right)
        = \\
        p^{(p+q-1)}_{n_k}\Bigl(
        \mathbf{n}^{(p+q-1)}_{k}, \dots, \mathbf{n}^{(1)}_{k}
        \,\Big|\,
        \mathbf{n}^{(0)}_k, \mathbf{v}^{(0)}_{i_k}
        \Bigr)
        \, \mathcal{S}^{(p+q)}\left(
        -\pmb{\varphi}^{(0)}_{r_{k-1}},
        \mathbf{n}^{(p+q-1)}_{k-1},\dots,\mathbf{n}^{(0)}_{k-1}
        \mid
        -\pmb{\varphi}^{(0)}_{i_k},
        \mathbf{n}^{(p+q-1)}_k,\dots,\mathbf{n}^{(0)}_k
        \right),
    \end{multline}
    $\forall \, \left(
    \pmb{\varphi}^{(0)}_{i_k},
    \pmb{\varphi}^{(0)}_{r_{k-1}},
    \mathbf{n}^{(p+q-1)}_k,\dots,\mathbf{n}^{(0)}_k,
    \mathbf{n}^{(p+q-1)}_{k-1},\dots,\mathbf{n}^{(0)}_{k-1}
    \right)
    \in
    \left(\mathbb{R}^3 \times \mathbb{R}^3 \times \mathbb{R}\right)^2
    \times
    \left(\mathbb{S}^2\right)^{2(p+q)}$. Second, the two-point and one-point shadowing functions satisfy the normalisation condition
    \begin{multline} \label{eq:shadowing_multi_scale_condition_2}
        \int_{\mathbb{R}} \int_{\mathbb{S}^2} \dots \int_{\mathbb{S}^2}
        \mathcal{S}^{(p+q)}\left(
        \pmb{\varphi}^{(0)}_{i_k},
        \mathbf{n}^{(p+q-1)}_k,\dots,\mathbf{n}^{(0)}_k
        \mid
        \pmb{\varphi}^{(0)}_{r_{k-1}},
        \mathbf{n}^{(p+q-1)}_{k-1},\dots,\mathbf{n}^{(0)}_{k-1}
        \right) \\
        f^{(p+q)}_{\varphi}\left(
        \pmb{\varphi}^{(0)}_{i_k}
        \right)
        \, \dif \mathbf{n}_k^{(p+q-1)} \dif \mathbf{n}_k^{(p+q-2)} \dots \dif \mathbf{n}_k^{(1)} \, \dif t_{i_k}
        = 
        1 -
        \mathcal{S}^{(p+q)}\left(
        \pmb{\varphi}^{(0)}_{r_{k-1}}
        \mid
        \mathbf{n}^{(p+q-1)}_{k-1},\dots,\mathbf{n}^{(0)}_{k-1}
        \right),
    \end{multline}
    $\forall \, \left(
    \pmb{\varphi}^{(0)}_{r_{k-1}},
    \mathbf{n}^{(p+q-1)}_{k-1},\dots,\mathbf{n}^{(0)}_{k-1},
    \mathbf{n}^{(p+q-1)}_k,\dots,\mathbf{n}^{(0)}_k
    \right)
    \in
    \mathbb{R}^3 \times \mathbb{R}^3 \times \mathbb{R} \times \left(\mathbb{S}^2\right)^{2(p+q)}$. Third, $ \exists \, \varepsilon \in (0,1]$ such that, for every $\left(\pmb{\varphi}^{(0)}_i,
    \mathbf{n}^{(p+q-1)},\dots,\mathbf{n}^{(0)}
    \right)\in\mathbb{R}^3 \times \mathbb{R}^3 \times \mathbb{R} \times \left(\mathbb{S}^2\right)^{(p+q)}$ and $\mathcal{K_K}\left(
        \pmb{\varphi}^{(0)}_i \rightarrow \pmb{\varphi}^{(0)}_r
        \right) \in \mathcal{T}^p$, the one-point shadowing function satisfies
    \begin{multline} \label{eq:shadowing_multi_scale_condition_3}
        \int_{\mathbb{S}^2} \dots \int_{\mathbb{S}^2} \int_{\mathbb{R}^3}\int_{\mathbb{R}^3}\int_{\mathbb{R}}
        \mathcal{K_K}\left(
        \pmb{\varphi}^{(0)}_i \rightarrow \pmb{\varphi}^{(0)}_r
        \right) \,
        \mathcal{S}^{(p+q)}\left(
        \pmb{\varphi}^{(0)}_r
        \mid
        \mathbf{n}^{(p+q-1)},\dots,\mathbf{n}^{(0)}
        \right) \\ \cdot \prod_{\alpha = 2}^{(p+q)} p_{n}^{(\alpha-1)}\left( \mathbf{n}^{(\alpha-1)} \, | \, \mathbf{n}^{(\alpha-2)}, \mathbf{v}^{(\alpha-2)}_{i}\right)
        \, \dif \pmb{\varphi}^{(0)}_r  \, \dif \mathbf{n}^{(p+q-1)} \dif \mathbf{n}^{(p+q-2)} \dots \dif \mathbf{n}^{(1)}
        >
        \varepsilon.
    \end{multline}
    Here, $\mathcal{S}^{(p+q)}\left(
    \pmb{\varphi}^{(0)}_{i},
    \mathbf{n}^{(p+q-1)}_{i},\dots,\mathbf{n}^{(0)}_{i}
    \mid
    \pmb{\varphi}^{(0)}_{r},
    \mathbf{n}^{(p+q-1)}_{r},\dots,\mathbf{n}^{(0)}_{r}
    \right)$, $\mathcal{S}^{(p+q)}\left(
    \pmb{\varphi}^{(0)}
    \mid
    \mathbf{n}^{(p+q-1)},\dots,\mathbf{n}^{(0)}
    \right)$, and $\mathcal{M}^{(p+q)}\left(
    \pmb{\varphi}^{(0)}
    \mid
    \mathbf{n}^{(p+q-1)},\dots,\mathbf{n}^{(0)}
    \right)$ denote, respectively, the two-point shadowing, one-point shadowing, and masking functions defined in \cref{def:shadowing_multi_scale}.
\end{theorem}

\begin{proof}
   Let frames $F_0,\dots,F_p$ denote the reference-frame hierarchy associated with $\Psi^{(p)}_1$, with frame $F_0$ global to $\Psi^{(p)}_1$ in the sense of \cref{def:global_frame}, and frame $F_p$ local to $\Psi^{(p)}_1$ and global to $\Psi^{(q)}_2$ in the sense of \cref{def:local_frame}. Similarly, let frames $F_p,\dots,F_{p+q}$ denote the reference-frame hierarchy associated with $\Psi^{(q)}_2$, with frame $F_p$ global to $\Psi^{(q)}_2$, and frame $F_{p+q}$ local to $\Psi^{(q)}_2$. We write $\Dot{\mathcal{S}}^{(p)}$ and $\Dot{\mathcal{M}}^{(p)}$ for the shadowing and masking functions associated with $\Psi^{(p)}_1$, $\mathcal{S}^{(q)}$ and $\Dot{\mathcal{M}}^{(q)}$ for those associated with $\Psi^{(q)}_2$, and $\Dot{\mathcal{M}}^{(p+q)}$ for the masking function associated with the composed surface, in the sense of \cref{def:shadowing_multi_scale}. Moreover, we adopt the convention that $a^{(\ell)}_{k,j}$ denotes a quantity associated with reflection $k$ on $\Psi^{(p)}_1$ and reflection $j$ on $\Psi^{(q)}_2$, expressed in frame $\ell$, with $\ell \in \overline{p,p+q}$. Quantities associated only with reflections on $\Psi^{(p)}_1$ are written as $a^{(\ell)}_{k}$, with $\ell \in \overline{0,p}$. For notational compactness, we now switch to state notation, writing $\pmb{\varphi} = (\mathbf{r},\mathbf{v},t)$. Furthermore, we explicitly expand the kernels $\Psi^{(q)}_2 \circ \mathcal{K_L}$ appearing in the action of $\Psi^{(p)}_1$. Since the kernel acting in frame $p$ is now $\Psi^{(q)}_2 \circ \mathcal{K_L}$ rather than the smooth local kernel $\mathcal{K_L}$ itself, frame $p$ still contains the unresolved roughness of $\Psi^{(q)}_2$. We then start from the left-hand side of \cref{eq:scattering_operator_composition} to obtain
\begin{equation} \label{eq:multiscale_composition_1}
    \begin{aligned}
        & \left[\Psi^{(p)}_1 \circ \left( \Psi^{(q)}_2 \circ \mathcal{K_L}\right)\right]
        \left(\pmb{\varphi}^{(0)}_i \rightarrow \pmb{\varphi}^{(0)}_r \right)
         = \\
         & \qquad\qquad\qquad = \frac{1}{\left\langle\Dot{\mathcal{M}}^{(p+q)}\left(
            \pmb{\varphi}^{(0)}_i
            \right)\right\rangle} \sum_{n=1}^{\infty}
        \Bigg\{
        \underset{\pmb{\varphi}^{(0)}_{r_k},\pmb{\varphi}^{(0)}_{i_k}}{\int\cdots\int}
        \underset{\mathbf{n}^{(1)}_k,\dots,\mathbf{n}^{(p)}_k}{\int\cdots\int}
        \Dot{\mathcal{M}}^{(p)}\left(
            \pmb{\varphi}^{(0)}_i
            \,\Big|\,
            \mathbf{n}^{(p-1)}_{1},\dots,\mathbf{n}^{(1)}_{1},\mathbf{n}^{(0)}_{1}
            \right) \\
        & \qquad\qquad\qquad \cdot \left\langle\Dot{\mathcal{M}}^{(q)}\left(
            \pmb{\varphi}^{(p)}_i
            \right)\right\rangle \left[\Psi^{(q)}_2 \circ \mathcal{K_{L}} \right]
        \left(
        \pmb{\varphi}^{(p)}_i\rightarrow\pmb{\varphi}^{(p)}_{r_1}
        \right)
        \frac{\mathcal{Q}(\mathbf{v}^{(p)}_i \, | \, \mathbf{n}^{(p)}_1)}{\mathcal{Q}(\mathbf{v}^{(0)}_{i} \, | \, \mathbf{n}^{(0)}_1)} \, p_{n_1}^{(p)}\left(\mathbf{n}^{(p)}_1, \dots, \mathbf{n}^{(1)}_1 \, \Bigm| \, \mathbf{n}^{(0)}_1, \mathbf{v_{i_1}}^{(0)}\right) \\
        & \qquad\qquad\qquad \cdot 
         \prod_{k=2}^{n}
        \Bigg[
        \Dot{\mathcal{S}}^{(p)}\left(
        \pmb{\varphi}^{(0)}_{i_k},
        \mathbf{n}^{(p-1)}_k,\dots,\mathbf{n}^{(0)}_k
        \mid
        \pmb{\varphi}^{(0)}_{r_{k-1}},
        \mathbf{n}^{(p-1)}_{k-1},\dots,\mathbf{n}^{(0)}_{k-1}
        \right)\, \left\langle\Dot{\mathcal{M}}^{(q)}\left(
            \pmb{\varphi}^{(p)}_{i_k}
            \right)\right\rangle 
        \\
        & \qquad\qquad\qquad \cdot \left[\Psi^{(q)}_2 \circ \mathcal{K_L} \right]
        \left(
        \pmb{\varphi}^{(p)}_{i_k}\rightarrow\pmb{\varphi}^{(p)}_{r_k}
        \right)  \, \frac{\mathcal{Q}(\mathbf{v}^{(p)}_{i_k} \, | \, \mathbf{n}^{(p)}_{k})}{\mathcal{Q}(\mathbf{v}^{(0)}_{i_k} \, | \, \mathbf{n}^{(0)}_{k})} \,
        p_{n_k}^{(p)}\left(\mathbf{n}^{(p)}_{k} \, \Bigm| \, \mathbf{n}^{(p-1)}_{k}, \mathbf{v_{i_k}}^{(p-1)}\right)
        \Bigg]  \\
        & \qquad\qquad\qquad \cdot \mathcal{S}^{(p)}\left(
        \pmb{\varphi}_{r}^{(0)}
        \mid
        \mathbf{n}^{(p-1)},\dots,\mathbf{n}^{(0)}
        \right)
        \dif\pmb{\varphi}^{(0)}_{r_1}\cdots\dif\pmb{\varphi}^{(0)}_{r_{n-1}}\,
        \dif\pmb{\varphi}^{(0)}_{i_2}\cdots\dif\pmb{\varphi}^{(0)}_{i_n}\, \\
        & \qquad\qquad\qquad \cdot 
        \dif \mathbf{n}^{(1)}_{1}\cdots\dif \mathbf{n}^{(1)}_{n}\,
        \dots \dif \mathbf{n}^{(p)}_{1}\cdots\dif \mathbf{n}^{(p)}_{n}
        \Bigg\}.
    \end{aligned}
\end{equation}
where the deterministic global normal PDF is given by $p_{n_k}^{(0)}\left(\mathbf{n}^{(0)}_k \, | \, \mathbf{v}^{(0)}_{i_k}\right) = \delta\left(\mathbf{n}^{(0)}_k - \mathbf{n_G}\right)$, with $\mathbf{n_G}$ denoting the deterministic global surface normal as in \cref{def:global_frame}. Moreover, $f_{\varphi_q}$ denotes the hazard probability associated with the unresolved reflections on $\Psi^{(q)}_2$, while the factors $\left\langle\Dot{\mathcal{M}}^{(q)}\right\rangle$ appear because each effective kernel $\Psi^{(q)}_2 \circ \mathcal{K_L}$ carries its own masking normalisation in frame $p$. Next, we substitute $\Psi^{(q)}_2 \circ \mathcal{K_L}$ with \cref{eq:scattering_operator} everywhere in the expression above, to obtain
\begin{equation} \label{eq:scattering_operator_proof_1}
\begin{aligned}
    &\left[\Psi^{(p)}_1 \circ \left( \Psi^{(q)}_2 \circ \mathcal{K_L}\right)\right]
    \left(\pmb{\varphi}^{(0)}_i \rightarrow \pmb{\varphi}^{(0)}_r \right)
    =
    \\
    &\qquad
    \frac{1}{\left\langle\Dot{\mathcal{M}}^{(p+q)}\left(
            \pmb{\varphi}^{(0)}_i
            \right)\right\rangle} \, \sum_{n=1}^{\infty}
    \Bigg\{
    \underset{\pmb{\varphi}^{(0)}_{r_k},\pmb{\varphi}^{(0)}_{i_k}}{\int\cdots\int}
    \underset{\mathbf{n}^{(0)}_k,\dots,\mathbf{n}^{(p)}_k}{\int\cdots\int}
    \Dot{\mathcal{M}}^{(p)}\left(
        \pmb{\varphi}^{(0)}_i
        \,\Big|\,
        \mathbf{n}^{(p-1)}_{1},\dots,\mathbf{n}^{(1)}_{1},\mathbf{n}^{(0)}_{1}
    \right) \, \left\langle\Dot{\mathcal{M}}^{(q)}\left(
            \pmb{\varphi}^{(p)}_i
            \right)\right\rangle
    \\
    &\qquad\cdot \Bigg[
    \frac{1}{\left\langle\Dot{\mathcal{M}}^{(q)}\left(
            \pmb{\varphi}^{(p)}_i
            \right)\right\rangle} \, \sum_{m_1=1}^{\infty}
    \Bigg\{
    \underset{\pmb{\varphi}^{(p)}_{r_{1,j}},\pmb{\varphi}^{(p)}_{i_{1,j}}}{\int\cdots\int}
    \underset{\mathbf{n}^{(p)}_{1,j},\dots,\mathbf{n}^{(p+q)}_{1,j}}{\int\cdots\int}
    \Dot{\mathcal{M}}^{(q)}\left(
        \pmb{\varphi}^{(p)}_i
        \,\Big|\,
        \mathbf{n}^{(p+q-1)}_{1},\dots,\mathbf{n}^{(p+1)}_{1},\mathbf{n}^{(p)}_{1}
    \right) \, \left\langle\Dot{\mathcal{M}}_L\left(
            \pmb{\varphi}^{(p+q)}_i
            \right)\right\rangle  \\
    & \qquad\cdot \mathcal{K_L}\Big(
    \pmb{\varphi}^{(p+q)}_{i}\rightarrow\pmb{\varphi}^{(p+q)}_{r_{1,1}}
    \Big) \,\frac{\mathcal{Q}(\mathbf{v}^{(p+q)}_{i} \, | \, \mathbf{n}^{(p+q)}_{1,1})}
    {\mathcal{Q}(\mathbf{v}^{(p)}_{i} \, | \, \mathbf{n}^{(p)}_{1,1})}
    p_{n_{1,1}}^{(q)}\Big(
    \mathbf{n}^{(p+q)}_{1,1} \, \Bigm| \, \mathbf{n}^{(p+q-1)}_{1,1}, \mathbf{v_{i_{1,1}}}^{(p+q-1)}
    \Big) \, p_{n_{1,1}}^{(q)}\Big(
    \mathbf{n}^{(p+q-1 : p+1)}_{1,1} \, \Bigm| \, \mathbf{n}^{(p)}_{1,1}, \mathbf{v_{i_{1,1}}}^{(p)}
    \Big)
    \\
    &\qquad\cdot
    \prod_{j=2}^{m_1}
    \Bigg[
    \Dot{\mathcal{S}}^{(q)}\Big(
    \pmb{\varphi}^{(p)}_{i_{1,j}},
    \mathbf{n}^{(p+q-1)}_{1,j},\dots,\mathbf{n}^{(p)}_{1,j}
    \mid
    \pmb{\varphi}^{(p)}_{r_{1,j-1}},
    \mathbf{n}^{(p+q-1)}_{1,j-1},\dots,\mathbf{n}^{(p)}_{1,j-1}
    \Big)\, \left\langle\Dot{\mathcal{M}}_L\left(
            \pmb{\varphi}^{(p+q)}_{i_{1,j}}
            \right)\right\rangle \, \mathcal{K_L}\Big(
    \pmb{\varphi}^{(p+q)}_{i_{1,j}}\rightarrow\pmb{\varphi}^{(p+q)}_{r_{1,j}}
    \Big)
    \\
    &\qquad\qquad\cdot
    \frac{\mathcal{Q}(\mathbf{v}^{(p+q)}_{i_{1,j}} \, | \, \mathbf{n}^{(p+q)}_{1,j})}
    {\mathcal{Q}(\mathbf{v}^{(p)}_{i_{1,j}} \, | \, \mathbf{n}^{(p)}_{1,j})} \, p_{n_{1,j}}^{(q)}\Big(
    \mathbf{n}^{(p+q)}_{1,j} \, \Bigm| \, \mathbf{n}^{(p+q-1)}_{1,j}, \mathbf{v_{i_{1,j}}}^{(p+q-1)}
    \Big)
    \Bigg]
    \mathcal{S}^{(q)}\Big(
    \pmb{\varphi}^{(p)}_{r_1} \, \Bigm| \, \mathbf{n}^{(p+q-1)}_{1,m_1},\dots,\mathbf{n}^{(p)}_{1,m_1}
    \Big)
    \\
    &\qquad\cdot
    \dif\pmb{\varphi}^{(p)}_{r_{1,1}}\cdots\dif\pmb{\varphi}^{(p)}_{r_{1,m_1-1}}
    \dif\pmb{\varphi}^{(p)}_{i_{1,2}}\cdots\dif\pmb{\varphi}^{(p)}_{i_{1,m_1}}
    \dif \mathbf{n}^{(p)}_{1,1}\cdots\dif \mathbf{n}^{(p)}_{1,m_1}
    \cdots
    \dif \mathbf{n}^{(p+q)}_{1,1}\cdots\dif \mathbf{n}^{(p+q)}_{1,m_1}
    \Bigg\}
    \Bigg]
    \\
    &\qquad\cdot
    \frac{\mathcal{Q}(\mathbf{v}^{(p)}_{i} \, | \, \mathbf{n}^{(p)}_{1})}
    {\mathcal{Q}(\mathbf{v}^{(0)}_{i} \, | \, \mathbf{n}^{(0)}_{1})}
    p_{n_1}^{(p)}\left(\mathbf{n}^{(p:1)}_{1} \, \Bigm| \, \mathbf{n}^{(0)}_{1}, \mathbf{v_{i_1}}^{(0)}\right)
    \prod_{k=2}^{n}
    \Bigg[
    \Dot{\mathcal{S}}^{(p)}\Big(
    \pmb{\varphi}^{(0)}_{i_k},
    \mathbf{n}^{(p-1)}_k,\dots,\mathbf{n}^{(0)}_{k}
    \mid
    \pmb{\varphi}^{(0)}_{r_{k-1}},
    \mathbf{n}^{(p-1)}_{k-1},\dots,\mathbf{n}^{(0)}_{k-1}
    \Big) \left\langle\Dot{\mathcal{M}}^{(q)}\left(
            \pmb{\varphi}^{(p)}_{i_k}
            \right)\right\rangle
    \\
    &\qquad\cdot
    \Bigg[\frac{1}{\left\langle\Dot{\mathcal{M}}^{(q)}\left(
            \pmb{\varphi}^{(p)}_{i_k}
            \right)\right\rangle}
    \sum_{m_k=1}^{\infty}
    \Bigg\{
    \underset{\pmb{\varphi}^{(p)}_{r_{k,j}},\pmb{\varphi}^{(p)}_{i_{k,j}}}{\int\cdots\int}
    \underset{\mathbf{n}^{(p)}_{k,j},\dots,\mathbf{n}^{(p+q)}_{k,j}}{\int\cdots\int}
    \Dot{\mathcal{M}}^{(q)}\left(
        \pmb{\varphi}^{(p)}_{i_k}
        \,\Big|\,
        \mathbf{n}^{(p+q-1)}_{k},\dots,\mathbf{n}^{(p+1)}_{k},\mathbf{n}^{(p)}_{k}
    \right) \, \left\langle\Dot{\mathcal{M}}_L\left(
            \pmb{\varphi}^{(p+q)}_{i_{k,1}}
            \right)\right\rangle \\
    & \qquad \cdot
    \mathcal{K_L}\Big(
    \pmb{\varphi}^{(p+q)}_{i_k}\rightarrow\pmb{\varphi}^{(p+q)}_{r_{k,1}}
    \Big) \, \frac{\mathcal{Q}(\mathbf{v}^{(p+q)}_{i_k} \, | \, \mathbf{n}^{(p+q)}_{k,1})}
    {\mathcal{Q}(\mathbf{v}^{(p)}_{i_k} \, | \, \mathbf{n}^{(p)}_{k,1})} \, p_{n_{k,1}}^{(q)}\Big(
    \mathbf{n}^{(p+q)}_{k,1} \, \Bigm| \, \mathbf{n}^{(p+q-1)}_{k,1}, \mathbf{v_{i_{k,1}}}^{(p+q-1)}
    \Big) \, p_{n_{k,1}}^{(q)}\Big(
    \mathbf{n}^{(p+q-1:p+1)}_{k,1} \, \Bigm| \, \mathbf{n}^{(p)}_{k,1}, \mathbf{v_{i_{k,1}}}^{(p)}
    \Big)
    \\
    &\qquad\cdot
    \prod_{j=2}^{m_k}
    \Bigg[
    \Dot{\mathcal{S}}^{(q)}\Big(
    \pmb{\varphi}^{(p)}_{i_{k,j}},
    \mathbf{n}^{(p+q-1)}_{k,j},\dots,\mathbf{n}^{(p)}_{k,j}
    \mid
    \pmb{\varphi}^{(p)}_{r_{k,j-1}},
    \mathbf{n}^{(p+q-1)}_{k,j-1},\dots,\mathbf{n}^{(p)}_{k,j-1}
    \Big)\,\left\langle\Dot{\mathcal{M}}_L\left(
            \pmb{\varphi}^{(p+q)}_{i_{k,j}}
            \right)\right\rangle\, \mathcal{K_L}\Big(
    \pmb{\varphi}^{(p+q)}_{i_{k,j}}\rightarrow\pmb{\varphi}^{(p+q)}_{r_{k,j}}
    \Big)
    \\
    &\qquad\qquad\cdot
    \frac{\mathcal{Q}(\mathbf{v}^{(p+q)}_{i_{k,j}} \, | \, \mathbf{n}^{(p+q)}_{k,j})}
    {\mathcal{Q}(\mathbf{v}^{(p)}_{i_{k,j}} \, | \, \mathbf{n}^{(p)}_{k,j})}\, p_{n_{k,j}}^{(q)}\Big(
    \mathbf{n}^{(p+q)}_{k,j} \, \Bigm| \, \mathbf{n}^{(p+q-1)}_{k,j}, \mathbf{v_{i_{k,j}}}^{(p+q-1)}
    \Big)
    \Bigg]
    \mathcal{S}^{(q)}\Big(
    \pmb{\varphi}^{(p)}_{r_k} \, \Bigm| \, \mathbf{n}^{(p+q-1)}_{k,m_k},\dots,\mathbf{n}^{(p)}_{k,m_k}
    \Big)
    \\
    &\qquad\cdot
    \dif\pmb{\varphi}^{(p)}_{r_{k,1}}\cdots\dif\pmb{\varphi}^{(p)}_{r_{k,m_k-1}}
    \dif\pmb{\varphi}^{(p)}_{i_{k,2}}\cdots\dif\pmb{\varphi}^{(p)}_{i_{k,m_k}}
    \dif \mathbf{n}^{(p)}_{k,1}\cdots\dif \mathbf{n}^{(p)}_{k,m_k}
    \cdots
    \dif \mathbf{n}^{(p+q)}_{k,1}\cdots\dif \mathbf{n}^{(p+q)}_{k,m_k}
    \Bigg\}
    \Bigg]
    \\
    &\qquad\cdot
    \frac{\mathcal{Q}(\mathbf{v}^{(p)}_{i_k} \, | \, \mathbf{n}^{(p)}_{k})}
    {\mathcal{Q}(\mathbf{v}^{(0)}_{i_k} \, | \, \mathbf{n}^{(0)}_{k})}
    p_{n_k}^{(p)}\left(\mathbf{n}^{(p)}_{k} \, \Bigm| \, \mathbf{n}^{(p-1)}_{k}, \mathbf{v_{i_k}}^{(p-1)}\right)
    \Bigg]
    \mathcal{S}^{(p)}\left(
    \pmb{\varphi}^{(0)}_r
    \,\Bigm|\,
    \mathbf{n}^{(p-1)}_n,\dots,\mathbf{n}^{(0)}_n
    \right)
    \\
    &\qquad\cdot
    \dif\pmb{\varphi}^{(0)}_{r_1}\cdots\dif\pmb{\varphi}^{(0)}_{r_{n-1}}
    \dif\pmb{\varphi}^{(0)}_{i_2}\cdots\dif\pmb{\varphi}^{(0)}_{i_n}
    \dif \mathbf{n}^{(1)}_{1}\cdots\dif \mathbf{n}^{(1)}_{n}
    \cdots
    \dif \mathbf{n}^{(p)}_{1}\cdots\dif \mathbf{n}^{(p)}_{n}
    \Bigg\},
\end{aligned}
\end{equation}
where we introduced the short-hand notation of $p\Bigl(a^{(p+q:p)} \, | \, b\Bigr) = p\Bigl(a^{(p+q)}, a^{(p+q-1)}, \dots, a^{(p)} \, | \, b\Bigr)$. We notice that the above expression can be written as a nested sum, i.e.
\begin{equation}
\begin{aligned}
    \left[\Psi^{(p)}_1 \circ \left( \Psi^{(q)}_2 \circ \mathcal{K_L}\right)\right]
    \left(\pmb{\varphi}^{(0)}_i \rightarrow \pmb{\varphi}^{(0)}_r \right)
    =
    \sum_{n=1}^{\infty}
    \left(\sum_{m_1=1}^{\infty}
    \cdots
    \sum_{m_n=1}^{\infty}
    \mathcal{K_K}_{n,m_1,\dots,m_n}
    \left(\pmb{\varphi}^{(0)}_i \rightarrow \pmb{\varphi}^{(0)}_r \right)\right),
\end{aligned}
\end{equation}
where $\mathcal{K_K}_{n,m_1,\dots,m_n}\left(\pmb{\varphi}^{(0)}_i \rightarrow \pmb{\varphi}^{(0)}_r\right)$ represents the multi-reflection kernel contribution with exactly $n$ reflections on $\Psi^{(p)}_1$, and $m_k$ unresolved reflections on $\Psi^{(q)}_2$ during the $k^{th}$ reflection on $\Psi^{(p)}_1$, with $k \in \overline{1,n}$. To prove the theorem, we must rewrite this nested contribution into a single-sum form, i.e.
\begin{equation}
    \left[\Psi^{(p)}_1 \circ \left( \Psi^{(q)}_2 \circ \mathcal{K_L}\right)\right]
    \left(\pmb{\varphi}^{(0)}_i \rightarrow \pmb{\varphi}^{(0)}_r \right)
    = \sum_{n=1}^{\infty} \mathcal{K_{K}}_{n}\left(\pmb{\varphi}^{(0)}_i \rightarrow \pmb{\varphi}^{(0)}_r \right),
\end{equation}
where the kernel contribution with $n$ total reflections must take the same structural form as in \cref{eq:scattering_operator}, but with the composed scale hierarchy running from frame $0$ to frame $p+q$. To do so, we first use the fact that
$p_{n_{k,1}}^{(p)}\Big(\mathbf{n}^{(p)}_{k,1} \, | \, \mathbf{v}^{(0)}_{i_{k,1}}\Big) = \delta\Bigl(\mathbf{n}^{(p)}_{k,1} - \mathbf{n}^{(p)}_{k}\Bigr)$ to collapse the $\mathcal{Q}$ ratios as
\begin{equation}
    \frac{\mathcal{Q}\Bigl(\mathbf{v}^{(p+q)}_{i_{k,j}} \, \Big| \, \mathbf{n}^{(p+q)}_{k,j}\Bigr)}
    {\mathcal{Q}\Bigl(\mathbf{v}^{(p)}_{i_{k,j}} \, \Big| \, \mathbf{n}^{(p)}_{k,j}\Bigr)} \,
    \frac{\mathcal{Q}\Bigl(\mathbf{v}^{(p)}_{i_k} \, \Big| \, \mathbf{n}^{(p)}_{k}\Bigr)}
    {\mathcal{Q}\Bigl(\mathbf{v}^{(0)}_{i_k} \, \Big| \, \mathbf{n}^{(0)}_{k}\Bigr)}
    =
    \frac{\mathcal{Q}\Bigl(\mathbf{v}^{(p+q)}_{i_{k,j}} \, \Big| \, \mathbf{n}^{(p+q)}_{k,j}\Bigr)}
    {\mathcal{Q}\Bigl(\mathbf{v}^{(p)}_{i_k} \, \Big| \, \mathbf{n}^{(p)}_{k}\Bigr)} \,
    \frac{\mathcal{Q}\Bigl(\mathbf{v}^{(p)}_{i_k} \, \Big| \, \mathbf{n}^{(p)}_{k}\Bigr)}
    {\mathcal{Q}\Bigl(\mathbf{v}^{(0)}_{i_k} \, \Big| \, \mathbf{n}^{(0)}_{k}\Bigr)}
    =
    \frac{\mathcal{Q}\Bigl(\mathbf{v}^{(p+q)}_{i_{k,j}} \, \Big| \, \mathbf{n}^{(p+q)}_{k,j}\Bigr)}
    {\mathcal{Q}\Bigl(\mathbf{v}^{(0)}_{i_k} \, \Big| \, \mathbf{n}^{(0)}_{k}\Bigr)}.
\end{equation}
Furthermore, we immediately see that the normal-averaged masking functions simplify, leaving only the outer factor in the expression. Next, after re-indexing by the total number of local reflections $N = M_n$, where $M_k = \sum_{s=1}^{k} m_s$, we see that the nested summands in \cref{eq:scattering_operator_proof_1} contain all possible transition paths by which a particle initially in state $\pmb{\varphi}_i^{(0)}$ may reach the final state $\pmb{\varphi}_r^{(0)}$ after $N$ reflections. Indeed, between any two successive flattened collision points, $l-1$ and $l$, there are exactly two possible transitions: either the particle remains within the unresolved scale hierarchy of $\Psi^{(q)}_2$, or it exits $\Psi^{(q)}_2$, propagates through the scale hierarchy of $\Psi^{(p)}_1$, and re-enters $\Psi^{(q)}_2$ at the subsequent collision point. Hence, we define the local normal PDFs at reflection $l$ as
\begin{equation}
    \label{eq:flattened_normal_pdf}
    p_{n_l}^{(p+q)}\Big(
    \mathbf{n}^{(p+q)}_{l}
    \,\Bigm|\, 
    \mathbf{n}^{(p+q-1)}_{l}, \mathbf{v}^{(p+q-1)}_{i_l}
    \Big)
    =
    \begin{cases}
        p_{n_{k,1}}^{(q)}\Big(
        \mathbf{n}^{(p+q)}_{k,1}
        \,\Bigm|\, 
        \mathbf{n}^{(p+q-1)}_{k,1}, \mathbf{v}^{(p+q-1)}_{i_{k,1}}
        \Big), 
        & l = M_{k-1} + 1,
        \\[1em]
        p_{n_{k,l - M_{k-1}}}^{(q)}\Big(
        \mathbf{n}^{(p+q)}_{k,l - M_{k-1}}
        \,\Bigm|\, 
        \mathbf{n}^{(p+q-1)}_{k,l-M_{k-1}}, \mathbf{v}^{(p+q-1)}_{i_{k,l - M_{k-1}}}
        \Big),
        & M_{k-1} + 1 < l \le M_k.
    \end{cases}
\end{equation}
We further define 
\begin{multline}
    p_{n_l}^{(p+q-1)}\Bigl(
        \mathbf{n}^{(p+q-1:p+1)}_l \, \Big| \, \mathbf{n}^{(p)}_l, \mathbf{v_{i_l}}^{(p)}
        \Bigr) \, p_{n_l}^{(p)}\Bigl(
        \mathbf{n}^{(p)}_l \, \Big| \, \mathbf{n}^{(p-1)}_l, \mathbf{v_{i_l}}^{(p-1)}
        \Bigr) = \\ p_{n_{k,1}}^{(p+q-1)}\Bigl(
        \mathbf{n}^{(p+q-1:p+1)}_{k,1} \, \Big| \, \mathbf{n}^{(p)}_{k,1}, \mathbf{v_{i_l}}^{(p)}
        \Bigr) \, p_{n_k}^{(p)}\Bigl(
        \mathbf{n}^{(p)}_k \, \Big| \, \mathbf{n}^{(p-1)}_k, \mathbf{v_{i_l}}^{(p-1)}
        \Bigr)
\end{multline}
and the incident and reflected states, $\pmb{\varphi}^{(p+q)}_{i_l}$ and $\pmb{\varphi}^{(p+q)}_{r_l}$, as
\begin{equation}
    \pmb{\varphi}^{(p+q)}_{i_l} =
    \begin{cases}
        \pmb{\varphi}^{(p+q)}_{i_{k,1}}, & l = M_{k-1} + 1, \\
        \pmb{\varphi}^{(p+q)}_{i_{k,l-M_{k-1}}}, & M_{k-1} + 1 < l \leq M_k,
    \end{cases}
\end{equation}
\begin{equation}
    \pmb{\varphi}^{(p+q)}_{r_l} =
    \begin{cases}
        \pmb{\varphi}^{(p+q)}_{r_{k,1}}, & l = M_{k-1} + 1, \\
        \pmb{\varphi}^{(p+q)}_{r_{k,l-M_{k-1}}}, & M_{k-1} + 1 < l \leq M_k.
    \end{cases}
\end{equation}
Making these substitutions in \cref{eq:scattering_operator_proof_1}, we obtain
\begin{equation} \label{eq:scattering_operator_proof_2}
    \begin{aligned}
        &\left[\Psi^{(p)}_1 \circ \left( \Psi^{(q)}_2 \circ \mathcal{K_L}\right)\right]
        \left(\pmb{\varphi}^{(0)}_i \rightarrow \pmb{\varphi}^{(0)}_r \right)
        =
        \\
        &\qquad
        \frac{1}{\left\langle\Dot{\mathcal{M}}^{(p+q)}\left(
            \pmb{\varphi}^{(0)}_i
            \right)\right\rangle}
        \sum_{N=1}^{\infty}
        \Bigg\{
        \underset{\pmb{\varphi}^{(0)}_{r_l},\pmb{\varphi}^{(0)}_{i_l}}{\int\cdots\int}
        \underset{\mathbf{n}^{(1)}_l,\dots,\mathbf{n}^{(p+q)}_l}{\int\cdots\int}
        \Dot{\mathcal{M}}\left(
            \pmb{\varphi}^{(0)}_i
            \,\Big|\,
            \mathbf{n}^{(p+q-1)}_1,\dots,\mathbf{n}^{(0)}_1
        \right)
        \left\langle\Dot{\mathcal{M}}_L\left(
            \pmb{\varphi}^{(p+q)}_{i}
            \right)\right\rangle\, \mathcal{K_L}
        \left(
        \pmb{\varphi}^{(p+q)}_i\rightarrow\pmb{\varphi}^{(p+q)}_{r_1}
        \right)
        \\
        &\qquad\cdot
        \frac{\mathcal{Q}(\mathbf{v}^{(p+q)}_{i} \, | \, \mathbf{n}^{(p+q)}_1)}
        {\mathcal{Q}(\mathbf{v}^{(0)}_{i} \, | \, \mathbf{n}^{(0)}_1)} \, p_{n_1}^{(p+q)}\left(\mathbf{n}^{(p+q:1)}_1 \, \Bigm| \, \mathbf{n}^{(0)}_1, \mathbf{v_{i_1}}^{(0)}\right) \, \prod_{l=2}^{N}
        \Bigg[
        \Dot{\mathcal{S}}\left(
        \pmb{\varphi}^{(0)}_{i_l},
        \mathbf{n}^{(p+q-1)}_l,\dots,\mathbf{n}^{(0)}_l
        \mid
        \pmb{\varphi}^{(0)}_{r_{l-1}},
        \mathbf{n}^{(p+q-1)}_{l-1},\dots,\mathbf{n}^{(0)}_{l-1}
        \right)
        \\
        &\qquad\cdot
         \left\langle\Dot{\mathcal{M}}_L\left(
            \pmb{\varphi}^{(p+q)}_{i_{l}}
            \right)\right\rangle\,\mathcal{K_L}
        \left(
        \pmb{\varphi}^{(p+q)}_{i_l}\rightarrow\pmb{\varphi}^{(p+q)}_{r_l}
        \right) \, \frac{\mathcal{Q}(\mathbf{v}^{(p+q)}_{i_l} \, | \, \mathbf{n}^{(p+q)}_{l})}
        {\mathcal{Q}(\mathbf{v}^{(0)}_{i_l} \, | \, \mathbf{n}^{(0)}_{l})}\, p_{n_l}^{(p+q)}\left(\mathbf{n}^{(p+q)}_{l} \, \Bigm| \, \mathbf{n}^{(p+q-1)}_{l}, \mathbf{v_{i_l}}^{(p+q-1)}\right)
        \Bigg] 
        \\
        &\qquad\cdot 
        \mathcal{S}\left(
        \pmb{\varphi}^{(0)}_r
        \,\Big|\,
        \mathbf{n}^{(p+q-1)}_N,\dots,\mathbf{n}^{(0)}_N
        \right)
        \dif\pmb{\varphi}^{(0)}_{r_1}\cdots\dif\pmb{\varphi}^{(0)}_{r_{N-1}}\,
        \dif\pmb{\varphi}^{(0)}_{i_2}\cdots\dif\pmb{\varphi}^{(0)}_{i_N}\,
        \dif \mathbf{n}^{(0)}_{1}\cdots\dif \mathbf{n}^{(0)}_{N}
        \cdots
        \dif \mathbf{n}^{(p+q)}_{1}\cdots\dif \mathbf{n}^{(p+q)}_{N}
        \Bigg\},
    \end{aligned}
\end{equation}
where the composite shadowing function $\mathcal{S}\left(\pmb{\varphi}^{(0)}_{i_l},\mathbf{n}^{(p+q-1)}_l,\dots,\mathbf{n}^{(0)}_l \mid \pmb{\varphi}^{(0)}_{r_{l-1}},\mathbf{n}^{(p+q-1)}_{l-1},\dots,\mathbf{n}^{(0)}_{l-1}\right)$ takes the form
\begin{equation} \label{eq:multi_scale_shadowing_composition_proof}
    \begin{aligned}
        &\mathcal{S}\Bigl(
        \pmb{\varphi}^{(0)}_{i_l},
        \mathbf{n}^{(p+q-1)}_l,\dots,\mathbf{n}^{(0)}_l
        \mid
        \pmb{\varphi}^{(0)}_{r_{l-1}},
        \mathbf{n}^{(p+q-1)}_{l-1},\dots,\mathbf{n}^{(0)}_{l-1}
        \Bigr)
        =
        \\
        &\qquad
        \mathcal{S}^{(q)}\Bigl(
        \pmb{\varphi}^{(p)}_{i_l},
        \mathbf{n}^{(p+q-1)}_l,\dots,\mathbf{n}^{(p)}_l
        \mid
        \pmb{\varphi}^{(p)}_{r_{l-1}},
        \mathbf{n}^{(p+q-1)}_{l-1},\dots,\mathbf{n}^{(p)}_{l-1}
        \Bigr)
        \prod_{s=0}^{p}
        \delta\Bigl(
        \mathbf{n}^{(s)}_l - \mathbf{n}^{(s)}_{l-1}
        \Bigr)
        \\
        &\qquad+
        \mathcal{S}^{(q)}\Bigl(
        \pmb{\varphi}^{(p)}_{r_{l-1}}
        \,\Big|\,
        \mathbf{n}^{(p+q-1)}_{l-1},\dots,\mathbf{n}^{(p)}_{l-1}
        \Bigr)
        \mathcal{S}^{(p)}\Bigl(
        \pmb{\varphi}^{(0)}_{i_l},
        \mathbf{n}^{(p-1)}_l,\dots,\mathbf{n}^{(0)}_l
        \mid
        \pmb{\varphi}^{(0)}_{r_{l-1}},
        \mathbf{n}^{(p-1)}_{l-1},\dots,\mathbf{n}^{(0)}_{l-1}
        \Bigr)
        \\
        &\qquad\cdot
        \mathcal{M}^{(q)}\Bigl(
        \pmb{\varphi}^{(p)}_{i_l}
        \,\Big|\,
        \mathbf{n}^{(p+q-1)}_l,\dots,\mathbf{n}^{(p)}_l
        \Bigr) \,
        p_{n_l}^{(p+q-1)}\Bigl(
        \mathbf{n}^{(p+q-1:p+1)}_l \, \Big| \, \mathbf{n}^{(p)}_l, \mathbf{v}_{i_l}^{(p)}
        \Bigr) \, p_{n_l}^{(p)}\Bigl(
        \mathbf{n}^{(p)}_l \, \Big| \, \mathbf{n}^{(p-1)}_l, \mathbf{v}_{i_l}^{(p-1)}
        \Bigr).
    \end{aligned}
\end{equation}
Our objective is now to show that the above equation is, in fact, $\mathcal{S}^{(p+q)}$, as given by \cref{def:shadowing_multi_scale}. Furthermore, for a scale $(\alpha)$ and reflection $l$, we denote the single-scale one-point shadowing function $\mathcal{S}\left(\pmb{\varphi}^{(\alpha-1)}_{r_{l-1}} \, | \, \mathbf{n}^{(\alpha-1)}_{l-1}\right)$ as $\mathcal{S}^{1P}_\alpha$, the single-scale, two-point shadowing function $\mathcal{S}\left(\pmb{\varphi}^{(\alpha-1)}_{i_l} \, | \,  \pmb{\varphi}^{(\alpha-1)}_{r_{l-1}}, \mathbf{n}^{(\alpha-1)}_{l-1}\right)$ as $\mathcal{S}^{2P}_\alpha$, the single-scale one-point masking function $\mathcal{M}\left(\pmb{\varphi}^{(\alpha-1)}_{i_{l}} \, | \, \mathbf{n}^{(\alpha-1)}_{l-1}\right)$ as $\mathcal{M}^{1P}_\alpha$, $p_{n_l}^{(\alpha)}\left(\mathbf{n}^{(\alpha)}_l \, | \, \mathbf{n}^{(\alpha-1)}_l \right)$ as $p_{n_l}^{(\alpha)}$, and $\delta\left(\mathbf{n}^{(\alpha)}_l - \mathbf{n}^{(\alpha)}_{l-1} \right)$ as $\delta_{n_l}^{(\alpha)}$. Then, we may write
\begin{multline}
    \mathcal{S}^{(q)}\Bigl(
        \pmb{\varphi}^{(p)}_{i_l},
        \mathbf{n}^{(p+q-1)}_l,\dots,\mathbf{n}^{(p)}_l
        \mid
        \pmb{\varphi}^{(p)}_{r_{l-1}},
        \mathbf{n}^{(p+q-1)}_{l-1},\dots,\mathbf{n}^{(p)}_{l-1}
        \Bigr) \\ = \sum_{j=p+1}^{p+q}\left\{\left[ \prod_{\alpha=j+1}^{p+q} \mathcal{S}^{1P}_\alpha\right] \, \mathcal{S}^{2P}_j \left[\prod_{\alpha=j+1}^{p+q} \mathcal{M}^{1P}_\alpha p_{n_l}^{(\alpha-1)} \right] \, \left[\prod_{\alpha=p+1}^{j-1} \delta_{n_l}^{(\alpha)} \right] \right\},
\end{multline}
\begin{multline}
    \mathcal{S}^{(p)}\Bigl(
        \pmb{\varphi}^{(0)}_{i_l},
        \mathbf{n}^{(p-1)}_l,\dots,\mathbf{n}^{(0)}_l
        \mid
        \pmb{\varphi}^{(0)}_{r_{l-1}},
        \mathbf{n}^{(p-1)}_{l-1},\dots,\mathbf{n}^{(0)}_{l-1}
        \Bigr) \\ = \sum_{j=1}^{p}\left\{\left[ \prod_{\alpha=j+1}^{p} \mathcal{S}^{1P}_\alpha\right] \, \mathcal{S}^{2P}_j \left[\prod_{\alpha=j+1}^{p} \mathcal{M}^{1P}_\alpha p_{n_l}^{(\alpha-1)} \right] \, \left[\prod_{\alpha=1}^{j-1} \delta_{n_l}^{(\alpha)} \right] \right\},
\end{multline}
\begin{equation}
    \mathcal{S}^{(q)}\Bigl(
        \pmb{\varphi}^{(p)}_{r_{l-1}}
        \,\Big|\,
        \mathbf{n}^{(p+q-1)}_{l-1},\dots,\mathbf{n}^{(p)}_{l-1}
        \Bigr) = \prod_{\alpha=p+1}^{p+q} \mathcal{S}^{1P}_\alpha \text{ and } \mathcal{M}^{(q)}\Bigl(
        \pmb{\varphi}^{(p)}_{i_{l}}
        \,\Big|\,
        \mathbf{n}^{(p+q-1)}_{l-1},\dots,\mathbf{n}^{(p)}_{l-1}
        \Bigr) = \prod_{\alpha=p+1}^{p+q} \mathcal{M}^{1P}_\alpha.
\end{equation}
Substituting the above equations into \cref{eq:multi_scale_shadowing_composition_proof}, we obtain
\begin{equation}
    \begin{aligned}
        &\mathcal{S}\Bigl(
            \pmb{\varphi}^{(0)}_{i_l},
            \mathbf{n}^{(p+q-1)}_l,\dots,\mathbf{n}^{(0)}_l
            \mid
            \pmb{\varphi}^{(0)}_{r_{l-1}},
            \mathbf{n}^{(p+q-1)}_{l-1},\dots,\mathbf{n}^{(0)}_{l-1}
        \Bigr)
        \\
        &\quad =
        \sum_{j=p+1}^{p+q}
        \Biggl\{
            \Biggl[
                \prod_{\alpha=j+1}^{p+q}
                \mathcal{S}^{1P}_{\alpha}
            \Biggr]
            \mathcal{S}^{2P}_{j}
            \Biggl[
                \prod_{\alpha=j+1}^{p+q}
                \mathcal{M}^{1P}_{\alpha}
                p_{n_l}^{(\alpha-1)}
            \Biggr]
            \Biggl[
                \prod_{\alpha=p+1}^{j-1}
                \delta_{n_l}^{(\alpha)}
            \Biggr]
        \Biggr\}
        \Biggl[
            \prod_{\alpha=1}^{p}
            \delta_{n_l}^{(\alpha)}
        \Biggr]
        \\
        &\qquad\qquad+
        \Biggl[
            \prod_{\alpha=p+1}^{p+q}
            \mathcal{S}^{1P}_{\alpha}
        \Biggr]
        \sum_{j=1}^{p}
        \Biggl\{
            \Biggl[
                \prod_{\alpha=j+1}^{p}
                \mathcal{S}^{1P}_{\alpha}
            \Biggr]
            \mathcal{S}^{2P}_{j}
            \Biggl[
                \prod_{\alpha=j+1}^{p}
                \mathcal{M}^{1P}_{\alpha}
                p_{n_l}^{(\alpha-1)}
            \Biggr]
            \Biggl[
                \prod_{\alpha=1}^{j-1}
                \delta_{n_l}^{(\alpha)}
            \Biggr]
        \Biggr\}
        \Biggl[
            \prod_{\alpha=p+1}^{p+q}
            \mathcal{M}^{1P}_{\alpha}
        \Biggr]
        \Biggl[
            \prod_{\alpha=p+1}^{p+q}
            p_{n_l}^{(\alpha-1)}
        \Biggr]
        \\
        &\quad =
        \sum_{j=p+1}^{p+q}
        \Biggl\{
            \Biggl[
                \prod_{\alpha=j+1}^{p+q}
                \mathcal{S}^{1P}_{\alpha}
            \Biggr]
            \mathcal{S}^{2P}_{j}
            \Biggl[
                \prod_{\alpha=j+1}^{p+q}
                \mathcal{M}^{1P}_{\alpha}
                p_{n_l}^{(\alpha-1)}
            \Biggr]
            \Biggl[
                \prod_{\alpha=1}^{j-1}
                \delta_{n_l}^{(\alpha)}
            \Biggr]
        \Biggr\}
        \\
        &\qquad\qquad+
        \sum_{j=1}^{p}
        \Biggl\{
            \Biggl[
                \prod_{\alpha=j+1}^{p+q}
                \mathcal{S}^{1P}_{\alpha}
            \Biggr]
            \mathcal{S}^{2P}_{j}
            \Biggl[
                \prod_{\alpha=j+1}^{p+q}
                \mathcal{M}^{1P}_{\alpha}
                p_{n_l}^{(\alpha-1)}
            \Biggr]
            \Biggl[
                \prod_{\alpha=1}^{j-1}
                \delta_{n_l}^{(\alpha)}
            \Biggr]
        \Biggr\}
        \\
        &\quad =
        \sum_{j=1}^{p+q}
        \Biggl\{
            \Biggl[
                \prod_{\alpha=j+1}^{p+q}
                \mathcal{S}^{1P}_{\alpha}
            \Biggr]
            \mathcal{S}^{2P}_{j}
            \Biggl[
                \prod_{\alpha=j+1}^{p+q}
                \mathcal{M}^{1P}_{\alpha}
                p_{n_l}^{(\alpha-1)}
            \Biggr]
            \Biggl[
                \prod_{\alpha=1}^{j-1}
                \delta_{n_l}^{(\alpha)}
            \Biggr]
        \Biggr\}
        \\
        &\quad =
        \mathcal{S}^{(p+q)}\Bigl(
            \pmb{\varphi}^{(0)}_{i_l},
            \mathbf{n}^{(p+q-1)}_l,\dots,\mathbf{n}^{(0)}_l
            \mid
            \pmb{\varphi}^{(0)}_{r_{l-1}},
            \mathbf{n}^{(p+q-1)}_{l-1},\dots,\mathbf{n}^{(0)}_{l-1}
        \Bigr).
    \end{aligned}
\end{equation}
Further, it is trivial to show that the composite shadowing and masking functions become
\begin{equation}
    \mathcal{M}\left(
            \pmb{\varphi}^{(0)}_i
            \,\Big|\,
            \mathbf{n}^{(p+q-1)}_1,\dots,\mathbf{n}^{(0)}_1
        \right) = \prod_{\alpha=1}^{p} \mathcal{M}^{1P}_\alpha \, \prod_{\alpha=p+1}^{p+q} \mathcal{M}^{1P}_\alpha = \prod_{\alpha=1}^{p+q} \mathcal{M}^{1P}_\alpha = \mathcal{M}^{(p+q)}\left(
            \pmb{\varphi}^{(0)}_i
            \,\Big|\,
            \mathbf{n}^{(p+q-1)}_1,\dots,\mathbf{n}^{(0)}_1
        \right)
\end{equation}
and
\begin{equation}
    \mathcal{S}\left(
            \pmb{\varphi}^{(0)}_{r}
            \,\Big|\,
            \mathbf{n}^{(p+q-1)}_1,\dots,\mathbf{n}^{(0)}_1
        \right) = \prod_{\alpha=1}^{p} \mathcal{S}^{1P}_\alpha \, \prod_{\alpha=p+1}^{p+q} \mathcal{S}^{1P}_\alpha = \prod_{\alpha=1}^{p+q} \mathcal{S}^{1P}_\alpha = \mathcal{S}^{(p+q)}\left(
            \pmb{\varphi}^{(0)}_r
            \,\Big|\,
            \mathbf{n}^{(p+q-1)}_1,\dots,\mathbf{n}^{(0)}_1
        \right),
\end{equation}
respectively. Hence, \cref{eq:scattering_operator_proof_2} becomes
\begin{equation} \label{eq:scattering_operator_proof_3}
    \begin{aligned}
        &\left[\Psi^{(p)}_1 \circ \left( \Psi^{(q)}_2 \circ \mathcal{K_L}\right)\right]
        \left(\pmb{\varphi}^{(0)}_i \rightarrow \pmb{\varphi}^{(0)}_r \right)
        =
        \\
        &\qquad
        =\frac{1}{\left\langle\Dot{\mathcal{M}}^{(p+q)}\left(
            \pmb{\varphi}^{(0)}_i
            \right)\right\rangle}
        \sum_{N=1}^{\infty}
        \Bigg\{
        \underset{\pmb{\varphi}^{(0)}_{r_l},\pmb{\varphi}^{(0)}_{i_l}}{\int\cdots\int}
        \underset{\mathbf{n}^{(0)}_l,\dots,\mathbf{n}^{(p+q)}_l}{\int\cdots\int}
        \Dot{\mathcal{M}}^{(p+q)}\left(
            \pmb{\varphi}^{(0)}_i
            \,\Big|\,
            \mathbf{n}^{(p+q-1)}_1,\dots,\mathbf{n}^{(0)}_1
        \right) \, \left\langle\Dot{\mathcal{M}}_L\left(
            \pmb{\varphi}^{(p+q)}_{i}
            \right)\right\rangle
        \\
        &\qquad\cdot
        \mathcal{K_L}
        \left(
        \pmb{\varphi}^{(p+q)}_i\rightarrow\pmb{\varphi}^{(p+q)}_{r_1}
        \right) \, \frac{\mathcal{Q}(\mathbf{v}^{(p+q)}_{i} \, | \, \mathbf{n}^{(p+q)}_1)}
        {\mathcal{Q}(\mathbf{v}^{(0)}_{i} \, | \, \mathbf{n}^{(0)}_1)} \, p_{n_1}^{(q)}\left(\mathbf{n}^{(p+q:p+1)}_1 \, \Bigm| \, \mathbf{n}^{(p)}_1, \mathbf{v_{i_1}}^{(p)}\right)
        p_{n_1}^{(p)}\left(\mathbf{n}^{(p:1)}_1 \, \Bigm| \, \mathbf{n}^{(0)}_1, \mathbf{v_{i_1}}^{(0)}\right)
        \\
        &\qquad\cdot
        \prod_{l=2}^{N}
        \Bigg[
        \Dot{\mathcal{S}}^{(p+q)}\left(
        \pmb{\varphi}^{(0)}_{i_l},
        \mathbf{n}^{(p+q-1)}_l,\dots,\mathbf{n}^{(0)}_l
        \mid
        \pmb{\varphi}^{(0)}_{r_{l-1}},
        \mathbf{n}^{(p+q-1)}_{l-1},\dots,\mathbf{n}^{(0)}_{l-1}
        \right) \, \left\langle\Dot{\mathcal{M}}_L\left(
            \pmb{\varphi}^{(p+q)}_{i_{l}}
            \right)\right\rangle\, \mathcal{K_L}
        \left(
        \pmb{\varphi}^{(p+q)}_{i_l}\rightarrow\pmb{\varphi}^{(p+q)}_{r_l}
        \right)
        \\
        &\qquad\cdot \frac{\mathcal{Q}(\mathbf{v}^{(p+q)}_{i_l} \, | \, \mathbf{n}^{(p+q)}_{l})}
        {\mathcal{Q}(\mathbf{v}^{(0)}_{i_l} \, | \, \mathbf{n}^{(0)}_{l})}
        p_{n_l}^{(q)}\left(\mathbf{n}^{(p+q)}_{l} \, \Bigm| \, \mathbf{n}^{(p+q-1)}_{l}, \mathbf{v_{i_l}}^{(p+q-1)}\right)
        \Bigg] \,
        \mathcal{S}^{(p+q)}\left(
        \pmb{\varphi}^{(0)}_r
        \,\Big|\,
        \mathbf{n}^{(p+q-1)}_N,\dots,\mathbf{n}^{(0)}_N
        \right)
        \\
        &\qquad\cdot
        \dif\pmb{\varphi}^{(0)}_{r_1}\cdots\dif\pmb{\varphi}^{(0)}_{r_{N-1}}\,
        \dif\pmb{\varphi}^{(0)}_{i_2}\cdots\dif\pmb{\varphi}^{(0)}_{i_N}\,
        \dif \mathbf{n}^{(1)}_{1}\cdots\dif \mathbf{n}^{(1)}_{N}
        \cdots
        \dif \mathbf{n}^{(p+q)}_{1}\cdots\dif \mathbf{n}^{(p+q)}_{N}
        \Bigg\} \\
        & \qquad =
            \frac{1}{\left\langle\Dot{\mathcal{M}}^{(p+q)}\left(
            \pmb{\varphi}^{(0)}_i
            \right)\right\rangle}
            \sum_{N=1}^{\infty}
            \Biggl\{
            \int \cdots \int
            \Dot{\mathcal{M}}^{(p+q)}\left(
            \pmb{\varphi}^{(0)}_i
            \,\Big|\,
            \mathbf{n}^{(p+q-1)}_{1},\dots,\mathbf{n}^{(1)}_{1},\mathbf{n}^{(0)}_{1}
            \right) \, p_{n_l}^{(p+q-1)}\left(\mathbf{n}^{(p+q-1:1)}_1 \, | \, \mathbf{n}^{(0)}_1, \mathbf{v}^{(0)}_{i_1}\right)
            \\
            &\qquad\cdot
            \mathcal{K_K}\left(
            \pmb{\varphi}^{(0)}_i \rightarrow \pmb{\varphi}^{(0)}_{r_1}
            \right) \, \prod_{l=2}^{N}
            \Biggl[
            \Dot{\mathcal{S}}^{(p+q)}\Bigl(
            \pmb{\varphi}^{(0)}_{i_{l}},\,
            \mathbf{n}^{(p+q-1)}_l,\dots,\mathbf{n}^{(0)}_l
            \,\Big|\,
            \pmb{\varphi}^{(0)}_{r_{l-1}},\,
            \mathbf{n}^{(p+q-1)}_{l-1},\dots,\mathbf{n}^{(0)}_{l-1}
            \Bigr) \, \mathcal{K_K}\left(
            \pmb{\varphi}^{(0)}_{i_l} \rightarrow \pmb{\varphi}^{(0)}_{r_l}
            \right)
            \Biggr]
            \\
            &\qquad\cdot
            \mathcal{S}^{(p+q)}\left(
            \pmb{\varphi}^{(0)}_r
            \,\Big|\,
            \mathbf{n}^{(p+q-1)}_{N},\dots,\mathbf{n}^{(1)}_{N},\mathbf{n}^{(0)}_{N}
            \right) \, \dif\pmb{\varphi}^{(0)}_{r_1}\cdots\dif\pmb{\varphi}^{(0)}_{r_{N-1}}\,
            \dif\pmb{\varphi}^{(0)}_{i_2}\cdots\dif\pmb{\varphi}^{(0)}_{i_N} \\
            & \qquad\cdot
            \dif\mathbf{n}^{(1)}_1\cdots\dif\mathbf{n}^{(1)}_N \, \dots \, \dif\mathbf{n}^{(p+q-1)}_1\cdots\dif\mathbf{n}^{(p+q-1)}_N
            \Biggr\} \\
        &\qquad = \left[\Psi^{(p+q)} \circ \mathcal{K_L}\right]\left(\pmb{\varphi}^{(0)}_i \rightarrow \pmb{\varphi}^{(0)}_r \right),
    \end{aligned}
\end{equation}
where $\Psi^{(p+q)}$ is a new surface with $p+q$ scales in the sense of \cref{def:surface_multiscale}. To complete the proof, we must show that $\Psi^{(p+q)} \circ \mathcal{K_L} \in \mathcal{T}^s$, i.e. that it satisfies impermeability, reciprocity, normalisation and non-negativity in the sense of \cref{def:impermeability,def:reciprocity,def:normalisation,def:nonnegativity}, provided that $\mathcal{K_L} \in \mathcal{T}^p$. From the structure of \cref{eq:scattering_operator}, the operator $\circ$ is a direct generalisation of the multi-reflection kernel in \cref{eq:multi_reflection_kernel}. The only difference is that the shadowing and masking functions may depend on the complete hierarchy of surface normals from $\mathbf{n}^{(0)}$ to $\mathbf{n}^{(p+q-1)}$. Therefore, \cref{lemma:reciprocity_multi_refl,lemma:normalisation_multi_refl,lemma:nonnegativity_multi_refl} apply directly. Furthermore, the restrictions imposed on the shadowing and masking functions in \cref{lemma:reciprocity_multi_refl,lemma:normalisation_multi_refl} carry over directly to the multi-scale setting, where they take the form of \cref{eq:shadowing_multi_scale_condition_1,eq:shadowing_multi_scale_condition_2,eq:shadowing_multi_scale_condition_3}. Hence, we must show that $\mathcal{S}^{(p+q)}$ complies with these conditions. We first focus on \cref{eq:shadowing_multi_scale_condition_1}. Substituting the definition of the multi-scale shadowing function from \cref{def:shadowing_multi_scale}, we obtain
\begin{equation}
\begin{aligned}
        & p^{(p+q-1)}_{n_{l-1}}\Bigl(
        \mathbf{n}^{(p+q-1:1)}_{l-1}
        \,\Big|\,
        \mathbf{n}^{(0)}_{l-1}, \mathbf{v}^{(0)}_{i_{l-1}}
        \Bigr)
        \, \mathcal{S}^{(p+q)}\left(
        \pmb{\varphi}^{(0)}_{i_l},
        \mathbf{n}^{(p+q-1)}_l,\dots,\mathbf{n}^{(0)}_l
        \mid
        \pmb{\varphi}^{(0)}_{r_{l-1}},
        \mathbf{n}^{(p+q-1)}_{l-1},\dots,\mathbf{n}^{(0)}_{l-1}
        \right) = \\
        &\qquad =
        p^{(p+q-1)}_{n_{l-1}}\Bigl(
        \mathbf{n}^{(p+q-1:1)}_{l-1}
        \,\Big|\,
        \mathbf{n}^{(0)}_{l-1}, \mathbf{v}^{(0)}_{i_{l-1}}
        \Bigr)
        \, \sum_{j=1}^{p+q}
        \Biggl\{
        \Biggl[
        \prod_{\alpha=j+1}^{p+q}
        \mathcal{S}_{\alpha}\Bigl(
        \pmb{\varphi}^{(\alpha-1)}_{r_{l-1}} \, | \, \mathbf{n}^{(\alpha-1)}_{l-1}
        \Bigr)
        \Biggr]
        \mathcal{S}_{j}\Bigl(
        \pmb{\varphi}^{(j-1)}_{i_l}
        \,\Big|\,
        \pmb{\varphi}^{(j-1)}_{r_{l-1}}
        \Bigr)
        \\
        &\qquad\qquad\cdot
        \Biggl[
        \prod_{\alpha=j+1}^{p+q}
        \mathcal{M}_{\alpha}\Bigl(
        \pmb{\varphi}^{(\alpha-1)}_{i_l} \, | \, \mathbf{n}^{(\alpha-1)}_{l}
        \Bigr)
        \,
        p_{n_l}^{(\alpha-1)}\Bigl(
        \mathbf{n}^{(\alpha-1)}_{l}
        \,\Big|\,
        \mathbf{n}^{(\alpha-2)}_{l}, \mathbf{v}^{(\alpha-2)}_{i_{l}}
        \Bigr)
        \Biggr]
        \Biggl[
        \prod_{\alpha=1}^{j-1}
        \delta\Bigl(
        \mathbf{n}^{(\alpha)}_{l}
        -
        \mathbf{n}^{(\alpha)}_{l-1}
        \Bigr) \,\delta\Bigl(
        \mathbf{v}^{(\alpha)}_{i_l}
        -
        \mathbf{v}^{(\alpha)}_{i_{l-1}}
        \Bigr)
        \Biggr]
        \Biggr\} \\
        &\qquad =
        \prod_{\alpha=2}^{p+q} p_{n_{l-1}}^{(\alpha-1)}\Bigl(
        \mathbf{n}^{(\alpha-1)}_{l-1}
        \,\Big|\,
        \mathbf{n}^{(\alpha-2)}_{l-1}, \mathbf{v}^{(\alpha-2)}_{i_{l-1}}
        \Bigr)
        \, \sum_{j=1}^{p+q}
        \Biggl\{
        \Biggl[
        \prod_{\alpha=j+1}^{p+q}
        \mathcal{S}_{\alpha}\Bigl(
        \pmb{\varphi}^{(\alpha-1)}_{r_{l-1}} \, | \, \mathbf{n}^{(\alpha-1)}_{l-1}
        \Bigr)
        \Biggr]
        \mathcal{S}_{j}\Bigl(
        \pmb{\varphi}^{(j-1)}_{i_l}
        \,\Big|\,
        \pmb{\varphi}^{(j-1)}_{r_{l-1}}
        \Bigr)
        \\
        &\qquad\qquad\cdot
        \Biggl[
        \prod_{\alpha=j+1}^{p+q}
        \mathcal{M}_{\alpha}\Bigl(
        \pmb{\varphi}^{(\alpha-1)}_{i_l} \, | \, \mathbf{n}^{(\alpha-1)}_{l}
        \Bigr)
        \,
        p_{n_l}^{(\alpha-1)}\Bigl(
        \mathbf{n}^{(\alpha-1)}_{l}
        \,\Big|\,
        \mathbf{n}^{(\alpha-2)}_{l}, \mathbf{v}^{(\alpha-2)}_{i_{l}}
        \Bigr)
        \Biggr]
        \Biggl[
        \prod_{\alpha=1}^{j-1}
        \delta\Bigl(
        \mathbf{n}^{(\alpha)}_{l}
        -
        \mathbf{n}^{(\alpha)}_{l-1}
        \Bigr) \,\delta\Bigl(
        \mathbf{v}^{(\alpha)}_{i_l}
        -
        \mathbf{v}^{(\alpha)}_{i_{l-1}}
        \Bigr)
        \Biggr]
        \Biggr\} \\
        &\qquad =
        \sum_{j=1}^{p+q}
        \Biggl\{
        \Biggl[\prod_{\alpha=1}^{j-1}p_{n_{l-1}}^{(\alpha)}\Bigl(
        \mathbf{n}^{(\alpha)}_{l-1}
        \,\Big|\,
        \mathbf{n}^{(\alpha-1)}_{l-1}, \mathbf{v}^{(\alpha-1)}_{i_{l-1}}
        \Bigr)\Biggr]
        \Biggl[
        \prod_{\alpha=j+1}^{p+q}
        \mathcal{S}_{\alpha}\Bigl(
        \pmb{\varphi}^{(\alpha-1)}_{r_{l-1}} \, | \, \mathbf{n}^{(\alpha-1)}_{l-1}
        \Bigr) \, p_{n_{l-1}}^{(\alpha-1)}\Bigl(
        \mathbf{n}^{(\alpha-1)}_{l-1}
        \,\Big|\,
        \mathbf{n}^{(\alpha-2)}_{l-1}, \mathbf{v}^{(\alpha-2)}_{i_{l-1}}
        \Bigr)
        \Biggr]
        \\
        &\qquad\qquad\cdot
        \mathcal{S}_{j}\Bigl(
        \pmb{\varphi}^{(j-1)}_{i_l}
        \,\Big|\,
        \pmb{\varphi}^{(j-1)}_{r_{l-1}}
        \Bigr)\,\Biggl[
        \prod_{\alpha=j+1}^{p+q}
        \mathcal{M}_{\alpha}\Bigl(
        \pmb{\varphi}^{(\alpha-1)}_{i_l} \, | \, \mathbf{n}^{(\alpha-1)}_{l}
        \Bigr)
        \,
        p_{n_l}^{(\alpha-1)}\Bigl(
        \mathbf{n}^{(\alpha-1)}_{l}
        \,\Big|\,
        \mathbf{n}^{(\alpha-2)}_{l}, \mathbf{v}^{(\alpha-2)}_{i_{l}}
        \Bigr)
        \Biggr] \\
        & \qquad\qquad \cdot 
        \Biggl[
        \prod_{\alpha=1}^{j-1}
        \delta\Bigl(
        \mathbf{n}^{(\alpha)}_{l}
        -
        \mathbf{n}^{(\alpha)}_{l-1}
        \Bigr) \,\delta\Bigl(
        \mathbf{v}^{(\alpha)}_{i_l}
        -
        \mathbf{v}^{(\alpha)}_{i_{l-1}}
        \Bigr)
        \Biggr]
        \Biggr\} \\
        &\qquad =
        \sum_{j=1}^{p+q}
        \Biggl\{
        \Biggl[\prod_{\alpha=1}^{p+q-1}p_{n_{l}}^{(\alpha)}\Bigl(
        \mathbf{n}^{(\alpha)}_{l}
        \,\Big|\,
        \mathbf{n}^{(\alpha-1)}_{l}, \mathbf{v}^{(\alpha-1)}_{i_{l}}
        \Bigr)\Biggr]
        \Biggl[
        \prod_{\alpha=j+1}^{p+q}
        \mathcal{M}_{\alpha}\Bigl(
        -\pmb{\varphi}^{(\alpha-1)}_{r_{l-1}} \, | \, \mathbf{n}^{(\alpha-1)}_{l-1}
        \Bigr) \, p_{n_{l-1}}^{(\alpha-1)}\Bigl(
        \mathbf{n}^{(\alpha-1)}_{l-1}
        \,\Big|\,
        \mathbf{n}^{(\alpha-2)}_{l-1}, \mathbf{v}^{(\alpha-2)}_{i_{l-1}}
        \Bigr)
        \Biggr]
        \\
        &\qquad\qquad\cdot
        \mathcal{S}_{j}\Bigl(
        -\pmb{\varphi}^{(j-1)}_{r_{l-1}}
        \,\Big|\,
        -\pmb{\varphi}^{(j-1)}_{i_l}
        \Bigr)\,\Biggl[
        \prod_{\alpha=j+1}^{p+q}
        \mathcal{S}_{\alpha}\Bigl(
        -\pmb{\varphi}^{(\alpha-1)}_{i_l} \, | \, \mathbf{n}^{(\alpha-1)}_{l}
        \Bigr)\, \Biggl[
        \prod_{\alpha=1}^{j-1}
        \delta\Bigl(
        \mathbf{n}^{(\alpha)}_{l}
        -
        \mathbf{n}^{(\alpha)}_{l-1}
        \Bigr) \,\delta\Bigl(
        \mathbf{v}^{(\alpha)}_{i_l}
        -
        \mathbf{v}^{(\alpha)}_{i_{l-1}}
        \Bigr)
        \Biggr]
        \Biggr\} \\
        & \qquad = p^{(p+q-1)}_{n_{l}}\Bigl(
        \mathbf{n}^{(p+q-1:1)}_{l}
        \,\Big|\,
        \mathbf{n}^{(0)}_{l}, \mathbf{v}^{(0)}_{i_{l}}
        \Bigr)
        \, \mathcal{S}^{(p+q)}\left(
        -\pmb{\varphi}^{(0)}_{r_{l-1}},
        \mathbf{n}^{(p+q-1)}_{l-1},\dots,\mathbf{n}^{(0)}_{l-1}
        \mid
        -\pmb{\varphi}^{(0)}_{i_l},
        \mathbf{n}^{(p+q-1)}_l,\dots,\mathbf{n}^{(0)}_l
        \right).
\end{aligned}
\end{equation}
Next, we prove compliance with \cref{eq:shadowing_multi_scale_condition_2}. For this, we substitute \cref{eq:shadowing_multiscale} into \cref{eq:shadowing_multi_scale_condition_2}, as
\begin{equation}
\begin{aligned}
    &\int_{\mathbb{R}}
    \prod_{\alpha=1}^{p+q-1}\int_{\mathbb{S}^2}
    \mathcal{S}^{(p+q)}\left(
        \pmb{\varphi}^{(0)}_{i_l},
        \mathbf{n}^{(p+q-1)}_l,\dots,\mathbf{n}^{(0)}_l
        \mid
        \pmb{\varphi}^{(0)}_{r_{l-1}},
        \mathbf{n}^{(p+q-1)}_{l-1},\dots,\mathbf{n}^{(0)}_{l-1}
    \right)
    f^{(p+q)}_{\varphi}\left(
        \pmb{\varphi}^{(0)}_{i_l}
    \right)
    \prod_{\alpha=1}^{p+q-1}\dif\mathbf{n}^{(\alpha)}_l\,
    \dif t_{i_l}
    \\
    &\qquad =
    \int_{\mathbb{R}}
    \prod_{\alpha=1}^{p+q-1}\int_{\mathbb{S}^2}
    \Dot{\mathcal{S}}^{(p+q)}\left(
        \pmb{\varphi}^{(0)}_{i_l},
        \mathbf{n}^{(p+q-1)}_l,\dots,\mathbf{n}^{(0)}_l
        \mid
        \pmb{\varphi}^{(0)}_{r_{l-1}},
        \mathbf{n}^{(p+q-1)}_{l-1},\dots,\mathbf{n}^{(0)}_{l-1}
    \right)
    \prod_{\alpha=1}^{p+q-1}\dif\mathbf{n}^{(\alpha)}_l\,
    \dif t_{i_l}
    \\
    &\qquad =
    \prod_{\alpha=1}^{p+q-1}\int_{\mathbb{S}^2}
    \left.
    \mathcal{S}^{(p+q)}\left(
        \pmb{\varphi}^{(0)}_{i_l},
        \mathbf{n}^{(p+q-1)}_l,\dots,\mathbf{n}^{(0)}_l
        \mid
        \pmb{\varphi}^{(0)}_{r_{l-1}},
        \mathbf{n}^{(p+q-1)}_{l-1},\dots,\mathbf{n}^{(0)}_{l-1}
    \right)
    \right|_{t_{i_l}=0}^{t_{i_l}=\infty}
    \prod_{\alpha=1}^{p+q-1}\dif\mathbf{n}^{(\alpha)}_l
    \\
    &\qquad =
    \prod_{\alpha=1}^{p+q-1}\int_{\mathbb{S}^2}
    \sum_{j=1}^{p+q}
    \Biggl\{
    \Biggl[
        \prod_{\alpha=j+1}^{p+q}
        \mathcal{S}_{\alpha}\Bigl(
            \pmb{\varphi}^{(\alpha-1)}_{r_{l-1}}
            \,\Big|\,
            \mathbf{n}^{(\alpha-1)}_{l-1}
        \Bigr)
    \Biggr]
    \left.
    \mathcal{S}_{j}\Bigl(
        \pmb{\varphi}^{(j-1)}_{i_l}
        \,\Big|\,
        \pmb{\varphi}^{(j-1)}_{r_{l-1}}
    \Bigr)
    \right|_{t_{i_l}=0}^{t_{i_l}=\infty}
    \\
    &\qquad\qquad\cdot
    \Biggl[
        \prod_{\alpha=j+1}^{p+q}
        \left.
        \mathcal{M}_{\alpha}\Bigl(
            \pmb{\varphi}^{(\alpha-1)}_{i_l}
            \,\Big|\,
            \mathbf{n}^{(\alpha-1)}_{l}
        \Bigr)
        \right|_{t_{i_l}=0}^{t_{i_l}=\infty}
        p_{n_l}^{(\alpha-1)}\Bigl(
            \mathbf{n}^{(\alpha-1)}_{l}
            \,\Big|\,
            \mathbf{n}^{(\alpha-2)}_{l},
            \mathbf{v}^{(\alpha-2)}_{i_l}
        \Bigr)
    \Biggr]
    \\
    &\qquad\qquad\cdot
    \Biggl[
        \prod_{\alpha=1}^{j-1}
        \delta\Bigl(
            \mathbf{n}^{(\alpha)}_{l}
            -
            \mathbf{n}^{(\alpha)}_{l-1}
        \Bigr)
        \delta\Bigl(
            \mathbf{v}^{(\alpha)}_{i_l}
            -
            \mathbf{v}^{(\alpha)}_{i_{l-1}}
        \Bigr)
    \Biggr]
    \Biggr\}
    \prod_{\alpha=1}^{p+q-1}\dif\mathbf{n}^{(\alpha)}_l
    \\
    &\qquad =
    \sum_{j=1}^{p+q}
    \Biggl\{
    \Biggl[
        \prod_{\alpha=j+1}^{p+q}
        \mathcal{S}_{\alpha}\Bigl(
            \pmb{\varphi}^{(\alpha-1)}_{r_{l-1}}
            \,\Big|\,
            \mathbf{n}^{(\alpha-1)}_{l-1}
        \Bigr)
    \Biggr]
    \Biggl[
        1-
        \mathcal{S}_{j}\Bigl(
            \pmb{\varphi}^{(j-1)}_{r_{l-1}}
            \,\Big|\,
            \mathbf{n}^{(j-1)}_{l-1}
        \Bigr)
    \Biggr]
    \Biggr\}
    \\
    &\qquad =
    1-
    \prod_{\alpha=1}^{p+q}
    \mathcal{S}_{\alpha}\Bigl(
        \pmb{\varphi}^{(\alpha-1)}_{r_{l-1}}
        \,\Big|\,
        \mathbf{n}^{(\alpha-1)}_{l-1}
    \Bigr)
    \\
    &\qquad =
    1-
    \mathcal{S}^{(p+q)}\left(
        \pmb{\varphi}^{(0)}_{r_{l-1}}
        \,\Big|\,
        \mathbf{n}^{(p+q-1)}_{l-1},\dots,
        \mathbf{n}^{(1)}_{l-1},\mathbf{n}^{(0)}_{l-1}
    \right).
\end{aligned}
\end{equation}
Finally, we must verify that \cref{eq:shadowing_multi_scale_condition_3} holds for the composed surface, provided that the corresponding single-scale condition in \cref{eq:shadowing_single_scale_condition_3} holds at each scale. Substituting the multi-scale one-point shadowing function from \cref{eq:shadowing_multiscale} into \cref{eq:shadowing_multi_scale_condition_3}, and expanding the kernel $\mathcal{K_K}$ in terms of the local kernel $\mathcal{K_L}$, gives
\begin{equation}
\begin{aligned}
    &\int_{\mathbb{S}^2}\dots\int_{\mathbb{S}^2}
    \int_{\mathbb{R}^3}\int_{\mathbb{R}^3}\int_{\mathbb{R}}
    \mathcal{K_K}\left(
        \pmb{\varphi}^{(0)}_{i_l} \rightarrow \pmb{\varphi}^{(0)}_{r_l}
    \right)
    \mathcal{S}^{(p+q)}\left(
        \pmb{\varphi}^{(0)}_{r_l}
        \mid
        \mathbf{n}^{(p+q-1)}_l,\dots,\mathbf{n}^{(0)}_l
    \right)
    \\
    &\qquad\cdot
    \prod_{\alpha=2}^{p+q}
    p_{n_l}^{(\alpha-1)}\left(
        \mathbf{n}_l^{(\alpha-1)}
        \mid
        \mathbf{n}_l^{(\alpha-2)},\mathbf{v}^{(\alpha-2)}_{i_l}
    \right)
    \dif\pmb{\varphi}^{(0)}_{r_l}
    \prod_{\alpha=1}^{p+q-1}\dif\mathbf{n}^{(\alpha)}_l
    \\
    &\qquad =
    \int_{\mathbb{S}^2}\dots\int_{\mathbb{S}^2}
    \int_{\mathbb{R}^3}\int_{\mathbb{R}^3}\int_{\mathbb{R}}
    \left\langle
        \Dot{\mathcal{M}}_L\left(
            \pmb{\varphi}^{(p+q)}_i
        \right)
    \right\rangle
    \mathcal{K_L}\left(
        \pmb{\varphi}^{(p+q)}_{i_l}
        \rightarrow
        \pmb{\varphi}^{(p+q)}_{r_l}
    \right)
    \frac{
        \mathcal{Q}\left(
            \mathbf{v}^{(p+q)}_{i_l}
            \mid
            \mathbf{n}^{(p+q)}_l
        \right)
    }{
        \mathcal{Q}\left(
            \mathbf{v}^{(0)}_{i_l}
            \mid
            \mathbf{n}^{(0)}_l
        \right)
    }
    \\
    &\qquad\cdot
    \prod_{\alpha=1}^{p+q}
    \mathcal{S}_{\alpha}\Bigl(
        \pmb{\varphi}^{(\alpha-1)}_{r_l}
        \mid
        \mathbf{n}^{(\alpha-1)}_l
    \Bigr)
    \prod_{\alpha=1}^{p+q}
    p_{n_l}^{(\alpha)}\left(
        \mathbf{n}_l^{(\alpha)}
        \mid
        \mathbf{n}_l^{(\alpha-1)},\mathbf{v}^{(\alpha-1)}_{i_l}
    \right)
    \dif\pmb{\varphi}^{(0)}_{r_l}
    \prod_{\alpha=1}^{p+q}\dif\mathbf{n}^{(\alpha)}_l
    \\
    &\qquad =
    \int_{\mathbb{R}^3}\int_{\mathbb{R}^3}\int_{\mathbb{R}}
    \Biggl\{
    \int_{\mathbb{S}^2}\dots\int_{\mathbb{S}^2}
    \left\langle
        \Dot{\mathcal{M}}_L\left(
            \pmb{\varphi}^{(p+q)}_i
        \right)
    \right\rangle
    \mathcal{K_L}\left(
        \pmb{\varphi}^{(p+q)}_{i_l}
        \rightarrow
        \pmb{\varphi}^{(p+q)}_{r_l}
    \right)
    \frac{
        \mathcal{Q}\left(
            \mathbf{v}^{(p+q)}_{i_l}
            \mid
            \mathbf{n}^{(p+q)}_l
        \right)
    }{
        \mathcal{Q}\left(
            \mathbf{v}^{(0)}_{i_l}
            \mid
            \mathbf{n}^{(0)}_l
        \right)
    }
    \\
    &\qquad\qquad\cdot
    \prod_{\alpha=2}^{p+q}
    \mathcal{S}_{\alpha}\Bigl(
        \pmb{\varphi}^{(\alpha-1)}_{r_l}
        \mid
        \mathbf{n}^{(\alpha-1)}_l
    \Bigr)
    \prod_{\alpha=1}^{p+q}
    p_{n_l}^{(\alpha)}\left(
        \mathbf{n}_l^{(\alpha)}
        \mid
        \mathbf{n}_l^{(\alpha-1)},\mathbf{v}^{(\alpha-1)}_{i_l}
    \right)
    \prod_{\alpha=1}^{p+q}\dif\mathbf{n}^{(\alpha)}_l
    \Biggr\}
    \mathcal{S}_{1}\Bigl(
        \pmb{\varphi}^{(0)}_{r_l}
        \mid
        \mathbf{n}^{(0)}_l
    \Bigr)
    \dif\pmb{\varphi}^{(0)}_{r_l}
    \\
    &\qquad =
    \int_{\mathbb{R}^3}\int_{\mathbb{R}^3}\int_{\mathbb{R}}
    \Biggl\{
    \int_{\mathbb{S}^2}\dots\int_{\mathbb{S}^2}
    \left\langle
        \Dot{\mathcal{M}}_L\left(
            \pmb{\varphi}^{(p+q)}_i
        \right)
    \right\rangle
    \mathcal{K_L}\left(
        \pmb{\varphi}^{(p+q)}_{i_l}
        \rightarrow
        \pmb{\varphi}^{(p+q)}_{r_l}
    \right)
    \frac{
        \mathcal{Q}\left(
            \mathbf{v}^{(p+q)}_{i_l}
            \mid
            \mathbf{n}^{(p+q)}_l
        \right)
    }{
        \mathcal{Q}\left(
            \mathbf{v}^{(0)}_{i_l}
            \mid
            \mathbf{n}^{(0)}_l
        \right)
    }
    \\
    &\qquad\qquad\cdot
    \mathcal{S}^{(2:p+q)}
    \Bigl(
        \pmb{\varphi}^{(p+q)}_{r_l}
        \mid
        \mathbf{n}^{(1:p+q-1)}_l
    \Bigr) \,
    p_{n_l}^{(1:p+q)}
    \left(
        \mathbf{n}_l^{(1:p+q)}
        \mid
        \mathbf{n}_l^{(0)},\mathbf{v}^{(0)}_{i_l}
    \right)
    \prod_{\alpha=1}^{p+q}\dif\mathbf{n}^{(\alpha)}_l
    \Biggr\}
    \mathcal{S}_{1}\Bigl(
        \pmb{\varphi}^{(0)}_{r_l}
        \mid
        \mathbf{n}^{(0)}_l
    \Bigr)
    \dif\pmb{\varphi}^{(0)}_{r_l}.
\end{aligned}
\end{equation}
In the last equality, the factors associated with scales $2,\dots,p+q$ have been collected into the shorthand quantities $\mathcal{S}^{(2:p+q)}$ and $p_{n_l}^{(1:p+q)}$ while the outermost one-point shadowing factor $\mathcal{S}_1$ has been kept explicit. We may therefore apply \cref{eq:shadowing_single_scale_condition_3} to the bracketed expression, interpreting $\mathcal{S}^{(2:p+q)}
    \Bigl(
        \pmb{\varphi}^{(p+q)}_{r_l}
        \mid
        \mathbf{n}^{(1:p+q-1)}_l
    \Bigr)$ as the effective local one-point shadowing function seen from scale $1$, and assuming $\left\langle
        \Dot{\mathcal{M}}_L\left(
            \pmb{\varphi}^{(p+q)}_i
        \right)
    \right\rangle = 1$ since $\Psi^{(p+q)}$ is locally smooth in the last scale, in the sense of \cref{def:local_smoothness}. Hence, there exists $\epsilon_1 > 0$ such that
\begin{equation}
    \begin{aligned}
            &\int_{\mathbb{R}^3}\int_{\mathbb{R}^3}\int_{\mathbb{R}}
    \Biggl\{
    \int_{\mathbb{S}^2}\dots\int_{\mathbb{S}^2}
    \left\langle
        \Dot{\mathcal{M}}_L\left(
            \pmb{\varphi}^{(p+q)}_i
        \right)
    \right\rangle
    \mathcal{K_L}\left(
        \pmb{\varphi}^{(p+q)}_{i_l}
        \rightarrow
        \pmb{\varphi}^{(p+q)}_{r_l}
    \right)
    \frac{
        \mathcal{Q}\left(
            \mathbf{v}^{(p+q)}_{i_l}
            \mid
            \mathbf{n}^{(p+q)}_l
        \right)
    }{
        \mathcal{Q}\left(
            \mathbf{v}^{(0)}_{i_l}
            \mid
            \mathbf{n}^{(0)}_l
        \right)
    }
    \\
    &\qquad\qquad\cdot
    \mathcal{S}^{(2:p+q)}
    \Bigl(
        \pmb{\varphi}^{(p+q)}_{r_l}
        \mid
        \mathbf{n}^{(1:p+q-1)}_l
    \Bigr) \,
    p_{n_l}^{(1:p+q)}
    \left(
        \mathbf{n}_l^{(1:p+q)}
        \mid
        \mathbf{n}_l^{(0)},\mathbf{v}^{(0)}_{i_l}
    \right)
    \prod_{\alpha=1}^{p+q}\dif\mathbf{n}^{(\alpha)}_l
    \Biggr\}
    \mathcal{S}_{1}\Bigl(
        \pmb{\varphi}^{(0)}_{r_l}
        \mid
        \mathbf{n}^{(0)}_l
    \Bigr)
    \dif\pmb{\varphi}^{(0)}_{r_l} > \epsilon_1.
    \end{aligned}
\end{equation}
We have now proven all three conditions outlined in \cref{eq:shadowing_multi_scale_condition_1,eq:shadowing_multi_scale_condition_2,eq:shadowing_multi_scale_condition_3}, and hence, shown that $\Psi^{(p+q)} \circ \mathcal{K_L}$ is a scattering kernel in the sense of \cref{def:scattering_kernel_set}, provided that $\mathcal{K_L}$ is a pointwise scattering kernel in the sense of \cref{def:point_wise_scattering_kernel_set}. This concludes the proof.
\end{proof}

Finally, we show that, for surfaces with a height representation in the sense of \cref{def:surface_multiscale_height}, the composed surface coincides with the sum of the components in the global frame. 

\begin{theorem} \label{theorem:scattering_kernel_summation}
    Let $\circ$ denote the scattering operator in the sense of \cref{def:scattering_operator}, and let $\Psi^{(p)}_1, \Psi^{(q)}_2 \in \mathcal{P}^h$ be two multi-scale surfaces \textbf{with height representation} in the sense of \cref{def:surface_multiscale_height}, with scales $p,q\in\mathbb{Z}^+$. Furthermore, let $\mathcal{K_L}: \mathbb{R}^3 \times \mathbb{R}^3 \times \mathbb{R}
    \times \mathbb{R}^3 \times \mathbb{R}^3 \times \mathbb{R} \rightarrow [0, \infty)$ be a Lebesgue-integrable function defined in the most local reference frame associated with the composed surface, in the sense of \cref{def:local_frame}. Then, 
    \begin{equation} \label{eq:scattering_operator_summation}
    \boxed{
        \left[\Psi^{(p)}_1 \circ \left(\Psi^{(q)}_2 \circ \mathcal{K_L}\right)\right]\left(
        \mathbf{r}^{(0)}_i \rightarrow \mathbf{r}^{(0)}_r,
        \mathbf{v}^{(0)}_i \rightarrow \mathbf{v}^{(0)}_r,
        t_i \rightarrow t_r
        \right)
        =
        \left[\left(\Psi_1^{(p)} + \Psi_2^{(q)} \right) \circ \mathcal{K_L}\right]\left(
        \mathbf{r}^{(0)}_i \rightarrow \mathbf{r}^{(0)}_r,
        \mathbf{v}^{(0)}_i \rightarrow \mathbf{v}^{(0)}_r,
        t_i \rightarrow t_r
        \right),
    }
    \end{equation}
     $\forall \, \Psi^{(p)}_1, \Psi^{(q)}_2 \in \mathcal{P}$ and $ \forall \, \left(\mathbf{r}^{(0)}_i,\mathbf{v}^{(0)}_i,t_i,\mathbf{r}^{(0)}_r,\mathbf{v}^{(0)}_r,t_r\right)\in \mathbb{R}^3 \times \mathbb{R}^3 \times \mathbb{R} \times \mathbb{R}^3 \times \mathbb{R}^3 \times \mathbb{R}$. Furthermore, $\left[\left(\Psi_1^{(p)} + \Psi_2^{(q)} \right) \circ \mathcal{K_L}\right]\left(
    \mathbf{r}^{(0)}_i \rightarrow \mathbf{r}^{(0)}_r,
    \mathbf{v}^{(0)}_i \rightarrow \mathbf{v}^{(0)}_r,
    t_i \rightarrow t_r
    \right)$ is a scattering kernel in the sense of \cref{def:reciprocity,def:normalisation,def:nonnegativity,def:impermeability} provided that $\mathcal{K_L}$ is a pointwise scattering kernel in the sense of \cref{def:pointwise_reciprocity,def:normalisation,def:nonnegativity}, and the two-point and one-point shadowing functions $\mathcal{S}^{(p+q)}\left(\pmb{\varphi}^{(0)}_{i_k}, \mathbf{n}^{(p+q-1)}_k,\dots,\mathbf{n}^{(0)}_k \mid \pmb{\varphi}^{(0)}_{r_{k-1}}, \mathbf{n}^{(p+q-1)}_{k-1},\dots,\mathbf{n}^{(0)}_{k-1}\right)$ and $\mathcal{S}^{(p+q)}\left(\pmb{\varphi}^{(0)}_{r} \mid \mathbf{n}^{(p+q-1)}_n,\dots,\mathbf{n}^{(0)}_n\right)$ associated with $\Psi^{(p+q)} = \Psi_1^{(p)} + \Psi_2^{(q)}$ satisfy \cref{eq:shadowing_multi_scale_condition_1,eq:shadowing_multi_scale_condition_2,eq:shadowing_multi_scale_condition_3}. Hence, $\circ$ defines \textbf{an action} of the abelian group $(\mathcal{P}^h, +)$ on $\mathcal{T}$.
\end{theorem}
\begin{proof}
    By \cref{theorem:scattering_kernel_composition}, $\exists \, \Psi^{(p+q)} \in \mathcal{P}$, such that
    \begin{equation}
        \left[\Psi^{(p)}_1 \circ \left(\Psi^{(q)}_2 \circ \mathcal{K_L}\right)\right]\left(
        \mathbf{r}^{(0)}_i \rightarrow \mathbf{r}^{(0)}_r,
        \mathbf{v}^{(0)}_i \rightarrow \mathbf{v}^{(0)}_r,
        t_i \rightarrow t_r
        \right)
        =
        \left[\Psi^{(p+q)} \circ \mathcal{K_L}\right]\left(
        \mathbf{r}^{(0)}_i \rightarrow \mathbf{r}^{(0)}_r,
        \mathbf{v}^{(0)}_i \rightarrow \mathbf{v}^{(0)}_r,
        t_i \rightarrow t_r
        \right).
    \end{equation}
    Furthermore, $\Psi^{(p+q)} \circ \mathcal{K_L}$ is a scattering kernel in the sense of \cref{def:scattering_kernel_set}, provided that $\mathcal{K_L}$ is a pointwise scattering kernel in the sense of \cref{def:point_wise_scattering_kernel_set}. Hence, we only have to prove that the joint PDF of the normals of $\Psi^{(p+q)}$ coincides with that of a surface obtained by summation of the two components, i.e. $\Psi_1^{(p)} + \Psi_2^{(q)}$.     By \cref{def:surface_multiscale_height}, since
    $\Psi^{(p)}_1,\Psi^{(q)}_2 \in \mathcal{P}^h$, there exist height
    representations
    \begin{equation}
        \xi^{(0)}_1
        =
        \sum_{\alpha=1}^{p}
        \xi^{(0)}_{1,\alpha},
        \qquad
        \xi^{(0)}_2
        =
        \sum_{\beta=1}^{q}
        \xi^{(0)}_{2,\beta},
    \end{equation}
    in the same global frame $F_0$. Hence, by the addition operation in
    \cref{def:surface_multiscale_height},
    \begin{equation}
        \Psi^{(p)}_1+\Psi^{(q)}_2
        =
        \left\{
        (x,y,z,t)\in \mathbb{R}^3\times\mathbb{R}
        \;\middle|\;
        z=\xi^{(0)}_1(x,y,t)+\xi^{(0)}_2(x,y,t)
        \right\}.
    \end{equation}
    Therefore,
    \begin{equation}
        \xi^{(0)}_1+\xi^{(0)}_2
        =
        \sum_{\alpha=1}^{p}\xi^{(0)}_{1,\alpha}
        +
        \sum_{\beta=1}^{q}\xi^{(0)}_{2,\beta}.
    \end{equation}
    Define the re-indexed height components
    \begin{equation}
        \xi^{(0)}_{+,\alpha}
        =
        \begin{cases}
            \xi^{(0)}_{1,\alpha},
            & 1\leq \alpha \leq p,
            \\[0.3em]
            \xi^{(0)}_{2,\alpha-p},
            & p+1\leq \alpha \leq p+q.
        \end{cases}
    \end{equation}
    Then,
    \begin{equation}
        \xi^{(0)}_1+\xi^{(0)}_2
        =
        \sum_{\alpha=1}^{p+q}
        \xi^{(0)}_{+,\alpha}.
    \end{equation}
    Thus, $\Psi^{(p)}_1+\Psi^{(q)}_2$ is a multi-scale surface with
    $p+q$ height components. Next, we show that this surface is precisely the surface
    $\Psi^{(p+q)}$ obtained in \cref{theorem:scattering_kernel_composition}.
    By \cref{def:surface_multiscale_height}, the slope field of
    $\Psi^{(p)}_1+\Psi^{(q)}_2$ is
    \begin{equation}
        \Dot{\mathbf{\xi}}^{(0)}_+
        =
        \sum_{\alpha=1}^{p+q}
        \Dot{\mathbf{\xi}}^{(0)}_{+,\alpha}
        =
        \sum_{\alpha=1}^{p}
        \Dot{\mathbf{\xi}}^{(0)}_{1,\alpha}
        +
        \sum_{\beta=1}^{q}
        \Dot{\mathbf{\xi}}^{(0)}_{2,\beta}.
    \end{equation}
    Hence, its induced normal distribution satisfies
    \begin{equation}
    \begin{aligned}
        &p_{n,+}^{(p+q)}
        \Bigl(
        \mathbf{n}^{(p+q)},\mathbf{n}^{(p+q-1)},\dots,\mathbf{n}^{(1)}
        \,\Big|\,
        \mathbf{n}^{(0)},\mathbf{v}^{(0)}_i
        \Bigr)
        \\
        &\quad =
        \prod_{\alpha=1}^{p+q}
        p_{n,+}^{(\alpha)}
        \Bigl(
        \mathbf{n}^{(\alpha)}
        \,\Big|\,
        \mathbf{n}^{(\alpha-1)},\mathbf{v}^{(\alpha-1)}_i
        \Bigr).
    \end{aligned}
    \end{equation}
    By the definition of the re-indexed height components, the conditional
    normal PDFs are
    \begin{equation}
        p_{n,+}^{(\alpha)}
        \Bigl(
        \mathbf{n}^{(\alpha)}
        \,\Big|\,
        \mathbf{n}^{(\alpha-1)},\mathbf{v}^{(\alpha-1)}_i
        \Bigr)
        =
        p_{n_1}^{(\alpha)}
        \Bigl(
        \mathbf{n}^{(\alpha)}
        \,\Big|\,
        \mathbf{n}^{(\alpha-1)},\mathbf{v}^{(\alpha-1)}_i
        \Bigr),
        \qquad
        1\leq \alpha \leq p,
    \end{equation}
    and
    \begin{equation}
        p_{n,+}^{(\alpha)}
        \Bigl(
        \mathbf{n}^{(\alpha)}
        \,\Big|\,
        \mathbf{n}^{(\alpha-1)},\mathbf{v}^{(\alpha-1)}_i
        \Bigr)
        =
        p_{n_2}^{(\alpha-p)}
        \Bigl(
        \mathbf{n}^{(\alpha)}
        \,\Big|\,
        \mathbf{n}^{(\alpha-1)},\mathbf{v}^{(\alpha-1)}_i
        \Bigr),
        \qquad
        p+1\leq \alpha \leq p+q.
    \end{equation}
    Consequently,
    \begin{equation}
    \begin{aligned}
        &p_{n,+}^{(p+q)}
        \Bigl(
        \mathbf{n}^{(p+q)},\mathbf{n}^{(p+q-1)},\dots,\mathbf{n}^{(1)}
        \,\Big|\,
        \mathbf{n}^{(0)},\mathbf{v}^{(0)}_i
        \Bigr)
        \\
        &\quad =
        \left[
        \prod_{\alpha=1}^{p}
        p_{n_1}^{(\alpha)}
        \Bigl(
        \mathbf{n}^{(\alpha)}
        \,\Big|\,
        \mathbf{n}^{(\alpha-1)},\mathbf{v}^{(\alpha-1)}_i
        \Bigr)
        \right]
        \left[
        \prod_{\beta=1}^{q}
        p_{n_2}^{(\beta)}
        \Bigl(
        \mathbf{n}^{(p+\beta)}
        \,\Big|\,
        \mathbf{n}^{(p+\beta-1)},\mathbf{v}^{(p+\beta-1)}_i
        \Bigr)
        \right]
        \\
        &\quad =
        p_{n_1}^{(p)}
        \Bigl(
        \mathbf{n}^{(p:1)}
        \,\Big|\,
        \mathbf{n}^{(0)},\mathbf{v}^{(0)}_i
        \Bigr)
        p_{n_2}^{(q)}
        \Bigl(
        \mathbf{n}^{(p+q:p+1)}
        \,\Big|\,
        \mathbf{n}^{(p)},\mathbf{v}^{(p)}_i
        \Bigr).
    \end{aligned}
    \end{equation}
    This is exactly the joint normal PDF obtained in
    \cref{theorem:scattering_kernel_composition} after flattening the
    reflections of $\Psi^{(p)}_1 \circ \left(\Psi^{(q)}_2 \circ
    \mathcal{K_L}\right)$. Hence, the normal hierarchy of the surface
    constructed in \cref{theorem:scattering_kernel_composition} coincides
    with the normal hierarchy induced by the height sum
    $\Psi^{(p)}_1+\Psi^{(q)}_2$. Since the shadowing and masking functions in
    \cref{def:shadowing_multi_scale} are defined with respect to the same
    ordered hierarchy of normals, the re-indexed hierarchy above also gives
    \begin{equation}
        \mathcal{M}^{(p+q)}
        =
        \mathcal{M}^{(p)}_1 \mathcal{M}^{(q)}_2,
        \qquad
        \mathcal{S}^{(p+q)}
        =
        \mathcal{S}^{(p)}_1 \mathcal{S}^{(q)}_2,
    \end{equation}
    for the corresponding one-point functions, and gives the same two-point
    shadowing function as the composite construction in
    \cref{theorem:scattering_kernel_composition}, in \cref{eq:multi_scale_shadowing_composition_proof}. Therefore, the surface
    $\Psi^{(p+q)}$ appearing in
    \cref{theorem:scattering_kernel_composition} is precisely
    \begin{equation}
        \Psi^{(p+q)}
        =
        \Psi^{(p)}_1+\Psi^{(q)}_2.
    \end{equation}
    Substituting this identity into the result of
    \cref{theorem:scattering_kernel_composition} gives
    \begin{equation}
        \left[\Psi^{(p)}_1 \circ
        \left(\Psi^{(q)}_2 \circ \mathcal{K_L}\right)\right]
        \left(
        \mathbf{r}^{(0)}_i \rightarrow \mathbf{r}^{(0)}_r,
        \mathbf{v}^{(0)}_i \rightarrow \mathbf{v}^{(0)}_r,
        t_i \rightarrow t_r
        \right)
        =
        \left[
        \left(\Psi^{(p)}_1+\Psi^{(q)}_2\right)
        \circ
        \mathcal{K_L}
        \right]
        \left(
        \mathbf{r}^{(0)}_i \rightarrow \mathbf{r}^{(0)}_r,
        \mathbf{v}^{(0)}_i \rightarrow \mathbf{v}^{(0)}_r,
        t_i \rightarrow t_r
        \right).
    \end{equation}
    This proves \cref{eq:scattering_operator_summation}. The statement that
    $\left[\left(\Psi^{(p)}_1+\Psi^{(q)}_2\right)\circ\mathcal{K_L}\right]$
    is a scattering kernel follows directly from
    \cref{theorem:scattering_kernel_composition}, because
    $\Psi^{(p)}_1+\Psi^{(q)}_2=\Psi^{(p+q)}$ and the associated shadowing
    functions are assumed to satisfy \cref{eq:shadowing_multi_scale_condition_1,eq:shadowing_multi_scale_condition_2,eq:shadowing_multi_scale_condition_3}. Finally, by \cref{def:surface_multiscale_height}, $(\mathcal{P}^h,+)$ is an abelian group. Moreover, the flat surface
    $0_{\mathcal{P}}$ has zero height, deterministic normal
    $\mathbf{n}^{(0)}=\mathbf{n_G}$, and introduces no additional roughness
    modification. Hence,
    \begin{equation}
        0_{\mathcal{P}}\circ\mathcal{K}
        =
        \mathcal{K},
        \qquad
        \forall \, \mathcal{K}\in\mathcal{T}.
    \end{equation}
    Together with
    \begin{equation}
        \left(\Psi^{(p)}_1+\Psi^{(q)}_2\right)\circ\mathcal{K_L}
        =
        \Psi^{(p)}_1\circ
        \left(\Psi^{(q)}_2\circ\mathcal{K_L}\right),
    \end{equation}
    it shows that $\circ$ defines an action of the abelian group
    $(\mathcal{P}^h,+)$ on $\mathcal{T}$, concluding the proof.
\end{proof}

\section{Study Case: Two-scale Poly-Gaussian Surface} \label{sec:Study_case}

We now apply the multi-scale theory developed in \cref{sec:Multiscale_formalism} to a specific study case: scattering from a multi-scale poly-Gaussian surface. To model each individual scale, we use the poly-Gaussian scattering model developed in \cite{Anton2025}. In this model, a single-scale surface $\Psi$ in the sense of \cref{def:surface_height} is described by a poly-Gaussian random surface, following \cite{Litvak2012}, whose height PDF is given by

\begin{equation} \label{eq:isotropic_height_pdf}
    p_h(\xi) = \int_{-\infty}^{\infty} p_{\gamma}(\gamma) \mathcal{G}\left(\xi, \mu(\gamma), \sigma(\gamma)\right) \, \dif \gamma,
\end{equation}
with $\mathcal{G}\left(\xi, \mu(\gamma), \sigma(\gamma)\right)$ denoting a Gaussian distribution
\begin{equation}
    \mathcal{G}\left(\xi, \mu(\gamma), \sigma(\gamma)\right) = \frac{1}{\sigma(\gamma) \sqrt{2\pi}} \exp\left[- \frac{\left(\xi - \mu(\gamma)\right)^2}{2\sigma(\gamma)^2} \right].
\end{equation}
Here, $\gamma : \mathbb{R} \times \mathbb{R} \rightarrow \mathbb{R}, \gamma = \gamma(x, y)$ represents the control process, while the functions $\mu : \mathbb{R} \rightarrow \mathbb{R}, \mu = \mu(\gamma)$ and $\sigma : \mathbb{R} \rightarrow \mathbb{R}, \sigma = \sigma(\gamma)$ define the mean and standard deviation on $\Psi$ at point $(x, y)$. The PDF of the $\gamma$ process is just the normal distribution, i.e.
\begin{equation}
    p_{\gamma}(\gamma) = \mathcal{G}\left(\gamma, 0, 1\right) = \frac{1}{\sqrt{2\pi}}\exp\left[- \frac{\gamma^2}{2} \right].
\end{equation}
On the other hand, the PDF of the slope profile of $\Psi$, $\Dot{\xi} = \derivative{\xi}{r}$, was derived in \cite{Anton2025} as
\begin{equation} \label{pp_2:eq:slope_PDF}
    p_s(\Dot{\xi}) = \int_{-\infty}^{\infty}\int_{-\infty}^{\infty} p_{\gamma}(\gamma) p_{\Dot{\gamma}}(\Dot{\gamma}) \mathcal{G}\left(\Dot{\xi}, \mu_s(\gamma, \Dot{\gamma}), \sigma_s(\gamma, \Dot{\gamma}) \right)  \, \dif \gamma \dif \Dot{\gamma},
\end{equation}
where the Gaussian distribution $\mathcal{G}\left(\Dot{\xi}, \mu_s(\gamma, \Dot{\gamma}), \sigma_s(\gamma, \Dot{\gamma}) \right)$ is given as
\begin{equation}
    \mathcal{G}\left(\Dot{\xi}, \mu_s(\gamma, \Dot{\gamma}), \sigma_s(\gamma, \Dot{\gamma}) \right) = \frac{1}{\sigma_s \sqrt{2\pi}} \exp\left[ - \frac{\left(\Dot{\xi} - \mu_s(\gamma, \Dot{\gamma}) \right)^2}{2\sigma_s(\gamma, \Dot{\gamma})^2}\right],
\end{equation}
and the PDF of the radial derivative of the $\gamma$ process, $\Dot{\gamma} = \frac{\dif \gamma}{\dif r}$ is given as in \cite{Anton2025}:
\begin{equation}
    p_{\Dot{\gamma}}(\Dot{\gamma}) = \mathcal{G}\left(\Dot{\gamma}, 0, \frac{\sqrt{2}}{R}\right) = \frac{R}{\sqrt{4\pi}}\exp\left[- \frac{R^2 \Dot{\gamma}^2}{4}\right],
\end{equation}
where $R$ denotes the autocorrelation length of $\gamma$, under the assumption of an autocorrelation function $\mathcal{C}_{\gamma} : [0, \infty) \rightarrow \mathbb{R}$, $\mathcal{C}_{\gamma} = \mathcal{C}_{\gamma}(r)$ of the form
\begin{equation}
    \mathcal{C}_{\gamma}(r) = \exp\left[- \frac{r^2}{R^2} \right].
\end{equation}
Furthermore, the slope mean and variance functions, $\mu_s(\gamma, \Dot{\gamma})$ and $\sigma_s(\gamma, \Dot{\gamma})$, are defined as
\begin{equation}
\begin{aligned}
    \mu_s(\gamma, \Dot{\gamma})
    &= \derivative{\mu(\gamma)}{\gamma} \Dot{\gamma}, \\
    \sigma_s(\gamma, \Dot{\gamma})
    &= \sqrt{2\left(\frac{\sigma(\gamma)}{R}\right)^2
        + \Dot{\gamma}^2\left(\derivative{\sigma(\gamma)}{\gamma}\right)^2 }.
\end{aligned}
\end{equation}
We may further generalise the surface model for a multi-scale framework by introducing a constant slope bias to the height profile, defined around a point $(x_0, y_0) \in \mathbb{R}^2 \times \mathbb{R}^2$ as
\begin{equation}
    \xi_b(x, y) = -\tan(\theta_{n_1}) \cos(\theta_{n_2}) \left( x - x_0\right) - \tan(\theta_{n_1}) \sin(\theta_{n_2}) \left(y - y_0\right),
\end{equation}
where $\theta_{n_1} \in [0, \pi)$ and $\theta_{n_2} \in [0, 2\pi)$ represent the bias angles. Then, the height and slope PDFs become
\begin{equation}
\begin{aligned}
    p_h^b(\xi \, | \, x, y) &= \int_{-\infty}^{\infty} p_{\gamma}(\gamma) \mathcal{G}\left(\xi, \mu(\gamma) - \xi_b(x, y), \sigma(\gamma)\right) \, \dif \gamma, \\
    p_{s_x}^b(\Dot{\xi}) &= \int_{-\infty}^{\infty}\int_{-\infty}^{\infty} p_{\gamma}(\gamma) p_{\Dot{\gamma}}(\Dot{\gamma}) \mathcal{G}\left(\Dot{\xi}, \mu_s(\gamma, \Dot{\gamma}) - \tan(\theta_{n_1}) \cos(\theta_{n_2}), \sigma_s(\gamma, \Dot{\gamma}) \right)  \, \dif \gamma \dif \Dot{\gamma}, \\
    p_{s_y}^b(\Dot{\xi}) &= \int_{-\infty}^{\infty}\int_{-\infty}^{\infty} p_{\gamma}(\gamma) p_{\Dot{\gamma}}(\Dot{\gamma}) \mathcal{G}\left(\Dot{\xi}, \mu_s(\gamma, \Dot{\gamma}) - \tan(\theta_{n_1}) \sin(\theta_{n_2}), \sigma_s(\gamma, \Dot{\gamma}) \right)  \, \dif \gamma \dif \Dot{\gamma}.
\end{aligned}
\end{equation}
Based on this, we may now define the PDF of sampling a local surface normal at scale $p$, $\mathbf{n_l}^p$, given normal $\mathbf{n_l}^{p-1}$ at scale $p-1$, for $p \geq 1$, as
\begin{multline} \label{eq:polygaussian_scattering}
     p_n^{(p)}\left(\mathbf{n_l}^{(p)} \, | \, \mathbf{n_l}^{(p-1)} \mathbf{v_i}^{(p-1)}, \sigma, \mu, R \right) = \frac{F_k^2}{4\pi v_z^2A^2}\infint{\infint{\infint{\frac{p_{\gamma}(\gamma)p_{\Dot{\gamma}_x}(\Dot{\gamma}_x)p_{\Dot{\gamma}_y}(\Dot{\gamma}_y)}{\sqrt{\left( \left(\frac{\sigma}{R} \right)^2  + \frac{1}{2} \sigma_{\gamma}^2 \Dot{\gamma}_x^2\right) \left( \left(\frac{\sigma}{R} \right)^2  + \frac{1}{2} \sigma_{\gamma}^2 \Dot{\gamma}_y^2 \right)}} \\ \times \exp \left[- \frac{(v_x + \left(\mu_{\gamma} \Dot{\gamma}_x - \tan(\theta_{n_1}) \cos(\theta_{n_2})\right) v_z)^2}{4 v_z^2 \left( \left(\frac{\sigma}{R} \right)^2  + \frac{1}{2} \sigma_{\gamma}^2 \Dot{\gamma}_x^2\right)} - \frac{(v_y + \left(\mu_{\gamma} \Dot{\gamma}_y - \tan(\theta_{n_1}) \sin(\theta_{n_2})\right) v_z)^2}{4 v_z^2 \left( \left(\frac{\sigma}{R} \right)^2  + \frac{1}{2} \sigma_{\gamma}^2 \Dot{\gamma}_y^2\right)}\right] \, \dif \gamma \dif \Dot{\gamma}_x \dif \Dot{\gamma}_y}}} \sin(\theta_{r_{1}}), 
\end{multline}
where $A$ is a normalisation constant, and the normal $\mathbf{n_l}^{(p)}$ is given by
\begin{equation}
    \mathbf{n_l}^{(p)} = \frac{1}{N}\begin{bmatrix}
        v_x \\ v_y \\ v_z
    \end{bmatrix} = \frac{1}{N} \begin{bmatrix}
       \sin(\theta_{i_1}) \cos(\theta_{i_2}) -  \sin(\theta_{r_1}) \cos(\theta_{r_2}) \\
        \sin(\theta_{i_1}) \sin(\theta_{i_2}) - \sin(\theta_{r_1}) \sin(\theta_{r_2}) \\
        - \cos(\theta_{i_1}) - \cos(\theta_{r_1})
    \end{bmatrix},
\end{equation}
with $N$ being a normalisation constant, $\theta_{i_1}$ and $\theta_{i_2}$ denoting the polar and azimuthal angles of the incident velocity, respectively, and $\theta_{r_1}$ and $\theta_{r_2}$ denoting the corresponding polar and azimuthal angles of the reflected velocity. The normal $\mathbf{n_l}^{(p-1)}$ and incident vector $\mathbf{v_i}$, on the other hand, take the form
\begin{equation}
    \mathbf{n_l}^{(p-1)} = \begin{bmatrix}
        \sin(\theta_{n_1}) \cos(\theta_{n_2}) \\
        \sin(\theta_{n_1}) \sin(\theta_{n_2}) \\
        \cos(\theta_{n_1})
    \end{bmatrix}, \qquad \mathbf{v_i}^{(p-1)} = \begin{bmatrix}
        \sin(\theta_{i_1}) \cos(\theta_{i_2}) \\
        \sin(\theta_{i_1}) \sin(\theta_{i_2}) \\
        \cos(\theta_{i_1})
    \end{bmatrix}.
\end{equation}
Furthermore, the term $F_k$ is given as
\begin{equation}
    F_k = \frac{1 + \cos(\theta_{i_1}) \cos(\theta_{r_1}) - \sin(\theta_{i_1})\sin(\theta_{r_1}) \cos(\theta_{r_2} - \theta_{i_2})}{2 \cos(\theta_{i_1}) \left(\cos(\theta_{i_1}) + \cos(\theta_{r_1})\right)}.
\end{equation}
We may further define the single-scale two-point shadowing function, $\mathcal{S}\left(\mathbf{r_{i_k}}, \mathbf{v_{i_k}}, t_{i_k} \, | \,  \mathbf{r_{i_{k-1}}}, \mathbf{v_{i_{k-1}}}, t_{i_{k-1}}\right)$ based on \cite{Anton2025}, as 
\begin{equation} \label{eq:polygaussian_shadow}
    \mathcal{S}\left(\mathbf{r_{i_k}}, \mathbf{v_{i_k}}, t_{i_k} \, | \,  \mathbf{r_{i_{k-1}}}, \mathbf{v_{i_{k-1}}}, t_{i_{k-1}}\right) = \mathcal{S}(\theta_{r_1}, \theta_{r_2}, \xi_{k} \, | \,  \xi_{k-1}, \theta_{n_1}, \theta_{n_2}) = \left[\frac{\mathcal{F}\left(\xi_{k-1} \, | \, \theta_{n_1}, \theta_{n_2} \right)}{\mathcal{F}\left(\xi_k \, | \, \theta_{n_1}, \theta_{n_2} \right)}\right]^{\mathcal{E}\left( \theta_{r_1}, \theta_{r_2}, \theta_{n_1}, \theta_{n_2}\right)},
\end{equation}
where the base function $\mathcal{F}\left(\xi \, | \, \theta_{n_1}, \theta_{n_2} \right)$ is given by
\begin{equation}
    \mathcal{F}\left(\xi \, | \, \theta_{n_1}, \theta_{n_2} \right) = \frac{1}{2}\infint{p_{\gamma}(\gamma)\left[1 + \erf\left( \frac{\xi - \left(\mu(\gamma) - \xi_b(x, y)\right)}{\sigma(\gamma)\sqrt{2}}\right) \right]\, \dif \gamma},
\end{equation}
and the exponent function $\mathcal{E}\left( \theta_{r_1}, \theta_{r_2}, \theta_{n_1}, \theta_{n_2}\right)$ takes the form
\begin{equation}
\begin{aligned}
        \mathcal{E}\left( \theta_{r_1}, \theta_{r_2}, \theta_{n_1}, \theta_{n_2}\right) &= {\frac{1}{\eta_b} \infint{\infint{p_{\gamma}(\gamma) p_{\Dot{\gamma}}(\Dot{\gamma}) \Delta(\gamma, \Dot{\gamma}) \, \dif \Dot{\gamma} \dif \gamma}}} \\
        \Delta(\gamma, \Dot{\gamma}) & = \frac{\sigma_s(\gamma, \Dot{\gamma})}{\sqrt{2\pi}}\exp\left[ - \frac{(\eta_b - \mu_{\gamma} \Dot{\gamma})^2}{2 \sigma_s(\gamma, \Dot{\gamma})^2}\right] - \frac{1}{2}\left( \eta_b - \mu_{\gamma} \Dot{\gamma}\right)\erfc\left[\frac{\eta_b - \mu_{\gamma} \Dot{\gamma}}{\sigma_s(\gamma, \Dot{\gamma})\sqrt{2}} \right] \\
        \text{with } \eta_b &= \frac{1}{\tan(\theta_{r_1})} - \tan(\theta_{n_1}) \cos(\theta_{n_2}) \cos(\theta_{r_2}) - \tan(\theta_{n_1}) \sin(\theta_{n_2}) \sin(\theta_{r_2}),
\end{aligned}
\end{equation}
where the new height $\xi_k = \xi_{k-1} + \eta_b r$. Finally, we may write the hazard probability as 
\begin{equation}
    f_{\varphi}\left(\mathbf{r_k}, \mathbf{v_k}, t_k\right) = -\frac{\partial }{\partial t}\left[\ln \mathcal{S}(\theta_{r_1}, \theta_{r_2}, \xi_{k}(t) \, | \,  \xi_{k-1}, \theta_{n_1}, \theta_{n_2}) \right] = \frac{v_k}{\sin(\theta_{r_1})} \frac{p_h\left(\xi_k\right)}{\mathcal{F}\left(\xi_{k} \, | \, \theta_{n_1}, \theta_{n_2} \right)} \mathcal{E}\left( \theta_{r_1}, \theta_{r_2}, \theta_{n_1}, \theta_{n_2}\right).
\end{equation}
It is trivial to show that the two-point shadowing function in \cref{eq:polygaussian_shadow} satisfies the three sufficient conditions in \cref{eq:shadowing_condition_1,eq:shadowing_condition_2,eq:shadowing_single_scale_condition_3}. Provided that the chosen local kernel $\mathcal{K_L}$ is a valid point-wise scattering kernel, i.e. $\mathcal{K_L} \in \mathcal{T}^p$, the single-reflection kernel $\mathcal{K_{MK}}$ constructed from the shadowing function and the normal PDF in \cref{eq:polygaussian_scattering} satisfies reciprocity, normalisation, and non-negativity in the sense of \cref{def:reciprocity,def:normalisation,def:nonnegativity}, by \cref{lemma:reciprocity_multi_refl,lemma:normalisation_multi_refl,lemma:nonnegativity_multi_refl}. Furthermore, by \cref{theorem:scattering_kernel_summation}, any multi-scale kernel constructed using the scattering operator associated with $\mathcal{K_{MK}}$, in the sense of \cref{def:scattering_operator}, also belongs to $\mathcal{T}^s$ as defined in \cref{def:scattering_kernel_set}. Here, we therefore aim to verify the physical validity of the multi-scale construction introduced in \cref{def:scattering_operator}. In particular, we examine the assumptions underlying \cref{eq:single_reflection_kernel,eq:normal_time_reversal}, namely that the normals associated with a given scale remain invariant over the distances travelled by particles during their local interactions with the surface. To this end, we use the Cercignani-Lampis-Lord (CLL) kernel of \cite{Lord1995} as the local kernel $\mathcal{K}_L$, with parameters $\alpha_n = 0.8$ and $\sigma_t = 0.1$. These values are taken as representative of satellite materials, following \cite{Jorge2025,Anton2025}. We then consider a two-scale poly-Gaussian surface composed of the single-scale surfaces $\Psi_1$ and $\Psi_2$. Their corresponding $\mu(\gamma)$ and $\sigma(\gamma)$ functions are shown in \cref{fig:polyGaussian_surface_renders}, with autocorrelation lengths $R_1=\SI{0.35}{\micro\meter}$ and $R_2=\SI{0.03}{\micro\meter}$, respectively. The same figure also shows generated square samples of $\Psi_1$ and $\Psi_2$, together with a sample of the resulting multi-scale surface $\Psi_1+\Psi_2$. Finally, we plot the power spectral densities of $\Psi_1$, $\Psi_2$, and the combined surface $\Psi_1+\Psi_2$.

The verification procedure consisted of generating three-dimensional sample geometries for each poly-Gaussian surface and performing Test-Particle Monte Carlo (TPMC) simulations on them. In each case, $N=20000$ oxygen atoms, with molar mass \SI{15.999}{\gram\per\mol}, were scattered from the surface. Two incidence conditions were considered: $(\theta_{i_1},\theta_{i_2})=(\SI{30}{\degree},\SI{0}{\degree})$ and $(\theta_{i_1},\theta_{i_2})=(\SI{60}{\degree},\SI{0}{\degree})$, both with a speed ratio of $S=8$. The surface temperature was kept fixed at $T_S=\SI{300}{\kelvin}$. The TPMC simulations were performed using the GSI\_Toolbox software \cite{GSIToolBox2025}. We then sampled the kernel $\mathcal{K_{MK}}$ $N=20000$ times using the poly-Gaussian functions $\mu(\gamma)$ and $\sigma(\gamma)$ shown in \cref{fig:polyGaussian_surface_renders}, under the same flow and surface conditions. The resulting reflected angular flux and kinetic energy marginals in the $x$-$z$ and $x$-$y$ planes were then compared with those obtained from the TPMC simulations. These comparisons are shown in \cref{fig:polyGaussian_scattering}.

We observe excellent overall agreement between the kernel predictions and the TPMC simulations across all cases considered. For the single-scale surfaces $\Psi_1$ and $\Psi_2$, this agreement is expected, since it is consistent with the accuracy previously reported for the poly-Gaussian scattering model in \cite{Anton2025} for single-scale roughness. Small discrepancies are observed only in the angular-flux histograms for $\Psi_2$ at an incidence angle of \SI{30}{\degree}. In this case, the kernel predicts a slightly more diffuse reflected distribution than the TPMC simulations. Similar discrepancies were reported in \cite{Anton2025} for Gaussian surfaces under comparable conditions, and were attributed to the assumption of independence between the height and slope PDFs entering the two-point shadowing function $\mathcal{S}$ used in the poly-Gaussian model. This discrepancy is therefore the only error source expected to propagate into the multi-scale scattering histograms. The large separation between the two surface scales, $R_1=\SI{0.35}{\micro\meter}$ and $R_2=\SI{0.03}{\micro\meter}$, supports the normal-invariance assumption used in the multi-scale kernel construction of \cref{eq:single_reflection_kernel}. This is indeed what is observed in the bottom panel of \cref{fig:polyGaussian_scattering}. For both incidence angles, \SI{30}{\degree} and \SI{60}{\degree}, the angular histograms predicted by the kernel are slightly more diffuse than the corresponding TPMC results, and a slightly larger degree of backscattering is predicted. 

\begin{figure}[H]
    \centering
    \includegraphics[width=0.99\linewidth]{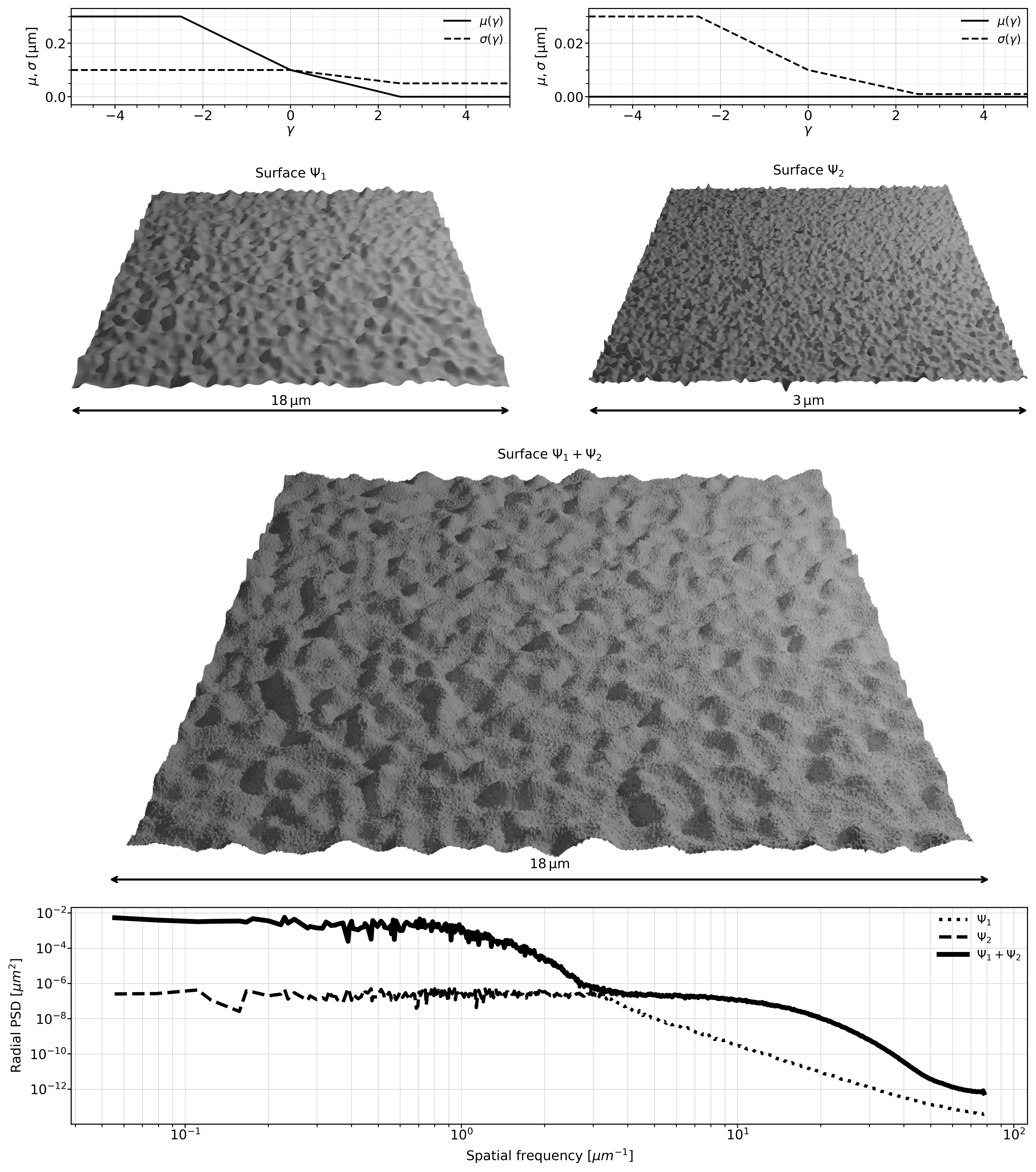}
    \caption{Representation of the multi-scale poly-Gaussian surface considered as a case study. The top panels show the poly-Gaussian functions $\mu(\gamma)$ and $\sigma(\gamma)$ for the component surfaces $\Psi_1$ and $\Psi_2$. The middle panels show realizations of $\Psi_1$, $\Psi_2$, and the combined multi-scale surface $\Psi_1+\Psi_2$. The bottom panel shows the corresponding power spectral densities.}
    \label{fig:polyGaussian_surface_renders}
\end{figure}

\begin{figure}[H]
    \centering
    \includegraphics[width=0.99\linewidth]{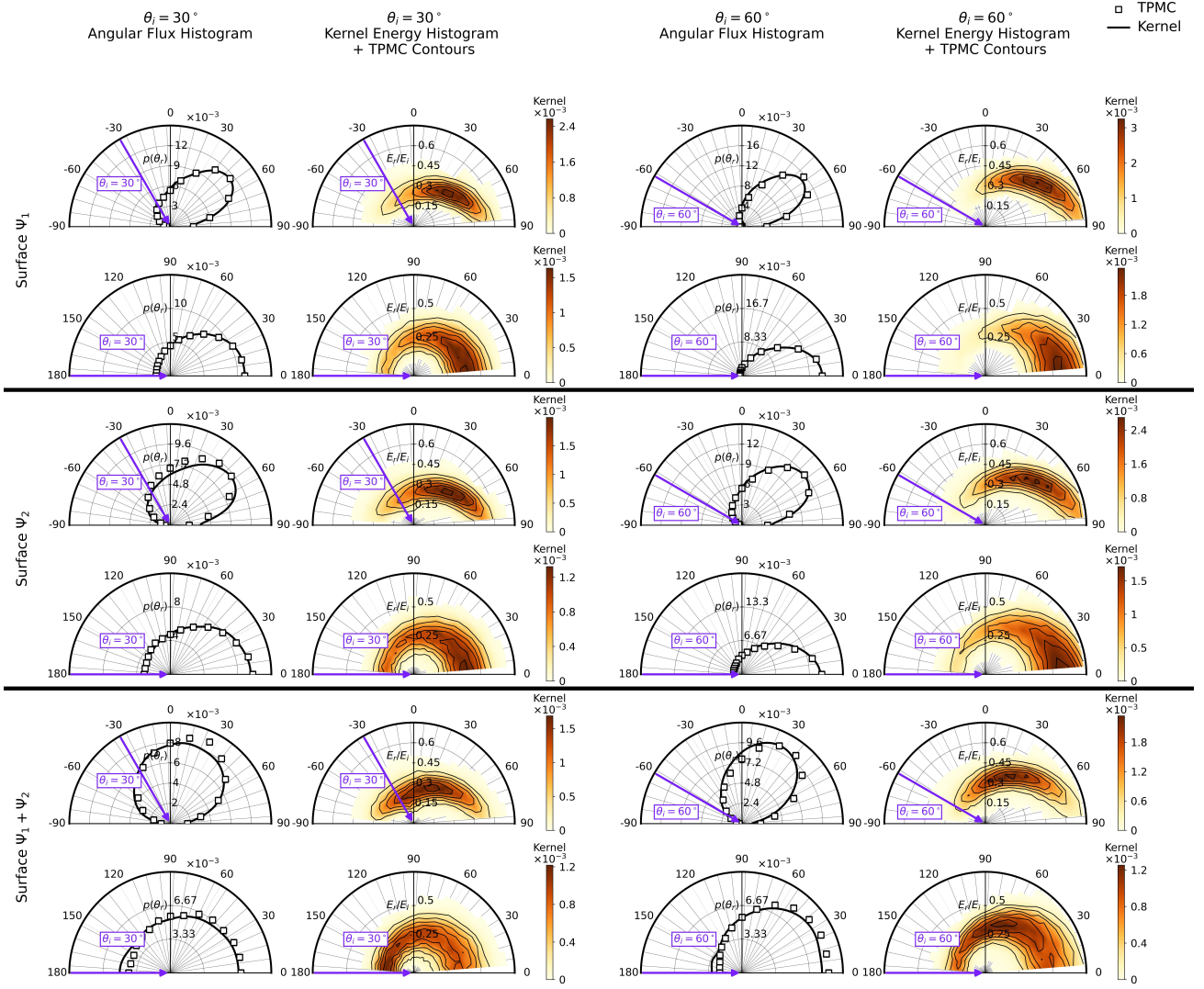}
    \caption{A comparison of the reflected atomic-oxygen particle distributions obtained from the two-scale poly-Gaussian kernel and from TPMC simulations performed on the corresponding generated poly-Gaussian surfaces. The top panel shows the reflected angular flux and angular--energy flux histograms for surface $\Psi_1$, with autocorrelation length $R=\SI{0.35}{\micro\meter}$. The middle panel shows the same quantities for surface $\Psi_2$, with $R=\SI{0.03}{\micro\meter}$. The bottom panel shows the corresponding results for the combined multi-scale surface $\Psi_1+\Psi_2$.}
    \label{fig:polyGaussian_scattering}
\end{figure}

\section{Discussion and Conclusions}\label{sec:Discussion_Conclusions}

The operator $\circ$ defined in \cref{def:scattering_operator} and encoding encoded in \cref{eq:multi_reflection_kernel} provides, to the best of our knowledge, the most general available formulation of gas-particle scattering from an arbitrary homogeneous, moving, isopotential surface for a given local reflection kernel, $\mathcal{K_L}$. It is also the first kernel formulation to provide a formal treatment of the aggregation of multi-scale roughness effects in gas scattering, which is the configuration most frequently encountered on real satellite surfaces \cite{Anton2025, Groh_MISSE2, Jorge2025, Xu2025}. Several previous efforts to model atomic-scale corrugation and macro-scale roughness have yielded similar, though noticeably simplified, formulations. A fixed Gaussian surface model using the Cercignani–Lampis–Lord kernel \cite{Cercignani1971,Lord1995} as a local scattering law was developed in \cite{Liang2018} and proven to satisfy reciprocity. This model was further extended in \cite{Park2025} to account for velocity-dependent corrugations typical of an isopotential surface. Indeed, our formulation reduces to the latter when the local kernel and shadowing function dependency on position and time is removed, i.e. $\mathcal{K_K}(\mathbf{r_i} \rightarrow \mathbf{r_r}, \mathbf{v_i} \rightarrow \mathbf{v_r}, t_i \rightarrow t_r) = \mathcal{K_K}( \mathbf{v_i} \rightarrow \mathbf{v_r})$, $\mathcal{S}(\mathbf{r_i}, \mathbf{v_i}, t_i \, | \, \mathbf{r_r}, \mathbf{v_r}, t_r) = 1$, $f_{\varphi}(\mathbf{r_{i_k}}, \mathbf{v_{i_k}}, t_{i_k}) = \delta(\mathbf{r_{i_k}} - \mathbf{r_{i}})\delta(\mathbf{v_{i_k}} - \mathbf{v_{i}})\delta(t_{i_k} - t_i)$ and $\mathcal{S}(\mathbf{r}, \mathbf{v}, t) = \mathcal{M}(\mathbf{r}, -\mathbf{v}, t) = \mathcal{S}(\mathbf{v})$, and $p_n$ is assumed to follow a Gaussian distribution. If the dependence of $p_n$ on $\mathbf{v}$ is further removed, then one recovers the formulation in \cite{Liang2018}.

Nevertheless, the applicability of the $\circ$ operator relies on several assumptions concerning the two-point shadowing function, $\mathcal{S}(\mathbf{r_i}, \mathbf{v_i}, t_i \, | \, \mathbf{r_r}, \mathbf{v_r}, t_r)$, and the local normal PDF, $p_n(\mathbf{n_L} \, | \, \mathbf{n_G}, \mathbf{v_i})$. In particular, we imposed the three sufficient conditions on the shadowing function given in \cref{eq:shadowing_condition_1}, \cref{eq:shadowing_condition_2}, and \cref{eq:shadowing_single_scale_condition_3}. However, the second is always satisfied by construction, owing to the definition of $f_{\varphi}(\mathbf{r_i}, \mathbf{v_i}, t_i)$ in \cref{def:shadowing_single_scale}. Therefore, only two non-trivial conditions remain for $\mathcal{S}$, namely \cref{eq:shadowing_condition_1} and \cref{eq:shadowing_single_scale_condition_3}. Both conditions are physically natural. The first states that visibility between two points is independent of the direction in which the connecting path is traversed. Equivalently, if a particle can travel unobstructed from one point to another, then the reversed path is also unobstructed. The second condition guarantees that escape is possible along every reachable direction in the open upper hemisphere.

The assumptions associated with the local normal PDF of $\Psi$, $p_n(\mathbf{n_L} \, | \, \mathbf{n_G}, \mathbf{v_i})$, are more restrictive. In particular, the single-reflection kernel in \cref{eq:single_reflection_kernel} assumes that the local surface normal $\mathbf{n_L}$ remains constant during the local gas-surface interaction and is equal to its value at the beginning of the interaction, namely $\mathbf{n_L} = \mathbf{n_L}(t_i)$. This assumption is physically reasonable only when the interaction takes place on a scale much smaller than the roughness scale $R$ of $\Psi$, such that the surface is locally smooth. By contrast, \cref{eq:scattering_operator_summation} allows the actions of the multi-scale surfaces $\Psi_1^{(p)}$ and $\Psi_2^{(q)}$ on $\mathcal{K_L}$ to commute. While this commutation is mathematically permitted in the current framework, it is not physically consistent with the assumptions introduced in \cref{eq:single_reflection_kernel}, \cref{def:surface_multiscale} and \cref{def:shadowing_multi_scale}. A further assumption, stated in \cref{eq:normal_time_reversal} and necessary for pointwise reciprocity, requires the normal PDF $p_n$ to remain unchanged under time reversal of the local particle trajectory. This property follows directly when $\Psi$ is an isopotential surface associated with a conservative field. Such an approximation is appropriate for the Coulomb fields generated by the surface atoms of most materials of interest in satellite aerodynamics \cite{Anton2025, Jorge2025, Xu2023, Xu2025}. 

The final assumption concerns the composition of scattering kernels under surface addition, as described by \cref{eq:scattering_operator_summation}. In particular, this construction requires a conditional normal PDF,
$p_{n}^{(2)}\bigl(\mathbf{n}^{(2)} \mid \mathbf{n}^{(1)},\mathbf{v}^{(1)}_i\bigr)$, to represent the roughness of a smaller-scale component surface $\Psi_2$ in frame $F_1$, which is local to the larger-scale component surface $\Psi_1$ and global to $\Psi_2$. The key point is that the roughness described by $p_{n}^{(2)}$, as described in \cref{def:surface_multiscale_height}, is generally biased with respect to frame $1$, since its mean orientation is determined by $\mathbf{n}^{(1)}$ as expressed in the most global frame $F_0$. Consequently, this roughness cannot, in general, be represented in frame $F_1$ by an unbiased height profile of the form $\xi_2^{(1)}(x^{(1)},y^{(1)},t)$. Such a height representation is well defined only in the common global frame $F_0$, where both component surfaces are expressed through $\xi_1^{(0)}$ and $\xi_2^{(0)}$ and may therefore be summed consistently. For the present framework to remain physically meaningful under surface addition, the statistical model for $\Psi_2$ must therefore allow for a slope-bias term whose direction is determined by $\mathbf{n}^{(1)}$ expressed in frame $F_0$. Under a sufficiently large separation of roughness scales, this bias may be treated as constant over a particle-surface interaction. In \cref{fig:polyGaussian_scattering}, we showed that this approximation remains accurate for a scale ratio $R_1/R_2 \approx 10$, using the poly-Gaussian scattering model developed in \cite{Anton2025}. Future work will investigate how this ratio affects the $L_2$ error norm of the aerodynamic force produced by gas particles scattering from a multi-scale surface.

\backmatter

\section*{Data Availability}

The data that supports the findings of this study are available from the corresponding author upon reasonable request.

\section*{Competing interest}

The authors declare no competing interests

\section*{Acknowledgements}

This work was funded by the Dutch Research Council (NWO) under Grant No. ENW.GO.001.008.





\bibliography{sn-bibliography}


\end{document}